%% file: main.tex
\documentclass[11pt]{article}
\usepackage[utf8]{inputenc}
\usepackage{amsmath, amssymb, amsthm, amsfonts}
\usepackage{url}
\usepackage{appendix}
\usepackage{fullpage}
\usepackage{caption}
\usepackage{subcaption}
\usepackage{graphicx}

\usepackage{xcolor}

\definecolor{ForestGreen}{rgb}{0.0333,0.4451,0.0333}
\definecolor{DarkRed}{rgb}{0.65,0,0}
\definecolor{Red}{rgb}{1,0,0}
\usepackage[linktocpage=true,
pagebackref=true,colorlinks,
linkcolor=DarkRed,citecolor=ForestGreen,
bookmarks,bookmarksopen,bookmarksnumbered]{hyperref}

\usepackage{cleveref}
\usepackage{thm-restate}

\newtheorem{theorem}{Theorem}[section]

\newtheorem{lemma}[theorem]{Lemma}

\newtheorem{definition}[theorem]{Definition}

\newcommand{\eps}{\varepsilon}

\newcommand{\poly}{\text{poly}}
\newcommand{\E}{\mathbb{E}}

\newcommand{\supp}{\text{supp}}
\newcommand{\st}{\text{val}}

\newcommand{\deficit}{\text{deficit}}
\renewcommand{\l}{\ell}
\newcommand{\w}{w}
\newcommand{\U}{U}
\newcommand{\OPT}{\text{OPT}}
\DeclareMathOperator*{\argmax}{arg\,max}

\newcommand{\calP}{\mathcal{P}}

\input{ellisPreamble}

\newif\ifcomments
\commentstrue

\ifcomments

\newcommand{\tnote}[1]{\colnote{red}{#1--Ohad}{TS}}

\else
\newcommand{\tnote}[1]{}
\newcommand{\yaowei}[1]{}
\fi 

\title{Maximum Length-Constrained Flows and Disjoint Paths:\\ Distributed, Deterministic and Fast
}
\author{
\begin{tabular}[t]{c@{\extracolsep{1.3em}}cc} 
        Bernhard Haeupler  & \qquad D Ellis Hershkowitz & Thatchaphol Saranurak \\
        \small Carnegie Mellon University \&  & \small \qquad Carnegie Mellon University \& & \small University of Michigan \\
        \small  ETH Z\"urich &  \small \qquad ETH Z\"urich  & \\
        \small\texttt{haeupler@cs.cmu.edu} & \small \qquad \texttt{dhershko@cs.cmu.edu} & \small \texttt{thsa@umich.edu}
\end{tabular}
}

\date{}
\begin{document}

\maketitle

\begin{abstract}
    Computing routing schemes that support both high throughput and low latency is one of the core challenges of network optimization. Such routes can be formalized as $h$-length flows which are defined as flows whose flow paths have length at most $h$. Many well-studied algorithmic primitives---such as maximal and maximum length-constrained disjoint paths---are special cases of $h$-length flows. Likewise the optimal $h$-length flow is a fundamental quantity in network optimization, characterizing, up to poly-log factors, how quickly a network can accomplish numerous distributed primitives.
    
    In this work, we give the first efficient algorithms for computing $(1 - \epsilon)$-approximate $h$-length flows that are nearly ``as integral as possible.'' We give deterministic algorithms that take $\tilde{O}(\text{poly}(h, \frac{1}{\epsilon}))$ parallel time and $\tilde{O}(\text{poly}(h, \frac{1}{\epsilon}) \cdot 2^{O(\sqrt{\log n})})$ distributed CONGEST time. We also give a CONGEST algorithm that succeeds with high probability and only takes $\tilde{O}(\text{poly}(h, \frac{1}{\epsilon}))$ time. 
    
    Using our $h$-length flow algorithms, we give the first efficient deterministic CONGEST algorithms for the maximal length-constrained disjoint paths problem---settling an open question of Chang and Saranurak (FOCS 2020)---as well as essentially-optimal parallel and distributed approximation algorithms for maximum length-constrained disjoint paths. The former greatly simplifies deterministic CONGEST algorithms for computing expander decompositions. We also use our techniques to give the first efficient and deterministic $(1-\epsilon)$-approximation algorithms for bipartite $b$-matching in CONGEST. Lastly, using our flow algorithms, we give the first algorithms to efficiently compute $h$-length cutmatches, an object at the heart of recent advances in length-constrained expander decompositions.
\end{abstract}

    
\pagenumbering{gobble}
\newpage

\tableofcontents
\newpage
\pagenumbering{arabic}

\setcounter{page}{1}

\section{Introduction}\label{sec:intro}

Throughput and latency are two of the most fundamental quantities in a communication network. Given node sets $S$ and $T$, throughput measures the rate at which bits can be delivered from $S$ to $T$ while the worst-case latency measures the maximum time it takes for a bit sent from $S$ to arrive at $T$. Thus, a natural question in network optimization is:
\begin{quote}\centering\textit{
    How can we achieve high throughput while maintaining a low latency?}
\end{quote}
If we imagine that each edge in a graph incurs some latency and edges in a graph can only support limited bandwidth, then achieving high throughput subject to a latency constraint reduces to finding a large collection of paths that are both short and non-overlapping. One of the simplest and most well-studied ways of formalizing this is the maximal edge-disjoint paths problems (henceforth we use $h$-length to mean length \emph{at most} $h$). \newcommand{\maximalEdgeDisjointUndirected}{
\begin{quote}
    \textbf{Maximal Edge-Disjoint Paths}: Given graph $G = (V, E)$, length constraint $h \geq 1$ and two disjoint sets $S, T \subseteq V$, find a collection of $h$-length edge-disjoint $S$ to $T$ paths $\mcP$ such that any $h$-length $S$ to $T$ path shares an edge with at least one path in $\mcP$.
\end{quote}
}\maximalEdgeDisjointUndirected
The simplicity of the maximal edge-disjoint paths problem has made it a crucial primitive in numerous algorithms. For example, algorithms for maximal edge-disjoint paths are used in approximating maximum matchings \cite{lotker2008improved} and computing expander decompositions \cite{chuzhoy2020deterministic,saranurak2019expander}. While efficient randomized algorithms are known for maximal edge-disjoint paths in the CONGEST model of distributed computation \cite{lotker2008improved,chang2020deterministic}, no deterministic CONGEST algorithms are known. Indeed, the existence of such algorithms was stated as an open question by \citet{chang2020deterministic}.

Of course, a maximal collection of routing paths need not be near-optimal in terms of cardinality and so a natural extension of the above problem is its \emph{maximum} version. 
\newcommand{\maximumEdgeDisjointUndirected}{
\begin{quote}
    \textbf{Maximum Edge-Disjoint Paths}: Given graph $G = (V, E)$, length constraint $h \geq 1$ and disjoint sets $S, T \subseteq V$, find a max cardinality collection of $h$-length edge-disjoint $S$ to $T$ paths.
\end{quote}
}\maximumEdgeDisjointUndirected
First studied by \citet{lovasz1978mengerian}, this problem and its variants have received considerable attention, especially for small constant $h$ \cite{exoo1983line,itai1982complexity,bley2003complexity,kleinberg1996approximation,baveja2000approximation,broder1994existence,golovach2011paths}. It is unfortunately known to suffer from strong hardness results: the above problem has an $\Omega(h)$ integrality gap and is $\Omega(h)$-hard-to-approximate under standard complexity assumptions in the directed case \cite{guruswami2003near,baier2010length}. Indeed, as observed in several works \cite{andoni2020parallel,haeupler2021tree,kleinberg1996approximation} adding length constraints can make otherwise tractable problems computationally infeasible and render otherwise structured objects poorly behaved.

The above problems are common primitives because their solutions are special cases of a more general class of routing schemes that are central to distributed computing, length-constrained flows.
\begin{quote}
    \textbf{Maximum Length-Constrained Flow}: Given digraph $D = (V, A)$, length constraint $h \geq 1$ and two disjoint sets $S, T \subseteq V$, find a collection of $h$-length $S$ to $T$ paths $\mcP$ and a value $f_P \geq 0$ for $P \in \mcP$ where $\sum_{P \ni a} f_P \leq 1$ for every $a \in A$ and $\sum_P f_P$ is maximized.
\end{quote}
In several formal senses, length-constrained flows are \emph{the} problem that describes how to efficiently communicate in a network. \citet{haeupler2020network} showed that, up to poly-log factors, the maximum length-constrained flow gives the minimum makespan of multiple unicasts in a network, even when (network) coding is allowed. Even stronger, the ``best'' length-constrained flow gives, up to poly-log factors, the optimal running time of a CONGEST algorithm for numerous distributed optimization problems, including minimum spanning tree (MST), approximate min-cut and approximate shortest paths \cite{haeupler2021universally}. 

Correspondingly, there has been considerable work on centralized, parallel and distributed algorithms for computing length-constrained flows, again especially for small constant $h$ \cite{mahjoub2010max,altmanova2019polynomial,awerbuch2007distributed,awerbuch2007distributedSteep,awerbuch2008greedy,cohen1995approximate,pienkosz2015integral,baier2010length,fleischer2007quickest,dahl2004directed}. Most notably for this work, \citet{awerbuch2007distributed} gave efficient (about $\poly(h)$) deterministic algorithms in the distributed ROUTERS model and \citet{altmanova2019polynomial} gave sequential algorithms that take about $O(m^2 \cdot \poly(h))$ time. The principal downside of the former's algorithms is that it may produce solutions that are arbitrarily fractional in the sense that they are a convex combination of arbitrarily-many integral solutions. The latter does not do this but does not clearly admit an efficient distributed or parallel implementation. Often, however, there is a need for efficient algorithms that produce (nearly) integral length-constrained flows; in particular computing many classic integral objects (such as matchings) reduces to length-constrained flows with $h = O(1)$ and so, if we hope to use length-constrained flows for computing such objects integrally, we often require that these flows be (nearly) integral.


Thus, in summary a well-studied class of routing problems aims to capture both latency and throughput concerns. These problems are known to serve as important algorithmic primitives as well as complete characterizations of the distributed complexity of many problems. However, the simplest of these problems---maximal edge-disjoint paths---lacks good deterministic CONGEST algorithms while the maximum version of this problem has no known efficient (distributed) approximation algorithms and its fractional generalization, length-constrained flows, lack efficient algorithms with reasonable integrality guarantees.

\subsection{Our Contributions}\label{sec:contributions}

We give the first efficient algorithms for computing these objects in several models of computation. 

\subsubsection{Algorithms for Length-Constrained Flows}

Given a digraph with $n$ nodes and $m$ arcs, our main theorem shows how to \emph{deterministically} compute $h$-length flows that are $(1 - \eps)$-approximate in $\tilde{O}(\poly(h, \frac{1}{\eps}))$ parallel time with $m$ processors and $\tilde{O}(\poly(h, \frac{1}{\eps}) \cdot 2^{O(\sqrt{\log n})})$ distributed CONGEST time. We additionally give a randomized CONGEST algorithm that succeeds with high probability and runs in time $\tilde{O}(\poly(h, \frac{1}{\eps}))$. Our distributed algorithms for length-constrained flows algorithms can be contrasted with the best distributed algorithms for (non-length-constrained) flows which run in $(d + \sqrt{n}) \cdot n^{o(1)}$ time \cite{ghaffari2015near}, nearly matching an $\tilde{\Omega}(d + \sqrt{n})$ lower bound of \citet{sarma2012distributed}.\footnote{We use $\tilde{O}$ notation to suppress dependence on $\poly(\log n)$ factors, ``with high probability'' to mean with probability at least $1 - \frac{1}{\poly(n)}$ and $d$ for the diameter of the input graph.}


Our algorithms work for general arc capacities (i.e.\ connection bandwidths), general lengths (i.e.\ connection latencies) and multi-commodity flow variants. Furthermore, they are are sparse with support size $\poly(h,1/\eps) \cdot |A|$ and also come with a certifying dual solution; a so-called moving cut \cite{haeupler2021universally,chlamtavc2020cut,awerbuch2007distributed}. Lastly, and most critically, the flows we compute are nearly ``as integral as possible'':
\begin{quote}
    \textbf{Optimal Integrality:} for constant $\eps > 0$ they are a convex combinations of $\tilde{O}(h)$ sets of arc-disjoint paths. No near-optimal $h$-length flow can be a convex combination of $o(h)$ such sets since, by an averaging argument, this would violate the aforementioned $\Omega(h)$ integrality gap.
\end{quote}
As an immediate consequence of our parallel algorithms we also get deterministic sequential algorithms running in $\tilde{O}(m \cdot \poly(h, \frac{1}{\eps}))$ which improves upon the aforementioned $O(m^2)$-dependence of \citet{altmanova2019polynomial}. Thus our work can be understood as getting the best of prior works---the (near)-integrality of \citet{altmanova2019polynomial} and the efficiency of \citet{awerbuch2007distributed}---both of which are necessary for our applications. See \Cref{sec:mainResults} for a formal description.


\subsubsection{Applications of our Length-Constrained Flow Algorithms}
Using the optimal integrality of our solutions, we are able to achieve several new results.

    \paragraph{Maximal and Maximum Edge-Disjoint Paths.}
    First, as an almost immediate corollary of our length-constrained flow algorithms and their near-optimal integrality, we derive the first efficient deterministic CONGEST algorithms for maximal edge-disjoint paths. This settles the open question of \citet{chang2020deterministic}.
    
    Similarly, we give efficient parallel and distributed  $\tilde{O}(h)$-approximation algorithms for the maximum edge-disjoint paths problem, nearly matching the known $\Omega(h)$ hardness. See \Cref{sec:disjointPaths} for details as well and additional results on variants of these problems.

    \paragraph{Simpler Distributed Expander Decompositions Deterministically.} As a consequence of our maximal edge-disjoint paths algorithms, we are able to greatly simplify known distributed algorithms for deterministically computing expander decompositions. 
    
    We refer the reader to \citet{chang2020deterministic} for a more thorough overview of the area, but provide a brief synopsis here. An $(\epsilon, \phi)$ expander decomposition removes an $\epsilon$ fraction of edges from a graph so as to ensure that each remaining connected component has conductance at least $\phi$. Expander decompositions have led to many recent exciting breakthroughs, including in linear systems \cite{spielman2004nearly}, unique games \cite{arora2015subexponential,trevisan2005approximation,raghavendra2010graph}, minimum cut \cite{kawarabayashi2018deterministic}, and dynamic algorithms \cite{nanongkai2017dynamic}. 
    
    \citet{chang2020deterministic} gave the first deterministic CONGEST algorithms for constructing expander decompositions. However, most existing paradigms for computing expander decompositions repeatedly find maximal disjoint paths. As a result of the lack of such algorithms, the authors employ significant technical work-arounds, observing:
    \begin{quote}\textit{
        In the deterministic setting, we are not aware of an algorithm that can [efficiently] solve [maximal disjoint paths]... [A solution to this problem would] simplify our deterministic expander decomposition and routing quite a bit. \cite{chang2020deterministic}}
    \end{quote}
    
    Our deterministic CONGEST algorithms for maximal edge-disjoint paths when plugged into \citet{chang2020deterministic} provide a conceptual simplification of deterministic distributed algorithms by bringing them in line with known paradigms. Additionally, we note that the algorithm of \citet{chang2020deterministic} incurs a $2^{O(\sqrt{\log n})}$ overhead regardless of the maximal disjoint paths algorithm used so further improvement requires a fundamentally different approach. See \Cref{sec:expander}.

    \paragraph{Bipartite $b$-Matching.} Using our length-constrained flow algorithms, we give the first efficient and deterministic $(1-\eps)$-approximations for bipartite $b$-matching in CONGEST. $b$-matching is a classical problem in combinatorial optimization which generalizes matching where we are given a graph $G = (V, E)$ and a function $b : V \to \mathbb{Z}_{> 0}$. Our goal is to assign integer values to edges so that each vertex $v$ has at most $b(v)$ assigned value across its incident edges. $b$-matching and its variants have been extensively studied in distributed settings \cite{fischer2020improved,balliu2021lower, brandt2020truly,koufogiannakis2009distributed,fischer2021local,ahmadi2018distributed,faour2021approximating,halldorsson2015distributed}. A standard folklore reduction which replaces  vertex $v$ with $b(v)$ non-adjacent copies and edge $e = \{u, v\}$ with a bipartite clique between the copies of $u$ and $v$ reduces $b$-matching to matching but requires overhead $\max_{\{u, v\} \in E} b(u) \cdot b(v)$ to run in CONGEST. Thus, the non-trivial goal here is a CONGEST algorithm whose running time does not depend on $b$. While $b$-matching has been extensively studied in distributed settings, currently all that is known is either deterministic algorithms which give $(\frac{1}{2} - \eps)$-approximations in $\tilde{O}(\poly(\log \frac{1}{\eps}))$ time \cite{fischer2020improved} or randomized $(1-\eps)$-approximations in $\tilde{O}(\poly(\frac{1}{\eps}))$ time but which only allow for each edge to be chosen at most once \cite{halldorsson2015distributed}.\footnote{The consensus in the literature generally seems to be that allowing for edges to be chosen multiple times  is the better generalization of matching: e.g.\  \citet{gabow2013algebraic} state ``The fact that b-matchings have an unlimited number of copies of each edge makes them decidedly simpler. For instance b-matchings have essentially the same blossom structure (and linear programming dual variables) as ordinary matching.''}
    
    Similarly to classical matching, it is easy to reduce bipartite $b$-matching to $O(1)$-length flow. Thus, applying our length-constrained flow algorithms and our flow rounding techniques allows us to give the first $(1-\eps)$-approximation for $b$-matching in bipartite graphs running in CONGEST time $\tilde{O}(\poly(\frac{1}{\eps}) \cdot 2^{O(\sqrt{\log n})})$. Our algorithms are deterministic and work for the more general problem where each edge has some capacity indicating the number of times it may be chosen. See \Cref{sec:bmatching}.
    
    \paragraph{Length-Constrained Cutmatches.} Our results allow us to give the first efficient constructions of length-constrained cutmatches. Informally, an $h$-length cutmatch with congestion $\gamma$ is a collection of $h$-length $\gamma$-congestion paths between two vertex subsets along with a moving cut that shows that adding any more $h$-length paths to this set would incur congestion greater than $\gamma$. Like our flows, our cutmatches are also sparse. See \Cref{sec:cutMatches}.
    
    A recent work \cite{haeuplerExpander2022} uses our length-constrained cutmatches algorithms to give the first efficient constructions of length-constrained expander decompositions. This work uses these constructions to give CONGEST algorithms for problems, including MST, $(1+\epsilon)$-min-cut and $(1+\epsilon)$-shortest paths, that are guaranteed to run in sub-linear rounds if such algorithms exist on the network. 

\section{Notation and Conventions}\label{sec:notation}

Before moving on to a formal statement of length-constrained flows, moving cuts and our results we introduce some notation and conventions. Suppose we are given a digraph $D = (V,A)$.

\paragraph{Digraph Notation.} We will associate three functions with the arcs of $D$. We clarify these here.
\begin{enumerate}
    \item \textbf{Lengths:} We will let $ \l = \{\l_a\}_a$ be the \emph{lengths} of arcs in $A$. These lengths will be input to our problem and determine the lengths of paths when we are computing length-constrained flows. Throughout this work we imagine each $\l_a$ is in $\mathbb{Z}_{> 0}$. Informally, one may think of $\l$ as giving link latencies. We assume $\l_a$ is $\poly(n)$.
    \item \textbf{Capacities:} We will let $\U = \{\U_a\}_a$ be the capacities of arcs in $A$. These capacities will specify a maximum amount of flow (either length-constrained or not) that is allowed over each arc. Throughout this work we imagine each $\U_a$ is in $\mathbb{Z}_{\geq 0}$ and we let $\U_{\max}$ give $\max_a \U_a$. We assume $\U_{\max}$ is $\poly(n)$. Informally, one may think of $\U$ as link bandwidths.
    \item \textbf{Weights:} We will let $\w = \{\w_a\}_a$ stand for the weights of arcs in $A$. These weights will be given by our moving cut solutions. Throughout this work each $\w_a$ will be in $\mathbb{R}_{> 0}$.
\end{enumerate}

In general we will treat a path $P = ((v_1, v_2), (v_2, v_3), \ldots)$ as series of consecutive arcs in $A$ (all oriented consistently towards one endpoint). For any one of these weighting functions $\phi \in \{\l, \U, \w\}$, we will let $d_\phi(u,v)$ give the minimum value of a path in $D$ that connects $u$ and $v$ where the value of a path $P$ is $\phi(P) := \sum_{a \in P} \phi(a)$. That is, we think of $d_\phi(u,v)$ as the distance from $u$ to $v$ with respect to $\phi$. We will refer to paths which minimize $\w$ as lightest paths (so as to distinguish them from e.g.\ shortest paths with respect to $\l$).

We let $\delta^+(v) := \{a : a = (v, u)\}$ and $N^+(v) := \{u : (v,u) \in A\}$ give the out arcs and out neighborhoods of vertex $v$. Likewise $\delta^+(W) := \bigcup_{v\in W} \delta^+(W)$. $\delta^-(v) := \{a : a = (u, v)\}$ and $N^-(v) := \{u : (u,v) \in A\}$ are defined symmetrically. We let $\mcP(u,v)$ be all simple paths between $u$ and $v$ and for $W,W' \subseteq V$, we let $\mcP(W, W') := \bigcup_{w \in W, w' \in W'} \mcP(w, w')$ give all paths between vertex subsets $W$ and $W'$.

Given sources $S \subseteq V$ and sinks $T \subseteq V$, we say that $D$ is an $S$-$T$ DAG if $\delta^-(v) = \emptyset$ iff $v \in S$ and $\delta^+(v) = \emptyset$ iff $v \in T$. We say that such an $S$-$T$ DAG is an $h$-layer DAG if the vertex set $V$ can be partitioned into $h+1$ layers $S = V_1 \sqcup V_2 \sqcup  \ldots \sqcup V_{h+1} = T$ where any arc $a = (u,v)$ is such that $u \in V_i$ and $v \in V_{j}$ for some $i$ and $j > i$. We say that $D$ has diameter at most $d$ if in the graph where we forget arc directions every pair of vertices is connected by a path of at most $d$ edges. Notice that the diameter of an $h$-layer $S$-$T$ DAG might be much larger than $h$.

For a (di)graph $D = (V,A)$ and a collection of subgraphs $\mcH$ of $D$, we let $D[\mcH]$ be the graph induced by the union of elements of $\mcH$. $A[\mcH]$ is defined as all arcs contained in some element of $\mcH$.

\paragraph{(Non-Length Constrained) Flow Notation and Conventions.}
We will make extensive use of non-length constrained flows and so clarify our notation for such flows here.

Given a DAG $D = (V,A)$ with capacities $\U$ we will let a flow $f$ be any assignment of non-negative values to arcs in $a$ where $f_a$ gives the value that $f$ assigns to $a$ and $f_a \leq \U_a$ for every $a$. If it is ever the case that $f_a > \U_a$ for some $a$, we will explicitly state that this ``flow'' does not respect capacities. We say that $f$ is an integral flow if it assigns an integer value to each arc. We let $f(A') := \sum_{a \in A'} f_a$ for any $A' \subseteq A$. We define the deficit of a vertex $v$ as $\deficit(f,v) := |\sum_{a \in \delta^+(f,v)} f_a - \sum_{a \in \delta^-(v)} f_a|$. We will let $\supp(f) := \{a : f_a > 0\}$ give the support of flow $f$.

Given desired sources $S$ and sinks $T$, we let $\deficit(f) := \sum_{v \not \in S \cup T} \deficit(f,v)$ be the total amount of flow produced but not at $S$ plus the amount of flow consumed but not at $T$; likewise, we say that a flow $f$ is an $S$-$T$ flow if $\deficit(f) = 0$. We let $\st(f) = \bigcup_{s \in S} f(\delta^+(s))$ be the amount of flow delivered by an $S$-$T$ flow $f$ and we say that $f$ is $\alpha$-approximate if $\st(f) \geq \alpha \cdot \st(f^*)$ where $f^*$ is the $S$-$T$ flow that maximizes $\st$. We say that $f$ is $\alpha$-blocking for $\alpha \in [0, 1]$ if for every path from $S$ to $T$ there is some $a \in P$ where $f_a \geq \alpha \cdot \U_a$. We say that a $1$-blocking flow is blocking.  We say that flow $f'$ is a subflow of $f$ if $f'_a \leq f_a$ for every $a$. 

Given a maximum capacity of $\U_{\max}$, we may assume that every flow $f$ is of the form $f = \sum_i f^{(i)}$ where $(f^{(i)})_a \in \{0, 2^{\log(\U_{\max}) - i}\}$ for every $a$ and $i$; that is, a given flow can always be decomposed into its values on each bit. We call $f^{(i)}$ the $i$th bit flow of $f$ and call the decomposition of $f$ into these flows be the bitwise decomposition of $f$.

\paragraph{Length-Constrained Notation.} 
Given a length function $\l$, vertices $u, v \in V$ and length constraint $h \geq 1$, we let $\mcP_h(u,v) := \{P \in \mcP(u,v) : \l(P) \leq h\}$ be all paths between $u$ and $v$ which have length at most $h$. For vertex sets $W$ and $W'$, we let $\mcP_h(W,W'):= \{P \in \mcP(W, W') : \l(P) \leq h\}$. 
If $G$ also has weights $\w$ then we let $d_\w^{(h)}(u,v) := \min_{P \in \mcP_h(u,v)}\w(P)$ give the minimum weight of a length at most $h$ path connecting $u$ and $v$. For vertex sets $W, W' \subseteq V$ we define $d_\w^{(h)}(W,W') := \min_{P \in \mcP_h(W, W')}\w(P)$ analogously. 
As mentioned an $h$-length path is a path of length at most $h$.

\paragraph{Parallel and Distributed Models.} 
Throughout this work the parallel model of computation we will use is the EREW PRAM model \cite{karp1989survey}. Here we are given some processors and shared random access memory; every memory cell can be read or written to by one processor at a time.

The distributed model we will make use of is the CONGEST model, defined as follows \cite{peleg2000distributed}. The network is modeled as a graph $G=(V,E)$ with $n=|V|$ nodes and $m=|E|$ edges. Communication is conducted over discrete, synchronous rounds. During each round each node can send an $O(\log n)$-bit message along each of its incident edges. Every node has an arbitrary and unique ID of $O(\log n)$ bits, first only known to itself. The running time of a CONGEST algorithm is the number of rounds it uses. We will slightly abuse terminology and talk about running a CONGEST algorithm in digraph $D$; when we do so we mean that the algorithm runs in the (undirected) graph $G$ which is identical to $D$ but where we forget the directions of arcs. In this work, we will assume that if an arc $a$ has capacity $\U_a$ then we allow nodes to send $O(\U_a \cdot \log n)$ bits over the corresponding edge, though none of our applications rely on this assumption.\footnote{We only make use of this assumption once and only make use of it in our deterministic algorithms (in \Cref{lem:decongest}). Furthermore, we do not require this assumption if the underlying digraph is a DAG.}

\section{Length-Constrained Flows, Moving Cuts and Main Result}\label{sec:mainResults}
We proceed to more formally define a length-constrained flow, moving cuts and our main result which computes them. While we have defined length-constrained flows in \Cref{sec:intro} for unit capacities, it will be convenient for us to formally define length-constrained flows for general lengths and capacities in terms of a relevant linear program (LP). We do so now.

Suppose we are given a digraph $D = (V,A)$ with arc capacities $\U$, lengths $\l$ and specified source and sink vertices $S$ and $T$. A maximum $S$ to $T$ flow in $D$ in the classic sense can be defined as a collection of paths between $S$ and $T$ where each path receives some value and the total value incident to an edge does not exceed its capacity. This definition naturally extends to the length-constrained setting where we imagine we are given some length constraint $h \geq 1$ and define a length-constrained flow as a collection of $S$ to $T$ paths each of length at most $h$ where each such path $P$ receives some some value $f_P$. Additionally, these values must respect the capacities of arcs. More precisely, we have the following LP with a variable $f_P$ for each path $P \in \mcP_h(S,T)$.
\begin{align*}\label{LP:hopConFlow}
    \max \sum_{P \in \calP_h(S,T)} f_P \quad \text{s.t.} \tag{Length-Constrained Flow LP} \\
    \sum_{P : a \in P} f_P \leq \U_a & \qquad \forall a \in A\\
    0 \leq f_P & \qquad \forall P \in \calP_h(S,T)
\end{align*}
For a length-constrained flow $f$, we will use the shorthand $f(a) := \sum_{P \ni a} f_P$ and $\supp(f) := \{P : f_P > 0\}$ to give the support of $f$. We will let $\st(f) := \sum_{P \in \mcP_h(S,T)} f_P$ give the value of $f$. An $h$-length flow, then, is simply a feasible solution to this LP.

\begin{definition}[$h$-Length Flow]
 Given digraph $D = (V,A)$ with lengths $\l$, capacities $\U$ and vertices $S, T \subseteq V$, an $h$-length $S$-$T$ flow is any feasible solution to \ref{LP:hopConFlow}.
\end{definition}

With the above definition of length-constrained flows we can now define moving cuts as the dual of length-constrained flows with the following moving cut LP with a variable $\w_a$ for each $a \in A$.
\begin{align*}\label{LP:movingCut}
    \min \sum_{a \in A} \U_a \cdot \w_a \quad \text{s.t.} \tag{Moving Cut LP}  \\
    \sum_{a \in P} \w_a \geq 1 & \qquad \forall P \in \calP_h(S,T)\\
    0 \leq \w_a & \qquad \forall a \in A
\end{align*}
An $h$-length moving cut is simply a feasible solution to this LP.

\begin{definition}[$h$-Length Moving Cut]
 Given digraph $D = (V,A)$ with lengths $\l$, capacities $\U$ and vertices $S, T \subseteq V$, an $h$-length moving cut is any feasible solution to \ref{LP:movingCut}.
\end{definition}

We will use $f$ and $\w$ to stand for solutions to \ref{LP:hopConFlow} and \ref{LP:movingCut} respectively. We say that $(f, \w)$ is a feasible pair if both $f$ and $\w$ are feasible for their respective LPs and that $(f, \w)$ is $(1 \pm \epsilon)$-approximate for $\epsilon \geq 0$ if the moving cut certifies the value of the length-constrained flow up to a $(1 - \epsilon)$; i.e.\ if $(1-\epsilon)\sum_{a} U_a \cdot \w_a \leq \sum_P f_P$. 

We clarify what it means to compute $(f, \w)$ in CONGEST. When we are working in CONGEST we will say that $f$ is computed if each vertex $v$ stores the value $f_a(h') := \sum_{P \in \mcP_{h,h'}(s,a,t)} f_P$ for every $a$ incident to $v$ and $h' \leq h$. Here, we let $ \mcP_{h,h'}(s,a,t)$ be all paths in $\mcP_h(S,T)$ of the form $P' =(a_1, a_2, \ldots a, b_1, b_2, \ldots)$ where the path $(a, b_1, b_2, \ldots)$ has length exactly $h'$ according to $\l$. We say moving cut $\w$ is computed if each vertex $v$ knows the value of $w_a$ for its incident arcs. Likewise, we imagine that each node initially knows the capacities and lengths of its incident arcs.

With the above notions, we can now state our main result. In the following we say $f$ is integral if $f_P$ is an integer for every path in $\mcP_h(S,T)$. The notable aspect of our results is the polynomial dependence on $h$ and $\frac{1}{\epsilon}$; the polynomials could be optimized to be much smaller.

\begin{restatable}{thm}{mainThm}
\label{thm:main}
Given a digraph $D=(V,A)$ with capacities $\U$, lengths $\l$, length constraint $h \geq 1$, $ \eps > 0$ and source and sink vertices $S, T \subseteq V$, one can compute a feasible $h$-length flow, moving cut pair $(f, \w)$ that is $(1 \pm \epsilon)$-approximate in:
\begin{enumerate}
    \item Deterministic parallel time $\tilde{O}(\frac{1}{\eps^9} \cdot h^{17})$ with $m$ processors where $|\supp(f)| \leq \tilde{O}(\frac{h^{10}}{\eps^7} \cdot |A|)$;
    \item Randomized CONGEST time $\tilde{O}(\frac{1}{\eps^{9}} \cdot h^{17})$ with high probability;
    \item Deterministic CONGEST time $\tilde{O}\left(\frac{1}{\eps^9} \cdot h^{17}  + \frac{1}{\eps^7} \cdot h^{16} \cdot (\rho_{CC})^{10} \right)$.
\end{enumerate}
Also, $f = \eta \cdot \sum_{j=1}^k f_j$ where $\eta = \tilde{\Theta}(\epsilon^2)$, $k = \tilde{O}\left(\frac{h}{\epsilon^4} \right)$ and each $f_j$ is an integral $h$-length $S$-$T$ flow.
\end{restatable}
All of our algorithms compute and separately store each $f_j$. The above result immediately gives the deterministic parallel and randomized CONGEST algorithms running in time $\tilde{O}(\poly(h, \frac{1}{\epsilon}))$ mentioned in \Cref{sec:contributions}. For our deterministic CONGEST algorithms, $\rho_{CC}$ in the above gives the quality of the optimal deterministic CONGEST cycle cover algorithm. We formally define this parameter in \Cref{sec:preliminaries} but for now we simply note that $\rho_{CC} \leq 2^{O(\sqrt{\log n})}$ by known results \cite{parter2019optimal,hitron2021general}. Applying this bound on $\rho_{CC}$ gives deterministic CONGEST algorithms running in time $\tilde{O}(\poly(h, \frac{1}{\epsilon})\cdot 2^{O(\sqrt{\log n})})$. If $\rho_{CC}$ is shown to be $\poly(\log n)$, we immediately would get an $\tilde{O}(\poly(h, \frac{1}{\epsilon}))$ time deterministic algorithm for solving $(1-\epsilon)$-approximate $h$-length flow in CONGEST. As mentioned in \Cref{sec:contributions}, $k$ in the above result is optimal up to $\tilde{O}(1)$ factors \cite{guruswami2003near,baier2010length}.

\section{Technical Highlights, Intuition and Overview of Approach}\label{sec:overview}
Before moving on, we give an overview of our strategy for length-constrained flows. In the interest of highlighting what is new in this work we begin by summarizing three key technical contributions. To our knowledge these ideas are new in our work. We will then proceed to provide more detail on how these ideas fit together. For simplicity, we assume the capacity $\U_a = 1$ for all $a$ in this section.
\begin{enumerate}
    \item \textbf{Batched Multiplicative Weights:} First, the core idea of our algorithm will be a ``batched'' version of the ``multiplicative weights'' framework. In particular, we will use what we call ``near-lightest path blockers'' to perform many independent multiplicative weight updates in parallel. Both this batched approach to multiplicative weights and our analysis showing that it  converges to a near-optimal solution quickly are new to our work.
    \item \textbf{Length-Weight Expanded DAG:} Second, we provide a new approximate representation of all near-lightest $h$-length paths by a ``length-weight expanded DAG.'' This representation can be efficiently simulated in a distributed setting and serves as a provably good proxy for flows on all near-lightest $h$-length paths by a DAG. It is a priori not clear such a DAG exists since lightest $h$-length paths do not even induce a DAG. Even harder, this representation has to summarize three arc values at once: lengths, weights and capacities. 
    \item \textbf{Deterministic Integral Blocking Flows:} Third, we give the first efficient distributed deterministic algorithms for computing so-called integral blocking flows. In particular, we show how to use flow rounding techniques to derandomize an approach of \citet{lotker2008improved}; previous works noted that this approach seems inherently randomized \cite{chang2020deterministic}. Our flow rounding techniques are, in turn, built around a novel application of the recently introduced idea of ``cycle covers.'' In particular, we will make use of a slight variant of cycle covers and show how to use them to efficiently round flows in a distributed setting.
\end{enumerate}

\subsection{Using Lightest Path Blockers for Multiplicative Weights}\label{sec:overviewMW}

Computing a length-constrained flow, moving cut pair is naturally suggestive of the following multiplicative-weights-type approach. We initialize our moving cut value $\w_a$ to some very small value for every $a$. Then, we find a lightest $h$-length path from $S$ to $T$ according to $\w$, send some small ($\approx \epsilon$) amount of flow along this path and multiplicatively increase the value of $\w$ on all arcs in this path by $\approx(1 + \epsilon)$. We repeat this until $S$ and $T$ are at least $1$ apart according to $d^{(h)}_{\w}$ (where $d^{(h)}_\w(u,v)$ gives the lightest according to $\w$ path from $u$ to $v$ with length at most $h$). This general idea is an adaptation of ideas of \citet{garg2007faster}.

The principle shortcoming of using such an algorithm for our setting is that it is easy to construct examples where there are polynomially-many arc-disjoint $h$-length paths between $S$ and $T$ and so we would clearly have to repeat the above process at least polynomially-many times until $S$ and $T$ are at least $1$ apart according to $d^{(h)}_\w$. This is not consistent with our goal of $\poly(h)$ complexities since $h$ may be much smaller than $n$. To solve this issue, we use an algorithm similar to the above but instead of sending flow along one path, we send it along a large batch of arc-disjoint paths. 

What can we hope to say about how long such an algorithm takes to make $S$ and $T$ at least $1$ apart according to $d^{(h)}_\w$? If it were the case that every lightest (according to $\w$) $h$-length path from $S$ to $T$ shared an arc with some path in our batch of paths then after each batch we would know that we increased $d^{(h)}_\w(S,T)$ by some non-zero amount. However, there is no way to lower bound this amount; in principle we might only increase $d^{(h)}_\w(S,T)$ by some tiny $\epsilon' > 0$. To solve this issue we find a batch of arc-disjoint paths which have weight essentially $d^{(h)}_\w(S,T)$ but which share an arc with every $h$-length path with weight at most $(1 + \epsilon) \cdot d^{(h)}_\w(S,T)$. Thus, when we increment weights in a batch we know that all \emph{near}-lightest $h$-length paths have their weights increased so we can lower bound the rate at which $d^{(h)}_\w(S,T)$ increases, meaning our algorithm completes quickly.

Thus, in summary we repeatedly find a batch of arc-disjoint $h$-length paths between $S$ and $T$ which have weight about $d^{(h)}_\w(S,T)$; these paths satisfy the property that every $h$-length path from $S$ to $T$ with weight at most $(1 + \epsilon) \cdot d^{(h)}_\w(S,T)$ shares an edge with at least one of these paths; we call such a collection an $h$-length $(1+\epsilon)$-lightest path blocker. We then send a small amount of flow along these paths and multiplicatively increase the weight of all incident edges, appreciably increasing $d^{(h)}_\w(S,T)$. We repeat this until our weights form a feasible moving cut. See \Cref{fig:MW}.

\begin{figure}[ht]
    \centering
    \begin{subfigure}[b]{0.49\textwidth}
        \centering
        \includegraphics[width=\textwidth,trim=0mm 0mm 0mm 0mm, clip]{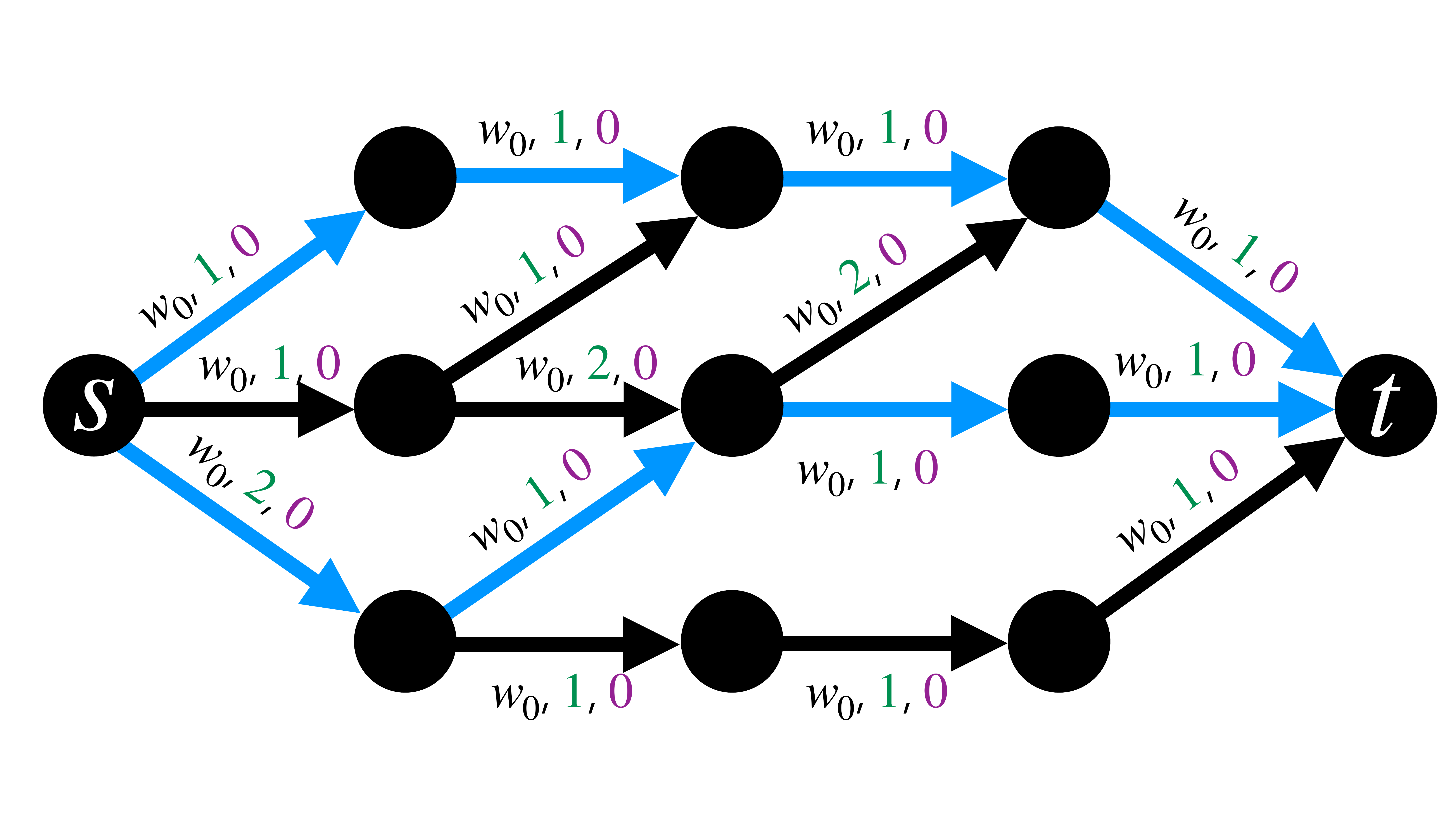}
        \caption{Compute lightest path blocker.}\label{sfig:MW2}
    \end{subfigure}    \hfill
    \begin{subfigure}[b]{0.49\textwidth}
        \centering
        \includegraphics[width=\textwidth,trim=0mm 0mm 0mm 0mm, clip]{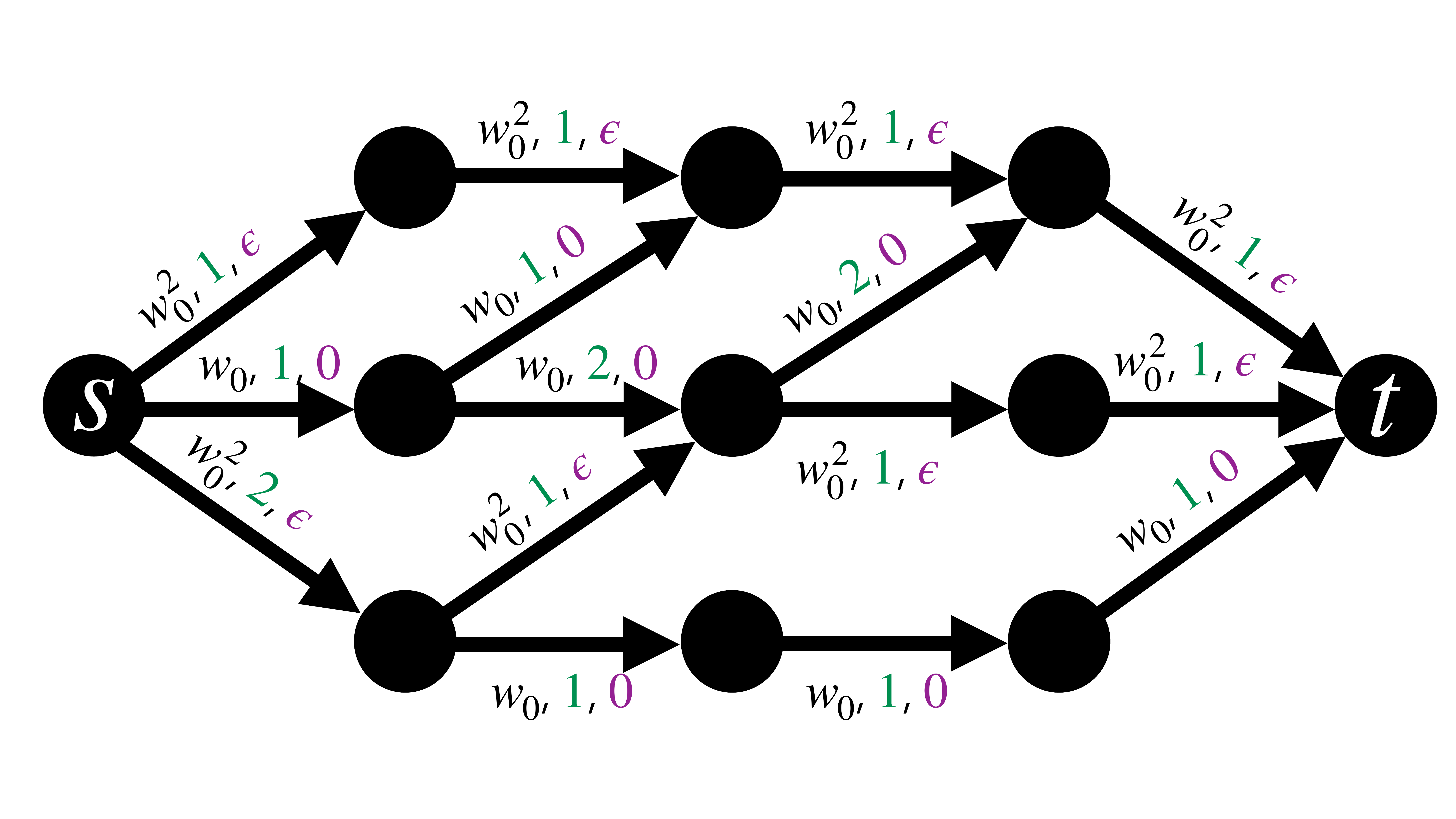}
        \caption{Update flow and weights.}\label{sfig:MW3}
    \end{subfigure}    \hfill
    \begin{subfigure}[b]{0.49\textwidth}
        \centering
        \includegraphics[width=\textwidth,trim=0mm 0mm 0mm 0mm, clip]{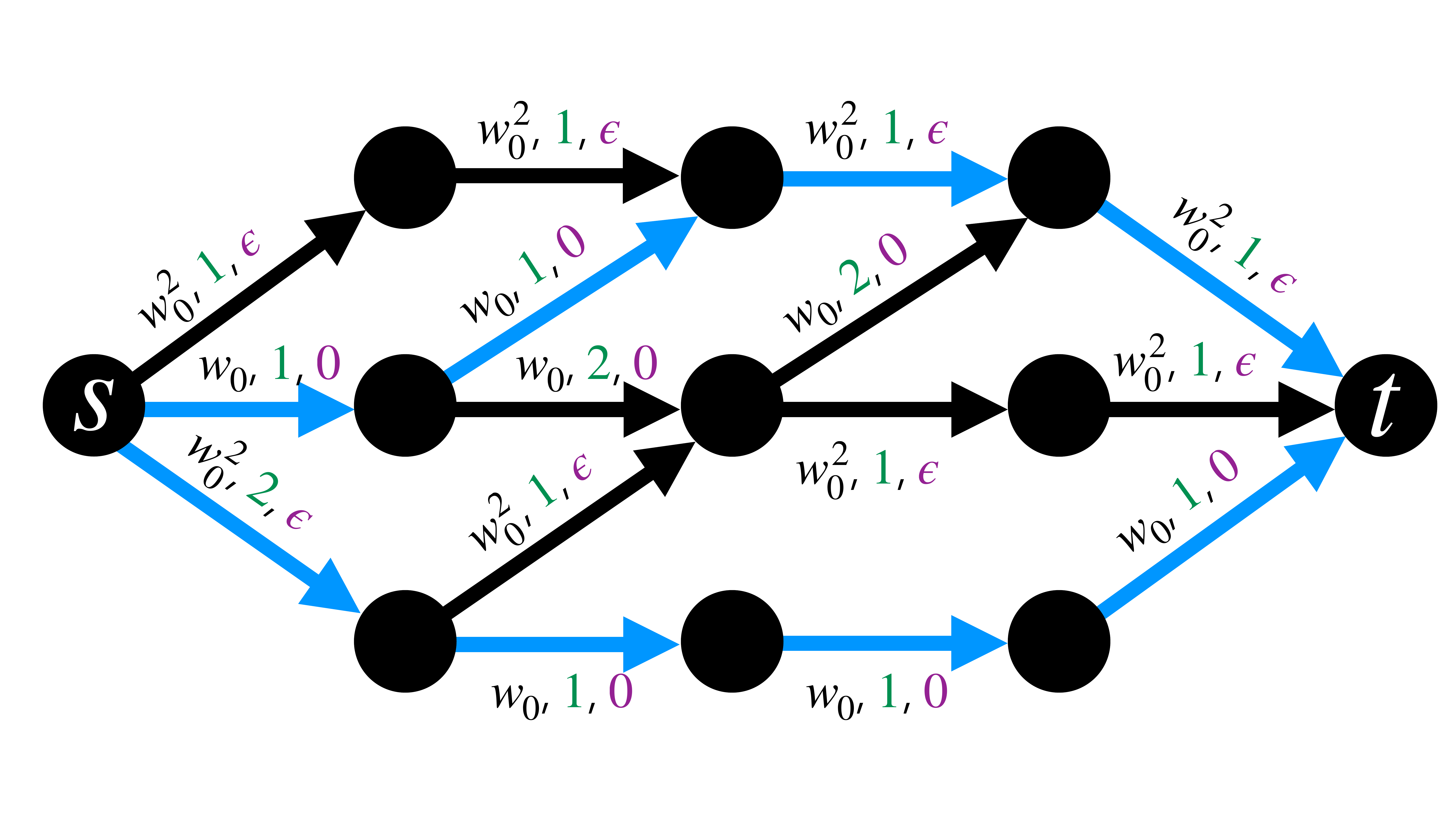}
        \caption{Compute lightest path blocker.}\label{sfig:MW4}
    \end{subfigure}    \hfill
    \begin{subfigure}[b]{0.49\textwidth}
        \centering
        \includegraphics[width=\textwidth,trim=0mm 0mm 0mm 0mm, clip]{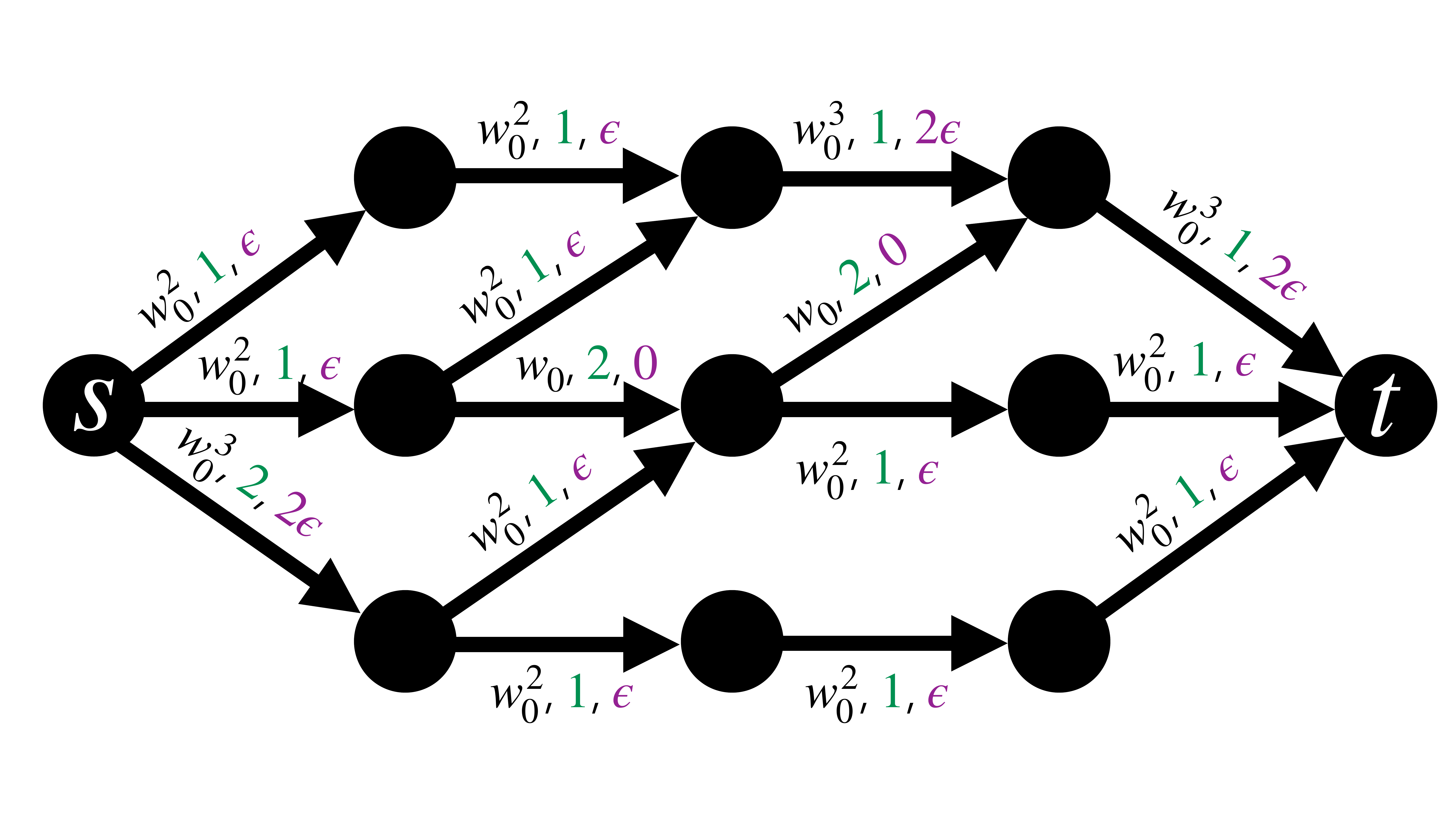}
        \caption{Update flow and weights.}\label{sfig:MW5}
    \end{subfigure}
    \caption{An illustration of the first two iterations of our multiplicative-weights-type algorithm where $h=5$, $S = \{s\}$ and $T = \{t\}$ and capacities are all $1$. Each arc is labelled with the value we multiply its initial weight by (initialized to $w_0 := 1+\epsilon$) then length then flow. Our $h$-length shortest path blockers are in blue.}\label{fig:MW}
\end{figure}


\subsection{Length-Weight Expanded DAG to Approximate $h$-Length Lightest Paths}

The above strategy relies on the computation of $h$-length lightest path blockers. Without the presence of a weight constraint computing such an object easily reduces to computing an integral blocking $S$-$T$ flow on an $h$-layer $S$-$T$ DAG. Specifically, consider the problem of computing a collection of paths from $S$ to $T$ so that every lightest $S$ to $T$ path shares an arc with one path in this collection. It is easy to see that all lightest paths between $S$ and $T$ induce an $h'$-layer $S$-$T$ DAG where $h'$ is the minimum weight of a path between $S$ and $T$. One can then consider this DAG and compute an integral blocking $S$-$T$ flow in it---i.e.\ a maximal arc-disjoint collection of $h'$-length $S$-$T$ paths. By maximality of the flow, the paths corresponding to this flow will guarantee that every $h'$-length $S$ to $T$ path shares an arc with one path in this collection.

However, the presence of a length constraint \emph{and} a weight constraint make such an object much tricker. Indeed, lightest paths subject to length constraints are known to be notoriously poorly behaved; not only do lightest paths subject to a length constraint not induce a metric but they are also arbitrarily far from any metric \cite{andoni2020parallel,haeupler2021tree}. As such, all $S$ to $T$ lightest paths subject to a length constraint do not induce a DAG, much less an $h$-layer $S$ to $T$ DAG; e.g.\ see \Cref{fig:noDAG}.

Our solution is to observe that, if we are allowed to duplicate vertices, then we can construct an $S$-$T$ DAG with about $h^2$ layers that approximately captures the structure of all $h$-length $(1+\eps)$-lightest paths. Specifically, we discretize weights and then make a small number of copies of each vertex to compute a DAG $D^{(h,\lambda)}$---which we call the length-weight expanded DAG. $D^{(h,\lambda)}$ will satisfy the property that if we compute an integral blocking flow and then project this back into $D$ as a set of paths $\mcP$, then $\mcP$ is \emph{almost} a $(1+\eps)$-lightest path blocker. In particular, $\mcP$ will guarantee that some arc of any $h$-length path with weight at most $(1 + \eps) \cdot d^{(h)}_{\w}(S,T)$ is used by some path in $\mcP$; however, the paths of $\mcP$ may not be arc-disjoint as required of a lightest path blockers. Nonetheless, by carefully setting capacities in $D^{(h,\lambda)}$, we will be able to argue that $\mcP$ is nearly arc-disjoint and these violations of arc-disjointness can be repaired with small loss by a  ``decongesting'' procedure. It remains to understand how to compute integral blocking flows in layered $S$-$T$ DAGs.

\begin{figure}
    \centering
    \begin{subfigure}[b]{0.3\textwidth}
        \centering
        \includegraphics[width=\textwidth,trim=0mm 0mm 0mm 0mm, clip]{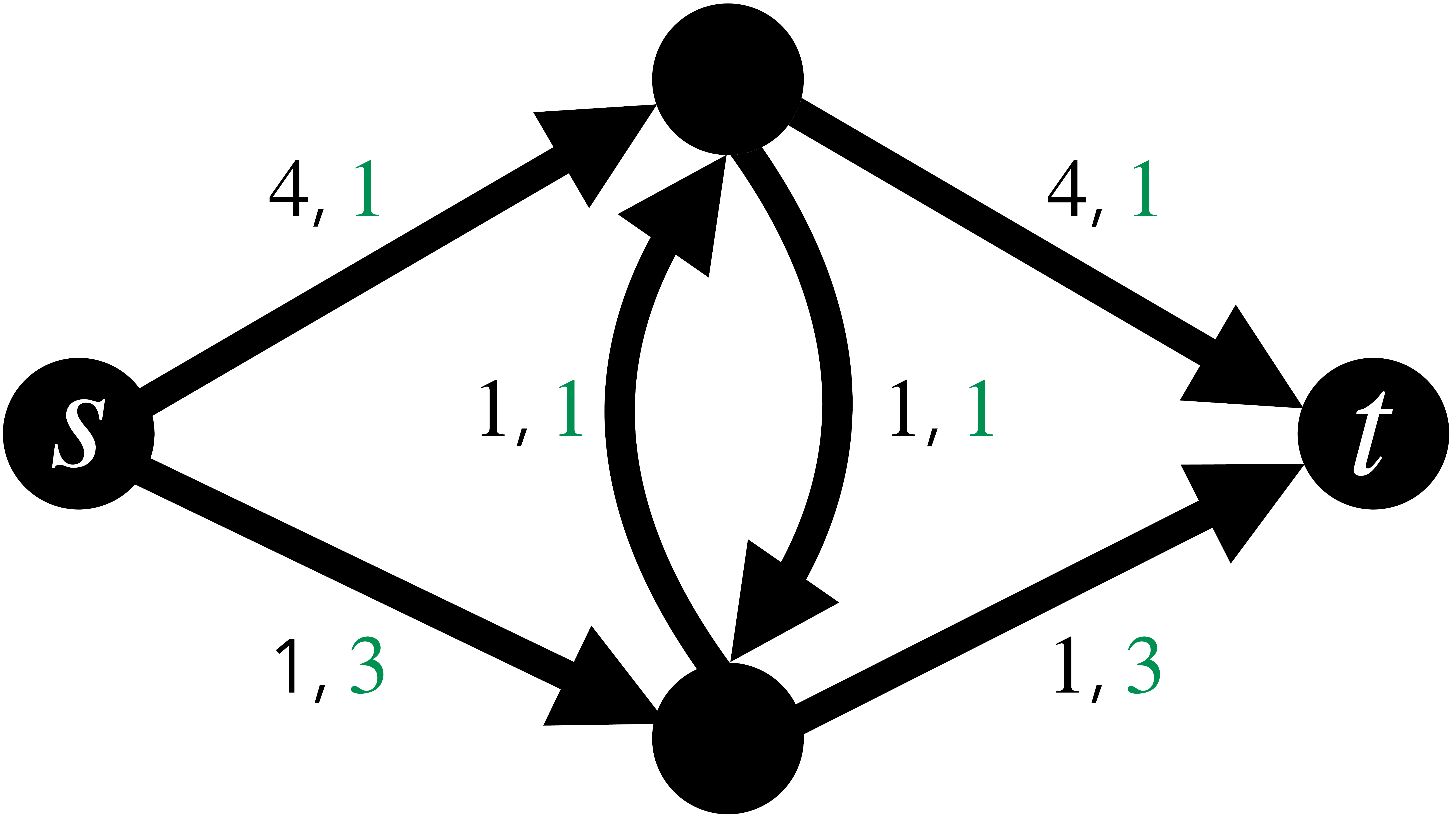}
        \caption{Digraph $D$.}\label{sfig:noDAG1}
    \end{subfigure}    \hfill
    \begin{subfigure}[b]{0.3\textwidth}
        \centering
        \includegraphics[width=\textwidth,trim=0mm 0mm 0mm 0mm, clip]{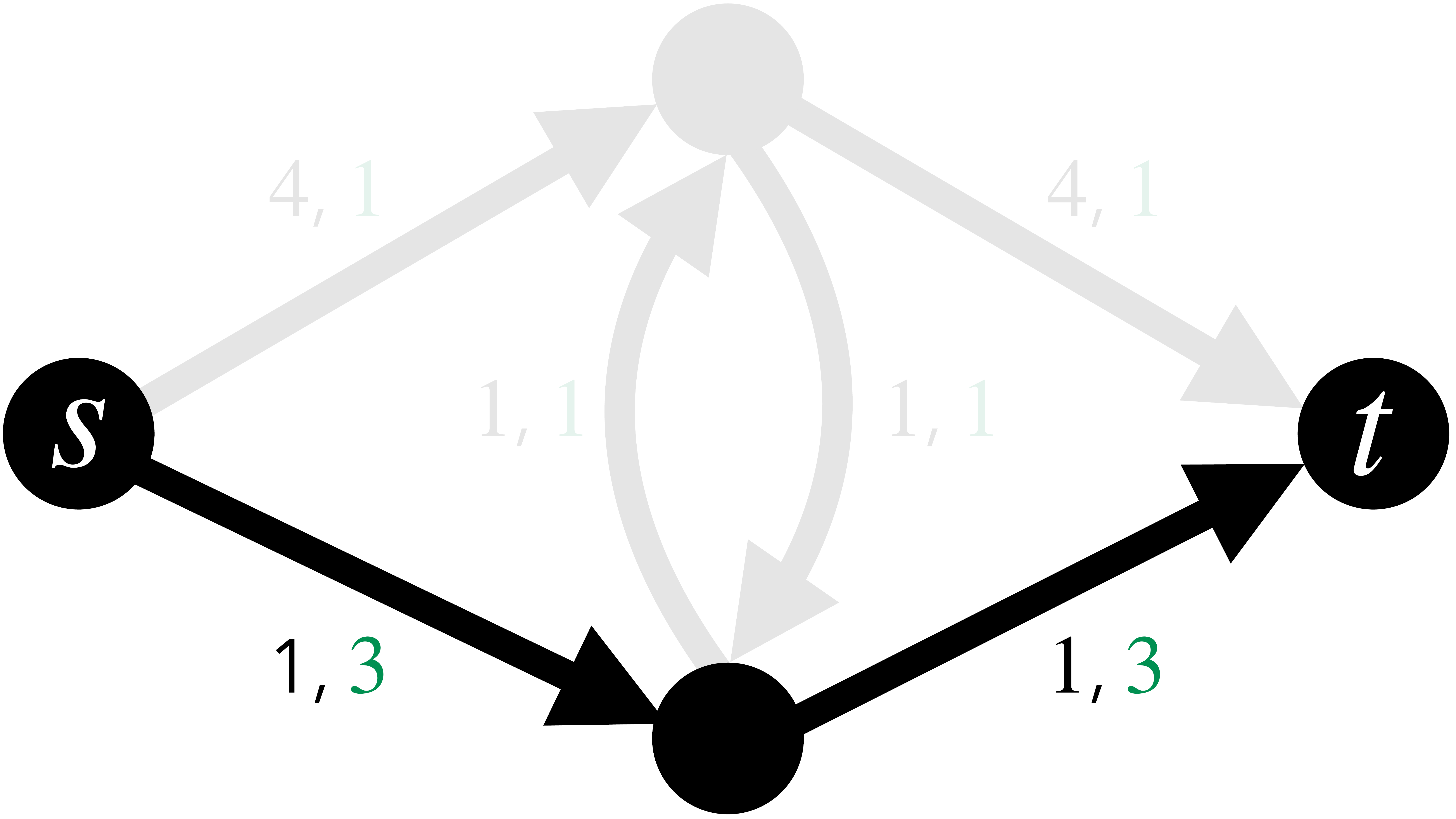}
        \caption{Lightest paths of $D$.}\label{sfig:noDAG2}
    \end{subfigure}    \hfill
    \begin{subfigure}[b]{0.3\textwidth}
        \centering
        \includegraphics[width=\textwidth,trim=0mm 0mm 0mm 0mm, clip]{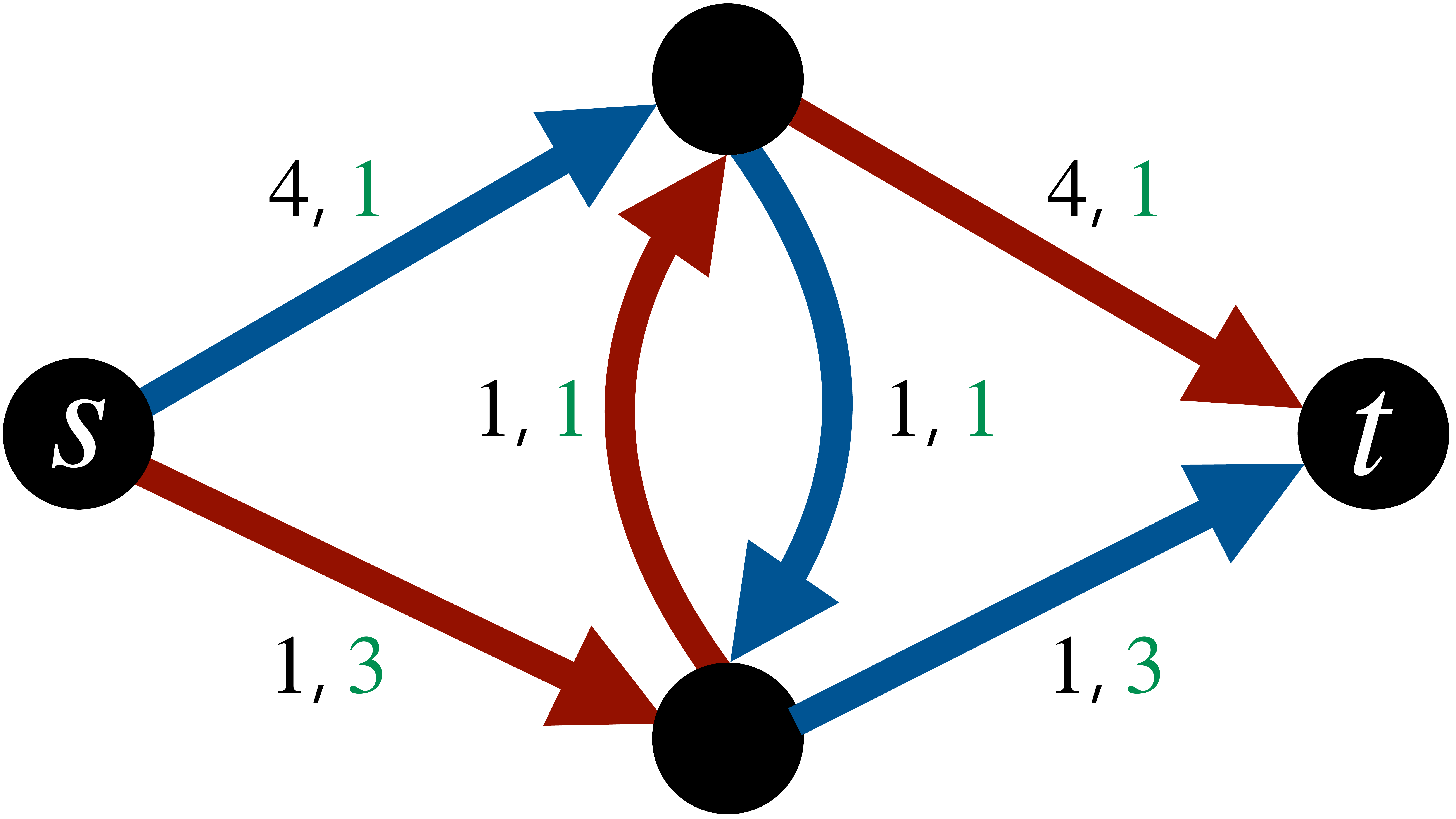}
        \caption{Lightest $5$-length paths of $D$.}\label{sfig:noDAG3}
    \end{subfigure}
    \caption{A digraph $D$ with $S = \{s\}$ and $T = \{t\}$ where the $5$-length lightest $S$-$T$ paths do not induce a DAG. \ref{sfig:noDAG1} gives $D$ where each arc is labeled with its weight (in black) and length (in green). \ref{sfig:noDAG2} shows how all lightest $S$-$T$ paths have weight $2$ and induce a DAG. \ref{sfig:noDAG3} shows how the two  $5$-length lightest $S$-$T$ paths (in blue and red) have weight $6$ and induce a digraph with a cycle.}\label{fig:noDAG}
\end{figure}

\subsection{Deterministic Integral Blocking Flows Paths via Flow Rounding}\label{sec:overviewSPB}
Lastly, we describe how we compute integral blocking flows in layered $S$-$T$ DAGs.

A somewhat straightforward adaptation of a \emph{randomized} algorithms of \citet{lotker2008improved} solves this problem in $\tilde{O}(\poly(h))$ time both in parallel and in CONGEST. This algorithm samples an integral $S$-$T$ flow in $D$ (i.e.\ a collection of arc-disjoint $S$ to $T$ paths) according to a carefully chosen distribution based on ``path counts'', deletes these paths and repeats. The returned solution is the flow induced by all paths that were ever deleted. Unfortunately \citet{lotker2008improved}'s algorithm seems inherently randomized and our goal is to solve this problem deterministically.

We derandomize the algorithm of \citet{lotker2008improved} in the following way. Rather than integrally sampling according to \citet{lotker2008improved}'s distribution and then deleting arcs that appear in sampled paths, we instead calculate the probability that an arc is in a path in this distribution and then ``fractionally delete'' it to this extent. We repeat this until every path between $S$ and $T$ has some arc which has been fully deleted. In other words, we run a smoothed version of \citet{lotker2008improved} which behaves (deterministically) like the algorithm of \citet{lotker2008improved} does in expectation. 
The fractional deletion values of arcs at the end of this process induce a blocking $S$-$T$ flow but a blocking flow that \emph{may be fractional}.  We call this flow the ``iterated path count flow.'' 

However, recall that our goal is to compute an \emph{integral} blocking flow in an $S$-$T$ DAG. Thus, we may naturally hope to round the iterated path count flow. Indeed, drawing on some flow rounding techniques of \citet{cohen1995approximate}, doing so is not too difficult in parallel. Unfortunately, it is less clear how to do so in CONGEST. Indeed, \citet{chang2020deterministic} state:
\begin{quote}\textit{
...Cohen’s algorithm that rounds a fractional flow into an integral flow does not seem to have an efficient implementation in CONGEST...} 
\end{quote}
Roughly, Cohen's technique relies on partitioning edges in a graph into cycles and paths and then rounding each cycle and path independently. The reason this seems infeasible in CONGEST is that the cycles and paths that Cohen's algorithm relies on can have unbounded diameter and so communicating within one of these cycles or paths is prohibitively slow. To get around this, we argue that, in fact, one may assume that these cycles and paths have low diameter \emph{if} we allow ourselves to discard some small number of arcs. This, in turn allows us to orient these cycles and paths and use them in rounding flows. We formalize such a decomposition with the idea of a $(1-\eps)$-near Eulerian partition.\footnote{Somewhat similar to our approach, \citet{chu2020graph} showed how to partition all but $O(n \log n)$ edges of a graph into short disjoint cycles. However, these guarantees are unsuitable for our us on e.g. graphs with $\Theta(n)$ edges since we may only discard a small fraction of all edges.} Arguing that discarding these arcs does bounded damage to our rounding then allows us to make use of Cohen-type rounding to deterministically round the path count flow, ultimately allowing us to compute $h$-length $(1+\epsilon)$-lightest path blockers.

\subsection{Summary of Our Algorithm}
We now summarize the above discussion with a bottom-up sketch of our algorithm and highlight where each of these components appear in our paper.

The most basic primitive that we provide is an algorithm for efficiently computing blocking integral flows in $h$-layer $S$-$T$ DAGs. To do so we make use of path count flows (formally defined in \Cref{sec:pathCounts}). In \Cref{sec:randMaxPaths} we observe that, essentially by the ideas of \citet{lotker2008improved}, sampling paths proportional to the path count flows gives an efficient randomized algorithm for blocking integral flows in such DAGs. In \Cref{sec:detMaxPaths} we give a deterministic algorithm for computing such flows. This algorithm relies on the idea of near Eulerian partitions (\Cref{sec:nearEul}) which is an adaptation of recent ideas in cycles covers for our purposes. Our deterministic algorithm takes the expected result of \citet{lotker2008improved} and deterministically rounds it by ``turning'' flow along the components of a near Eulerian partition and then repairs the resulting solution into a true flow by discarding flows not from $S$ to $T$. More generally, we show how to efficiently round any fractional flow on such a DAG with only a small loss in flow value.

Next, in \Cref{sec:shortestPathBlockers} we use our algorithms for blocking integral flows in $h$-layer $S$-$T$ DAGs to show how to compute $(1+\eps)$-lightest path blockers which, informally, are a collection of paths that share an edge with every $h$-length near-lightest path. We do this by constructing the length-weight expanded DAG (\Cref{sec:lengthExpanded}), a DAG that approximates the structure of $h$-length near-lightest paths. We then apply our blocking flow algorithms on this DAG, project the result back into our original graph and then ``decongest" the result by finding an appropriate subflow that respects capacities. We use an algorithm from \Cref{sec:sparseDecomp} to guarantee that the result is sparse (i.e.\ has small support size).

Lastly, in \Cref{sec:MW} we plug our $(1+\eps)$-lightest path blocker algorithm into a multiplicative-weights-type framework. In particular, we repeatedly compute a lightest path blocker, send some small amount of flow along the paths of this blocker and then update the weight of all edges that have flow sent along them by a multiplicative $(1+\eps)$.

The remainder of our paper gives applications and extensions of our results. In \Cref{sec:disjointPaths} we observe that our main result solves the aforementioned problem of \citet{chang2020deterministic} by giving deterministic algorithms for many disjoint paths problems in CONGEST. We also observe that our algorithms give essentially optimal parallel and distributed algorithms for maximum arc-disjoint paths. In \Cref{sec:expander} we give more details of how our results simplify expander decomposition constructions.  In \Cref{sec:bmatching} we give our new algorithms for bipartite $b$-matching based on our flow algorithms and in \Cref{sec:cutMatches} we show how to compute length-constrained cutmatches using our main theorem. Lastly, in \Cref{sec:multiComm} we observe that our length-constrained flow algorithms generalize to the multi-commodity setting.

\section{Preliminaries}\label{sec:preliminaries}
Before moving on to our own technical content, we briefly review some well-known  algorithmic tools and slight variants thereof (mostly for deterministic CONGEST).

\subsection{Deterministic CONGEST Maximum Independent Set}

We will rely on deterministic CONGEST primitives for maximal and maximum independent sets. Given graph $G = (V,E)$, a subset of vertices $V' \subseteq V$ is independent if no two vertices in $V'$ are adjacent in $G$. A maximal independent set (MIS) is an independent set $V'$ such that any $w \in V \setminus V'$ is adjacent to at least one node in $V'$.
If we are additionally given node weights $\{x_v\}_v$ where $x_v > 0$ for every $v$, then a maximum independent set is an independent set $V'$ maximizing $\sum_{v \in V'} x_v$; we say that an independent set is $\alpha$-approximate if its total weight is within $\alpha$ of that of the maximum independent set.


The following gives the deterministic CONGEST algorithm we will use for maximum independent set.
\begin{theorem}[\citet{bar2017distributed}]\label{thm:maxIS}
There is a deterministic CONGEST algorithm which given an instance of maximum independent in a graph $G = (V, E)$ with maximum degree $\Delta$ and node weights $\{x_v\}_v$, outputs a solution that is $\frac{1}{\Delta}$-approximate in time $O(\Delta + \log^* n)$.
\end{theorem}

\subsection{Deterministic Low Diameter Decompositions}

A well-studied object in metric theory is the low diameter decomposition which is usually defined as a distribution over vertex partitions \cite{linial1993low,miller2013parallel}. For our deterministic algorithms, we will make use of a deterministic version of these objects defined as follows where $G[V_i] := (V_i, \{\{u,v\} \in E: u,v \in V_i\})$ gives the induced graph on $V_i$.
\begin{definition}[Deterministic Low Diameter Decomposition]\label{dfn:DLDD}
Given graph $G = (V,E)$, a deterministic low diameter decomposition (DLDD) with diameter $d$ and cut fraction $\epsilon$ is a partition of $V$ into sets $V_1, V_2, \ldots$ where:
\begin{enumerate}
    \item \textbf{Low Diameter:} $G[V_i]$ has diameter at most $d$ for every $i$;
    \item \textbf{Cut Edges:} The number of cut edges is at most $\epsilon|E|$; i.e.\ $|\{e = (u, v): u \in V_i \wedge v \in V_j \wedge i \neq j\}| \leq \epsilon|E|$.
\end{enumerate}
\end{definition}
One can efficiently compute DLDDs deterministically in CONGEST as a consequence of many well-known results in distributed computing. We will use a result of \citet{chang2021strong} to do so.
\begin{theorem}\label{thm:LDD}
Given a graph $G = (V,E)$ and desired diameter $d$, one can compute a DLDD with diameter $d$ and cut fraction $\epsilon = \tilde{O}(\frac{1}{d})$ in deterministic CONGEST time $\tilde{O}(d)$.
\end{theorem}
\begin{proof}

Theorem 1.2 of \citet{chang2021strong} states that there is a deterministic CONGEST algorithm which, given a graph $G = (V,E)$ and desired diameter $d'$, computes a set $\bar{V} \subseteq V$ where $|\bar{V}| \leq \frac{1}{d'} \cdot |V|$ and $G[V \setminus \bar{V}]$ has connected components $C_1, C_2, \ldots, C_k$ where each $C_i$ has diameter at most $\tilde{O}(d')$ in $\tilde{O}(d')$ rounds.

Given graph $G = (V,E)$ we can compute a DLDD in $G$ by applying the above result in a new graph $G' = (V', E')$. For each vertex $v \in V$, $G'$ will have a clique of $\Delta(v)$-many vertices where $\Delta(v)$ is the degree of $v$ in $G$. We then connect these cliques in the natural way. More formally, to construct $G'$ we do the following. For each $v$ with edges to vertices $v_1, v_2, \ldots, v_{\Delta(v)}$ we create a clique of vertices $v(v_1), v(v_2), \ldots, v(v_{\Delta(v)})$. Next, for each edge $e = \{u, v\}$ in $E$, we add the edge $\{v(u), u(v)\}$ to $G'$. Observe that each vertex of $G'$ corresponds to exactly one edge in $G$; that is, $v(u)$ in $V'$ corresponds to the edge $\{u,v\} \in E$.

Next, we apply the above theorem of \citet{chang2021strong} to $G'$ to get set $\bar{V}$. Let $\bar{E} \subseteq E$ be the set of edges to which these vertices correspond. We return as our solution $\bar{E}$. Observe that the size of $\bar{E}$ is
\begin{align*}
    |\bar{E}| &\leq |\bar{V}|\\
    & \leq \frac{1}{d'} \cdot |V'|\\
    & = \frac{2}{d'} |E|.
\end{align*}
Letting $d' = \frac{1}{\tilde{\Theta}(1)} \cdot d$ for an appropriately large hidden poly-log in $\tilde{\Theta}(1)$ gives us that each component in $G$ has diameter at most $d$ since otherwise there would be a component in $G'$ after deleting $\bar{v}$ with diameter more than $d'$. Likewise, the above gives us cut fraction at most $\tilde{O}(\frac{1}{d})$.

Simulating a CONGEST algorithm on $G'$ on $G$ is trivial since each vertex can simulate its corresponding clique and so the entire algorithm runs in time $\tilde{O}(d') = \tilde{O}(d)$.
\end{proof}

\subsection{Sparse Neighborhood Covers}\label{sec:sparseNeigh}
A closely related notion to low diameter decompositions is that of the sparse neighborhood cover~\cite{awerbuch1990sparse}. We use the following definition phrased in terms of partitions.

\begin{definition}[Sparse Neighborhood Cover]
Given a simple graph $G = (V,E)$, an $s$-sparse $k$-neighborhood cover with weak-diameter $d$ and overlap $o$ is a set of partitions $\mcV_1, \mcV_2, \ldots, \mcV_{s}$ of $V$ where each partition is a collection of disjoint vertex sets $V_i^{(j)} \subset V$ whose union is $V$, i.e.,  $\mcV_i= \{V_i^{(1)}, V_i^{(2)}, \ldots \}$ and:
\begin{enumerate}
    \item \textbf{Weak-Diameter and Overlap:} Each $V_i^{(j)}$ comes with a rooted tree $T_i^{(j)}$ in $G$ of diameter at most $d$ that spans all nodes in $V_i^{(j)}$; Any node in $G$ is contained in at most $o$ trees overall. 
    \item \textbf{Neighborhood Covering:} For every node $v$ its $k$ neighborhood $B_k(v)$, containing all vertices in $G$ within distance $k$ of $v$, is fully covered by at least one cluster, i.e., $\forall v \ \exists i,j: \  B_k(v) \subseteq V_i^{(j)}$.
\end{enumerate}
\end{definition}

The below summarizes the current state of the art in deterministic sparse neighborhood covers in CONGEST.
\begin{lemma}[\cite{ghaffari2018derandomizing, rozhovn2020polylogarithmic,chang2021strong}]\label{lem:sparseNeigh}
There is a deterministic CONGEST algorithm which given any radius $k \geq 1$, computes an $s$-sparse $k$-neighborhood cover with $s,o = \tilde{O}(1)$ and diameter at most $\tilde{O}(k)$ in $\tilde{O}(k)$ time.

Furthermore, there is a deterministic CONGEST algorithm which given an $O(l)$-bit value $x_v$ for every $v$ computes $x_{i,v}$ for every $v$ and $i$ in $\tilde{O}(k + l)$ rounds, where $x_{i,v}$ is the maximum $x$-value among nodes in the same cluster as $v$ in the partition $\mcV_i$. That is, letting $\mcV_i(v)$ be the one cluster in $\mcV_i$ containing $v$, we have
\begin{align*}
   x_{i,v} = \max_{u \in \mcV_i(v)} x_u.
\end{align*}
\end{lemma}

\subsection{Cycle Covers}\label{sec:cycCovers}

Our flow rounding algorithm will make use of low diameter cycles. Thus, it will be useful for us to make use of some recent insights into distributely and deterministically decomposing graphs into low diameter cycles. We define the diameter of a cycle $C$ as $|C|$ and the diameter of a collection of cycles $\mcC$ as the maximum diameter of any cycle in it. Likewise the congestion of $\mcC$ is $\max_{e} |\{C : e \in C\}|$.

The idea of covering a graph with low congestion cycles is well-studied \cite{chu2020graph,parter2019optimal,hitron2021general} and formalized by the idea of a cycle cover.
\begin{definition}[Cycle Cover]
Given a simple graph $G = (V,E)$ where $E_0$ is the set of all non-bridge edges\footnote{Recall that a bridge edge of a graph is one whose removal increases the number of connected components in the graph.} of $G$, a $(d, c)$ cycle cover is a collection of (simple) cycles $\mcC$ in $G$ such that:
\begin{enumerate}
    \item \textbf{Covering:} Every $e \in E_0$ is contained in some cycle of $\mcC$;
    \item \textbf{Low Diameter:} $\max_{C \in \mcC} |C| \leq d$;
    \item \textbf{Low Congestion:} $\max_{e \in E} |\{ C : e \in C\}| \leq c$.
\end{enumerate}
\end{definition}
We now formally define the parameter $\rho_{CC}$; recall that this parameter appears in the running time of our deterministic CONGEST algorithm in our main theorem (\Cref{thm:main}).

\begin{definition}[$\rho_{CC}$]\label{dfn:rhoCC}
 Given a deterministic CONGEST algorithm that constructs a $(d,c)$ cycle cover in worst-case time $T$ in graphs of diameter $D$, we say that the quality of this algorithm is $\max\{\frac{d}{D}, c,\frac{T}{D}\}$. We let $\rho_{CC}$ be the smallest quality of any deterministic CONGEST algorithm for constructing cycle covers.
\end{definition}

The following summarizes the current state-of-the-art in deterministic cycle cover computation in CONGEST.
\begin{theorem}[\cite{parter2019optimal,hitron2021general}]
There is a deterministic CONGEST algorithm that given a graph $G$ with diameter $D$ computes a $(d, c)$ cycle cover with $d = 2^{O(\sqrt{\log n})} \cdot D$ and $c = 2^{O(\sqrt{\log n})}$ in time $2^{O(\sqrt{\log n})}\cdot D$. In other words, $\rho_{CC} \leq 2^{O(\sqrt{\log n})}$
\end{theorem}

\section{Path Counts for $h$-Layer $S$-$T$ DAGs}\label{sec:pathCounts}

We begin by recounting the notion of path counts which we will use for our randomized algorithm to sample flows and for our deterministic algorithms to compute the iterated path count flow. This idea has been used in several prior works \cite{lotker2008improved,cohen1995approximate,chang2020deterministic}.

Suppose we are given an $h$-layer $S$-$T$ DAG $D$ with capacities $\U$. We define these path counts as follows. We define the capacity of a path as the product of its edge capacities, namely given a path $P$ we let $\U(P) := \prod_{a \in P} \U_a$. Recall that we use $\mcP(S,T)$ to stand for all paths between $S$ and $T$. We will slightly abuse notation and let $\mcP(v,T) = \mcP(\{v\},T)$ and $\mcP(S,v) = \mcP(S, \{v\})$. For vertex $v$ we let $n_v^+$ be the number of paths from $v$ to $T$, weighted by $\U$, namely $n_v^+ := \sum_{P \in \mcP(v, T)} \U(P)$. Symmetrically, we let $n_v^- := \sum_{P \in \mcP(S, v)} \U(P)$. For any arc $a=(u,v)$, we define $n_a$ as
\begin{align*}
    n_a := n^-_u \cdot \U_a \cdot n^+_v.
\end{align*}
Equivalently, we have that $n_a$ is the number of paths in $\mcP(S,T)$ that use $a$ weighted by capacities:
\begin{align*}
    n_a = \sum_{P \in \mcP(S,T) : a \in P} \U(P).
\end{align*}
It may be useful to notice that if we replace each arc $a$ with $\U_a$-many parallel arcs then $n_a$ exactly counts the number of unique paths from $S$ to $T$ that use $a$ in the resulting (multi) digraph. A simple dynamic-programming type algorithms that does a ``sweep'' from $S$ to $T$ and $T$ to $S$ shows that one can efficiently compute the path counts.
\begin{lemma}\label{lem:computePathCounts}
Let $D$ be a capacitated $h$-layer $S$-$T$ DAG. Then one can compute $n_v^+$ and $n_v^-$ for every vertex $v$ and $n_a$ for every arc $a$ in:
\begin{enumerate}
    \item Parallel time $O(h)$ with $m$ processors;
    \item CONGEST time $\tilde{O}\left(h^2 \right)$.
\end{enumerate}
\end{lemma}
\begin{proof}

To compute $n_a$ it suffices to compute $n^+_v$ and $n^-_v$. We proceed to describe how to compute $n^-_v$; computing $n_v^+$ is symmetric.

First, notice that $n^-_v$ can be described by the recurrence
\begin{align*}
        n_v^- :=
    \begin{cases}
        1 & \text{if $v \in S$}\\
        \sum_{ (u, v) \in \delta^-(v)} \U_{uv} \cdot n_u^- & \text{otherwise}
    \end{cases}
\end{align*}

We repeat the following for iteration $i = 2, \ldots, h +1$. Let $V_i$ be all vertices in the $i$th layer of our graph. In iteration $i$ we will compute $n_v^-$ for every $v \in V_i$ by applying the above recurrence.

Running one of the above iterations in parallel is trivial to do in $O(1)$ parallel time with $m$ processors, leading to the above parallel runtime. Running one iteration of this algorithm in CONGEST requires that every vertex in $v \in V_{j}$ for $j < i$ broadcast its $n_v^-$ to its neighbors. Since $n_v^- \leq (n \cdot \U_{\max})^h$ this can be done in $h \left(1 + \frac{\log \U_{\max}}{\log n} \right)$ rounds of CONGEST, leading to the stated CONGEST runtime.
\end{proof}

\section{Randomized Blocking Integral Flows in $h$-Layer DAGs}\label{sec:randMaxPaths}
We now describe how to compute blocking integral flows in $h$-layer $S$-$T$ DAGs with high probability by using the path counts of the previous section. This is the general capacities version of the problem described in \Cref{sec:overviewSPB}. More or less, the algorithm we use is one of \citet{chang2020deterministic} adapted to the general capacities case; the algorithm of \citet{chang2020deterministic} is itself an adaptation of an algorithm of \citet{lotker2008improved}. As such, we defer the proofs in this section to \Cref{sec:deferredProofs}; we mostly include these results for the sake of completeness.

Our randomized algorithm will repeatedly sample an integral flow proportional to the path counts of \Cref{sec:pathCounts}, add this to our existing flow, reduce capacities and then repeat. We will argue that we need only iterate this process a small number of times until we get a blocking integral flow by appealing to the fact that ``high degree'' paths have their capacities reduced with decent probability.

One can see this as essentially running the randomized MIS algorithm of \citet{luby1986simple} but with two caveats: (1) the underlying graph in which we compute an MIS has a node for every path between $S$ and $T$ and so has up to $O(n^h)$-many nodes; as such we cannot explicitly construct this graph but rather can only implicitly run Luby's algorithm on it; (2) Luby's analysis assumes nodes attempt to enter the MIS independently but our sampling will have some dependencies between nodes (i.e.\ paths) entering the MIS which must be addressed in our analysis.


More formally, suppose we are given a capacitated $S$-$T$ DAG $D$. For a given path $P \in \mcP(S,T)$ we let $\Delta_P$ be $\sum_{P'} \prod_{a \in P' \setminus P}\U_a$ be the ``degree'' of path $P$ where the sum over $P'$ ranges over all $P'$ that share at least one arc with $P$ and are in $\mcP(S,T)$. We let $\Delta = \max_{P \in \mcP(S,T)} \Delta_P$ be the maximum degree. Similarly, we let $\mcP_{\approx \max} := \{ P : \Delta_P \geq \frac{\Delta}{2}\} $ be all paths with near-maximum degree. The following summarizes the flow we repeatedly compute; in this lemma the constant $\frac{2046}{2047}$ is arbitrary and could be optimized to be much smaller.
\begin{restatable}{lemma}{randCovPaths}\label{lem:randCovPaths}
Given a $h$-layer $S$-$T$ DAG $D$ with capacities $\U$ and $\tilde{\Delta}$ satisfying $\frac{\Delta}{2} \leq \tilde{\Delta} \leq \Delta$, one can sample an integral $S$-$T$ flow $f$ where for each $P \in \mcP_{\approx \max}$ we have $ \prod_{a \in P}(\U_a - f_a) \leq \frac{2047}{2048} \cdot \U(P)$ with probability at least $\Omega(1)$. This can be done in:
\begin{enumerate}
    \item Parallel time $O(h)$ with $m$ processors;
    \item CONGEST time $\tilde{O}\left(h^2\right)$ with high probability.
\end{enumerate}
\end{restatable}

Repeatedly applying the above lemma gives our randomized algorithm for blocking integral $S$-$T$ flows.
\begin{restatable}{lemma}{randMax}\label{lem:randMax}
There is an algorithm which, given an $h$-layer $S$-$T$ DAG $D$ with capacities $\U$, computes an integral $S$-$T$ flow that is blocking in:
\begin{enumerate}
    \item Parallel time $\tilde{O}(h^3)$ with $m$ processors with high probability;
    \item CONGEST time $\tilde{O}(h^4)$ with high probability.
\end{enumerate}
\end{restatable}

\section{Deterministic and Distributed Near Eulerian Partitions}\label{sec:nearEul}
In the previous section we showed how to efficiently compute blocking integral flows in $h$-layer DAGs with high probability. In this section, we introduce the key idea we make use of in doing so deterministically, a near Eulerian partition.

Informally, a near Eulerian partition will discard a small number of edges and then partition the remaining edges into cycles and paths. Because these cycles and paths will have small diameter in our construction, we will be able to efficiently orient them in CONGEST. In \Cref{sec:detMaxPaths} we will see how to use these oriented cycles and paths to efficiently round flows in a distributed fashion in order to computer a blocking integral flow in $h$-layer DAGs.

We now formalize the idea of a $(1-\eps)$-near Eulerian partition.

\begin{definition}[$(1-\eps)$-Near Eulerian Partition]\label{def:Euler part}
Let $G = (V, E)$ be an undirected graph and $\epsilon \geq 0$. A $(1-\eps)$-near Eulerian partition $\mcH$ is a collection of edge-disjoint cycles and paths in $G$, where 
\begin{enumerate}
    \item \textbf{$(1-\eps)$-Near Covering:} The number of edges in $E[\mcH]$ is at least $(1-\eps) \cdot |E|$;
    \item \textbf{Eulerian Partition:} Each vertex is the endpoint of at most one path in $\mcH$.
\end{enumerate}
\end{definition}

The following is the main result of this section and summarizes our algorithms for construction $(1-\eps)$-near Eulerian partitions. In what follows we say that a cycle is oriented if every edge is directed so that every vertex in the cycle has in and out degree $1$; a path $P$ is oriented if it has some designated source and sink $s_P$ and $t_P$. We say that a collection of paths and cycles $\mcH$ is oriented if each element of $\mcH$ is oriented. In CONGEST we will imagine that a cycle is oriented if each vertex knows the orientation of its incident arcs and a path is oriented if every vertex knows which of its neighbors are closer to $s_P$.
\begin{restatable}{lemma}{nearEuler}\label{lem:nearEuler}
One can deterministically compute an oriented $(1-\eps)$-near Eulerian partitions in:
\begin{enumerate}
    \item Parallel time $\tilde{O}(1)$ with $m$ processors and $\epsilon = 0$; 
    \item CONGEST time $\tilde{O}(\frac{1}{\eps^5} \cdot (\rho_{CC})^{10})$ for any $\eps > 0$. 
\end{enumerate}
\end{restatable}
Again, see \Cref{sec:cycCovers} for a definition of $\rho_{CC}$.

\subsection{High-Girth Cycle Decompositions}
In order to compute our near Eulerian partitions we will make use of a slight variant of cycle covers which we call high-girth cycle decompositions (as introduced in \Cref{sec:cycCovers}). The ideas underpinning these decompositions seem to be known in the literature but there does not seem to be a readily citable version of quite what we need; hence we give details below.

To begin, in our near-Eulerian partitions we would like for our cycles to be edge-disjoint so that each cycle can be rounded independently. Thus, we give a subroutine for taking a collection of cycles and computing a large edge-disjoint subset of this collection. This result comes easily from applying a deterministic approximation algorithm for maximum independent set (MIS). Congestion and dilation in what follows are defined in \Cref{sec:cycCovers}.
\begin{lemma}\label{lem:cycDecong}
There is a deterministic CONGEST algorithm that, given a graph $G = (V,E)$ and a collection of (not necessarily edge-disjoint) cycles $\mcC$ with congestion $c$ and diameter $d$, outputs a set of edge disjoint cycles $\mcC' \subseteq \mcC$ which satisfies $|E[\mcC']| \geq \frac{1}{d^2c^2} \cdot |E[\mcC]|$ in time $\tilde{O}(c^3d^3)$.
\end{lemma}
\begin{proof}
Our algorithm simply computes an approximately-maximum independent set in the conflict graph which has a node for each cycle. In particular, we construct conflict graph $G' = (\mcC, E')$ as follows. Our vertex set is $\mcC$. We include edge $\{C , C'\}$ in $E'$ if $C \in \mcC$ and $C' \in \mcC$ overlap on an edge; that is, if $E[C] \cap E[C'] \neq \emptyset$.

Observe that since each cycle in $\mcC$ has at most $d$-many edges and since each edge is in at most $c$-many cycles, we have that the maximum degree of $G'$ is $cd$. Next, we let the ``node-weight'' of cycle $C \in \mcC$ be $|C|$. We apply \Cref{thm:maxIS} with these node-weights to compute a $\frac{1}{cd}$-approximate maximum independent set $\mcC'$. We return $\mcC'$ as our solution.

First, observe that since $\mcC'$ is an independent set in $G'$, we have that the cycles of $\mcC'$ are indeed edge-disjoint.

Next, we claim that $|E[\mcC']| \geq \frac{1}{d^2c^2} \cdot |E[\mcC]|$. Since \Cref{thm:maxIS} guarantees that $\mcC'$ is a $\frac{1}{dc}$-approximate solution, to show this, it suffices to argue that $|E[\mcC^*]| \geq \frac{1}{dc} \cdot |E[\mcC]|$ where $\mcC^* \subseteq \mcC$ is the set of edge-disjoint cycles of maximum edge cardinality, i.e.\ the maximum node-weight independent set in $G'$. However, notice that since the total node weight in $G'$ is $\sum_{C \in \mcC} |E[C]|$ and the max degree in $G'$ is at most $cd$, we have that the maximum node-weight independent set in $G'$ must have node-weight at least $\frac{1}{cd} \sum_{C \in \mcC} |E[C]| \geq \frac{1}{cd} |E[\mcC]|$. Thus, we conclude that $|E[\mcC']| \geq \frac{1}{d^2c^2} \cdot |E[\mcC]|$.

Next, we argue that we can implement the above in the stated running times. Computing our $\frac{1}{cd}$-approximate maximum independent set on $G'$ takes deterministic CONGEST time $\tilde{O}(cd)$ on $G'$ by \Cref{thm:maxIS}. Furthermore, we claim that we can simulate a CONGEST algorithm on $G'$ in $G$ with only an overhead of $O(c^2 d^2)$. In particular, since the maximum degree on $G'$ is $cd$, in each CONGEST round on $G'$ each node (i.e.\ cycle in $G$) receives at most $cd$-many messages. Fix a single round of CONGEST on $G'$. We will maintain the invariant that if $v \in V$ is a node in a cycle $C \in \mcC$, then in our simulation $v$ receives all the same messages as $C$ in our CONGEST algorithm on $G'$. We do so by broadcasting all messages that $C$ receives in this one round on $G'$ to all nodes in $C$. As a cycle in $G'$ receives at most $cd$ messages in one round of CONGEST on $G'$ and each edge is in at most $c$-many cycles, it follows that in such a broadcast the number of messages that need to cross any one edge is at most $c^2 d$. Since the diameter of each cycle is at most $d$, we conclude that this entire broadcast can be done deterministically in time $O(c^2 d^2)$, giving us our simulation.

Combining this $O(c^2 d^2)$-overhead simulation with the $\tilde{O}(cd)$ running time of our approximate maximum independent set algorithm on $G'$ gives an overall running time of $O(c^3 d^3)$.
\end{proof}

Recall that the girth of a graph is the minimum length of a cycle in it. The following formalizes the notion of high-girth cycle decompositions that we will need.

\begin{definition}[High-Girth Cycle Decomposition]\label{dfn:HGCC}
Given a graph $G = (V,E)$ and $\eps > 0$ where $E_0$ are all non-bridge edges of $G$, a high-girth cycle decomposition with diameter $d$ and deletion girth $k$ is a collection of edge-disjoint (simple) cycles $\mcC$ such that:
\begin{enumerate}
    \item \textbf{High Deletion Girth:} The graph $(V, E \setminus E[\mcC])$ has girth at least $k$.
    \item \textbf{Low Diameter:} $\max_{C \in \mcC} |C| \leq d $;
\end{enumerate}
\end{definition}
The following theorem gives the construction of high-girth cycle decompositions that we will use.
\begin{theorem}\label{thm:highGirthCCAlg}
There is a deterministic CONGEST algorithm that, given a graph $G = (V,E)$ and desired girth $k \geq 0$, computes a high-girth cycle decomposition with diameter $\tilde{O}(k \cdot \rho_{CC})$ and girth $k$ in time $\tilde{O}(k^5 \cdot (\rho_{CC})^{10})$.
\end{theorem}
\begin{proof}
The basic idea is: take a sparse neighborhood cover; compute cycle covers on each part of our neighborhood cover; combine all of these into a single cycle cover; decongest this cycle cover into a collection of edge-disjoint cycles; delete these cycles and; repeat. 

More formally, our algorithm is as follows, We initialize our collection of cycles $\mcC$ to $\emptyset$.

Next, we repeat the following $\tilde{\Theta}\left(k^2 \cdot (\rho_{CC})^4\right)$ times. Apply \Cref{lem:sparseNeigh} to compute an $\tilde{O}(1)$-sparse $k$-neighborhood cover of $G$ with diameter $\tilde{O}(k)$ and overlap $\tilde{O}(1)$. Let $\mcV_1, \mcV_2, \ldots$ be the partitions of this neighborhood cover. By definition of a neighborhood cover, for each $\mcV_i$ and each $V_i^{(j)} \in \mcV_i$, we have that $V_i^{(j)}$ comes with a tree $T_i^{(j)}$ where each node in the tree is in $\tilde{O}(1)$ other $V_i^{(j)}$. We let $H_i^{(j)} := G[V_i^{(j)}] \cup T_i^{(j)}$ be the union of this tree and the graph induced on $V_i^{(j)}$. By the guarantees of our neighborhood cover we have that the diameter of $H_i^{(j)}$ is at most $\tilde{O}(k)$. We then compute a cycle cover $\mcC_i^{(j)}$ of each $H_i^{(j)}$ with diameter $\tilde{O}(k \cdot \rho_{CC})$ and congestion $\rho_{CC}$ (we may do so by definition of $\rho_{CC}$). We let $\mcC_0 = \bigcup_{i, j} \mcC_i^{(j)}$ be the union of all of these cycle covers. Next, we apply \Cref{lem:cycDecong} to compute a large edge-disjoint subset $\mcC_0' \subseteq \mcC_0$ of $\mcC_0$. We add $\mcC_0'$ to $\mcC$ and delete from $G$ any edge that occurs in a cycle in $\mcC_0'$.

We first argue that the solution we return is indeed a high-girth cycle decomposition. Our solution consists of edge-disjoint cycles by construction. Next, consider one iteration of our algorithm. Observe that since each $\mcC_i^{(j)}$ has congestion at most $\rho_{CC}$, it follows by the $\tilde{O}(1)$ overlap and $\tilde{O}(1)$ sparsity of our neighborhood cover that $\mcC_0$ has congestion $\tilde{O}(\rho_{CC})$. Likewise, since each $H_i^{(j)}$ has diameter $\tilde{O}(k)$, it follows that each $\mcC_i^{(j)}$ has diameter at most $\tilde{O}(k \cdot \rho_{CC})$ and so $\mcC_0$ has diameter at most $\tilde{O}(k \cdot \rho_{CC})$. Thus, $\mcC_0$ has congestion at most $\tilde{O}(\rho_{CC})$ and diameter at most $\tilde{O}(k \cdot \rho_{CC})$. Since $\mcC_0' \subseteq \mcC_0$, it immediately follows that the solution we return has diameter at most $\tilde{O}(k \cdot \rho_{CC})$. 

It remains to show that the deletion of our solution induces a graph with high girth. Towards this, observe that applying the congestion and diameter of $\mcC_0$ and the guarantees of \Cref{lem:cycDecong}, it follows that
\begin{align}
    |E[\mcC_0']| \geq \tilde{\Omega}\left(\frac{1}{k^2 (\rho_{CC})^4}\right) \cdot |E[\mcC_0]|.\label{eq:x}
\end{align}
On the other hand, let $E_0$ be all edges in cycle of diameter at most $k$ at the beginning of this iteration. Consider an $e \in E_0$. Since $\mcV_1, \mcV_2, \ldots$ is a $k$-neighborhood cover we know that there is some $\mcC_i^{(j)}$ which contains a cycle which contains $e$. Thus, we have
\begin{align}
    |E[\mcC_0]| \geq |E_0|.\label{eq:y}
\end{align}
Combining \Cref{eq:x} and \Cref{eq:y}, we conclude that 
\begin{align*}
    |E[\mcC_0']| \geq \tilde{\Omega}\left(\frac{1}{k^2 (\rho_{CC})^4}\right) \cdot|E_0|.
\end{align*}
However, since in this iteration we delete every edge in $E[\mcC_0']$, it follows that we reduce the number of edges that are in a cycle of diameter at most $k$ by at least a $1 - \tilde{\Omega}\left(\frac{1}{k^2 (\rho_{CC})^4}\right)$ multiplicative factor. Since initially the number of such edges is at most $|E|$, it follows that after $\tilde{O}(k^2 \cdot (\rho_{CC})^4)$-many iterations we have reduced the number of edges in a cycle of diameter at most $k$ to $0$; in other words, our graph has girth at most $k$. This shows the high girth of our solution, namely that $(V, E \setminus E[\mcC])$ has girth at least $k$ after the last iteration of our algorithm.

Next, we argue that we achieve the stated running times. Fix an iteration.
\begin{itemize}
    \item By the guarantees of \Cref{lem:sparseNeigh}, the sparse neighborhood cover that we compute takes time $\tilde{O}(k)$. 
    \item We claim that by definition of $\rho_{CC}$, the $\tilde{O}(k)$ diameter of each part in our sparse neighborhood cover and the $\tilde{O}(1)$ overlap of our sparse neighborhood cover, we can compute every $\mcC_i^{(j)}$ in time $\tilde{O}(k \cdot \rho_{CC})$. Specifically, for a fixed $i$ we run the cycle cover algorithm simultaneously in meta-rounds, each consisting of $\tilde{\Theta}(1)$ rounds. In each meta-round a node can send the messages that it must send for the cycle cover algorithm of each of the $H_i^{(j)}$ to which it is incident by our overlap guarantees. Since the total number of $i$ is $\tilde{O}(1)$ by our sparsity guarantee, we conclude that we can compute all $\mcC_i^{(j)}$ in a single iteration in at most $\tilde{O}(k \cdot \rho_{CC})$ time.
    \item Lastly, by the guarantees of \Cref{lem:cycDecong} and the fact that $\mcC_0$ has congestion at most $\tilde{O}(\rho_{CC})$ and diameter at most $\tilde{O}(k \cdot \rho_{CC})$, we can compute $\mcC_0'$ in time $\tilde{O}(k^3 \cdot (\rho_{CC})^6)$.
\end{itemize}
Combining the above running times with the fact that we have $\tilde{\Theta}\left(k^2 \cdot (\rho_{CC})^4\right)$-many iterations gives us a running time of $\tilde{O}(k^5 \cdot (\rho_{CC})^{10})$.
\end{proof}

\subsection{Efficient Algorithms for Computing Near Eulerian Partitions}
We conclude by proving the main section of this theorem, namely the following which shows how to efficiently compute near Eulerian partitions in deterministic CONGEST by making use of our high-girth cycle decomposition construction and DLDDs.
\nearEuler*
\begin{proof}
The parallel result is well-known since a $1$-near Eulerian partition is just a so-called Eulerian partition; see e.g.\ \citet{karp1989survey}. 

The rough idea of our CONGEST algorithm is as follows. First we compute a high-girth cycle decomposition (\Cref{dfn:HGCC}), orient these cycles and remove all edges covered by this decomposition. The remaining graph has high girth by assumption. Next we compute a DLDD (\Cref{dfn:DLDD}) on the remaining graph; by the high girth of our graph each part of our DLDD is a low diameter tree. Lastly, we decompose each such tree into a collection of paths.


More formally, our CONGEST algorithm to return cycles $\mcC$ and paths $\mcP$ is as follows. Apply \Cref{thm:highGirthCCAlg} to compute a high-girth cycle decomposition $\mcC$ with deletion girth $\tilde{\Theta}(\frac{1}{\eps})$ and diameter $\tilde{O}(\frac{1}{\eps} \cdot \rho_{CC})$. Orient each cycle in $\mcC$ and delete from $G$ any edge in a cycle in $\mcC$. Next, apply \Cref{thm:LDD} to compute a DLDD with diameter $\tilde{\Theta}(\frac{1}{\eps})$ and cut fraction $\eps$. Delete all edges cut by this DLDD. Since $\mcC$ has deletion girth $\tilde{\Theta}(\frac{1}{\eps})$, by appropriately setting our hidden constant and poly-logs, it follows that no connected component in the remaining graph contains a cycle; in other words, each connected component is a tree with diameter $\tilde{\Theta}(\frac{1}{\eps})$.

We decompose each tree $T$ in the remaining forest as follows. Fix an arbitrary root $r$ of $T$. We imagine that each vertex of odd degree in $T$ starts with a ball. Each vertex waits until it has received a ball from each of its children. Once a vertex has received all such balls, it pairs off the balls of its children arbitrarily, deletes these balls and adds to $\mcP$ the concatenation of the two paths traced by these balls in the tree. It then passes its up to one remaining ball to its parent. Lastly, we orient each path in $\mcP$ arbitrarily.

We begin by arguing that the above results in a $(1-\eps)$-near Eulerian partition. Our paths and cycles are edge-disjoint by construction. The only edges that are not included in some element of $\mcC \sqcup \mcP$ are those that are cut by our DLDD; by our choice of parameters this is at most an $\eps$ fraction of all edges in $E$. To see the Eulerian partition property, observe that every vertex of odd degree in $G[\mcC \sqcup \mcP]$ is an endpoint of exactly one path in $\mcP$ since each odd degree vertex starts with exactly one ball. Likewise, a vertex of even degree will never be the endpoint of a path since no such vertex starts with a ball.

It remains to argue that the above algorithm achieves the stated CONGEST running time.
\begin{itemize}
    \item Computing $\mcC$ takes time at most $\tilde{O}(\frac{1}{\eps^5} \cdot (\rho_{CC})^{10})$ by \Cref{thm:highGirthCCAlg}. Furthermore, by \Cref{thm:highGirthCCAlg}, each cycle in $\mcC$ has diameter $\tilde{O}(\frac{1}{\eps} \cdot \rho_{CC})$ and so can be oriented in time $\tilde{O}(\frac{1}{\eps} \cdot \rho_{CC})$.
    \item Computing our DLDD takes time $\tilde{O}(\frac{1}{\eps})$ by \Cref{thm:LDD}.
    \item Since our DLDD has diameter $\tilde{O}(\frac{1}{\eps})$, we have that the above ball-passing to comptue $\mcP$ can be implemented in time at most $\tilde{O}(\frac{1}{\eps})$.
\end{itemize}
Thus, overall our CONGEST algorithm takes time $\tilde{O}(\frac{1}{\eps^5} \cdot (\rho_{CC})^{10})$.
\end{proof}

\section{Deterministic Blocking Integral Flows in $h$-Layer DAGs}\label{sec:detMaxPaths}

In \Cref{sec:randMaxPaths} we showed how to efficiently compute blocking integral flows in $h$-layer DAGs with high probability. In this section, we show how to do so deterministically by making use of the near Eulerian partitions of \Cref{sec:nearEul}. Specifically, we show the following.

\begin{restatable}{lemma}{detMax}\label{lem:detMax}
There is a deterministic algorithm which, given a capacitated $h$-layer $S$-$T$ DAG $D$, computes an integral $S$-$T$ flow that is blocking in:
\begin{enumerate}
    \item Deterministic parallel time $\tilde{O}(h^3)$ with $m$ processors;
    \item Deterministic CONGEST time $\tilde{O}(h^6 \cdot (\rho_{CC})^{10})$.
\end{enumerate}
\end{restatable}

The above parallel algorithm is more or less implied by the work of \citet{cohen1995approximate}. However, the key technical challenge we solve in this section is a distributed implementation of the above. Nonetheless, for the sake of completeness we will include the parallel result as well alongside our distributed implementation.

Our strategy for showing the above lemma has two key ingredients.
    \paragraph{Iterated Path Count Flow.} First, we construct the iterated path count flow. This corresponds to repeatedly taking the expected flow induced by the sampling of our randomized algorithm (as given by \Cref{lem:randCovPaths}). As the flow we compute is the expected flow of the aforementioned sampling, this process is deterministic. The result of this is a $\tilde{\Omega}(\frac{1}{h})$-blocking but not necessarily integral flow. We argue that any such flow is also $\tilde{\Omega}(\frac{1}{h^2})$-approximate and so the iterated path count flow is nearly optimal but fractional.
    \paragraph{Flow Rounding.} Next, we provide a generic way of rounding a fractional flow to be in integral in an $h$-layer DAG while approximately preserving its value. Here, the main challenge is implementing such a rounding in CONGEST; the key idea we use is that of a $(1-\eps)$-near Eulerian partition from \Cref{sec:nearEul} which discards a small number of edges and then partitions the remaining graph into cycles and paths.
    
    These partitions enables us to implement a rounding in the style of \citet{cohen1995approximate}. In particular, we start with the least significant bit of our flow, compute a $(1-\eps)$-near Eulerian partition of the graph induced by all arcs which set this bit to $1$ and then use this partition to round all these bits to $0$. Working our way from least to most significant bit results in an integral flow.  The last major hurdle to this strategy is showing that discarding a small number of edges does not damage our resulting integral flow too much; in particular discarding edges in the above way can increase the $\deficit$ of our flow. However, by always discarding an appropriately small number of edges we show that this $\deficit$ is small and so after deleting all flow that originates or ends at vertices not in $S$ or $T$, we are left with a flow of essentially the same value of the input fraction flow. The end result of this is a rounding procedure which rounds the input fractional flow to an integral flow while preserving the value of the flow up to a constant.\\

Our algorithm to compute blocking integral flows in $h$-layer DAGs deterministically combines the above two tools. Specifically, we repeatedly compute the iterated path count flow, round it to be integral and add the resulting flow to our output. As the iterated path count flow is $\tilde{\Omega}(\frac{1}{h^2})$-approximate, we can only repeat this about $h^2$ time (otherwise we would end up with a flow of value greater than that of the optimal flow).


\subsection{Iterated Path Count Flows}
In this section we define our iterated path count flows and prove that they are $\tilde{\Omega}(\frac{1}{h})$-approximate.

Specifically, the path counts of \Cref{sec:pathCounts} naturally induce a flow. In particular, they induce what we will call the path count flow where the flow on arc $(u,v)$ is defined as:
\begin{align*}
    f_{a} =  \U_a \cdot \frac{n_a}{\max_{a \in A} n_a}.
\end{align*}
It is easy to see these path counts induce an $S$-$T$ flow.
\begin{lemma}\label{lem:pathCountFlowFeasible}
For a given capacitated $S$-$T$ DAG the path count flow is an $S$-$T$ flow.
\end{lemma}
\begin{proof}
The above flow does not violate capacities by construction. Moreover, it obeys flow conservation for all vertices other than those in $S$ and $T$ since it is a convex combination of paths between $S$ and $T$. More formally, for any vertex $v \not \in S \cup T$ we have flow conservation by the calculation:
\begin{align*}
    \sum_{a = (u,v) \in \delta^-(v)} f_a &= 
    \frac{\U_a}{{\max_{a \in A} n_a}}\sum_{a = (u,v) \in \delta^-(v)}  \sum_{P \in \mcP(S,T) : a \in P} \U(P)\\
    &= \frac{\U_a}{\max_{a \in A} n_a}\sum_{a = (u,v) \in \delta^+(v)}  \sum_{P \in \mcP(S,T) : a \in P} \U(P)\\
    &= \sum_{a = (u,v) \in \delta^+(v)} f_a
\end{align*}
where the second line follows from the fact that every path from $S$ to $T$ which enters $v$ must also exit $v$.
\end{proof}

Path count flows were first introduced by \citet{cohen1995approximate}. Our notion of an iterated path count flow is closely related to \citet{cohen1995approximate}'s algorithm for computing blocking flows in parallel. In particular, in order to compute an integral blocking flow, \citet{cohen1995approximate} iteratively computes a path count flow, rounds it, decrements capacities and then iterates. For us it will be more convenient to do something slightly different; namely, we will compute a path count flow, decrement capacities and iterate; once we have \emph{a single} blocking fractional flow we will apply our rounding \emph{once}. Nonetheless, we note that many of the ideas of this section appear implicitly in \citet{cohen1995approximate}.

We proceed to define the iterated path count flow which is always guaranteed to be near-optimal. The iterated path count flow will be a sum of several path count flows. More formally, suppose we are given an $h$-layer capacitated $S$-$T$ DAG $D = (V, A)$ with capacities $\U$. In such a DAG we have $n_a \leq (n \cdot \U_{\max})^h$. We initialize $f_0$ to be the flow that assigns $0$ to every arc and $\U_0 = \U$. We then let $D_i = (V, A)$ with capacities $\U_i$ where $\U_i = \U_{i-1} - f_{i-1}$ and $f_{i-1}$ is the path count flow of $D_{i-1}$. Lastly, we define the iterated path count flow as a convex combination of these path count flows iterated $k = \Theta(h \cdot (\log n )\cdot \log(n \cdot \U_{\max} ))$ times. That is, the iterated path count flow is
\begin{align*}
    \tilde{f} := \sum_{i = 0}^k f_i.
\end{align*}

We begin by observing that the iterated path count flow is reasonably blocking.
\begin{lemma}\label{lem:pathCountFlowBlocking}
The iterated path count flow $\tilde{f}$ is a (not necessarily integral) blocking $S$-$T$ flow.
\end{lemma}
\begin{proof}
Since each path count flow is an $S$-$T$ flow by \Cref{lem:pathCountFlowFeasible}, by how we reduce capacities it immediately follows that $\tilde{f}$ is an $S$-$T$ flow. 

Thus, it remains to argue that $\tilde{f}$ is blocking. Towards this, consider computing the $i$th path count flow when the current path counts are $\{n_a\}_a$ and the flow over arc $a$ is $(f_i)_a = (\U_i)_a \cdot \frac{n_a}{\max_a n_a} $. Letting $A_{\approx \max}$ be all arcs for which $n_a \geq \frac{1}{2}\max_{a}n_a$, we get that $(\U_{i+1})_a \leq \frac{1}{2} \cdot (\U_{i})_a$ for all $a \in A_{\approx \max}$. It follows that after $\Theta(\log n)$ iterations we will reduce $\max_a n_a$ by at least a multiplicative factor of $2$. Since initially $n_a \leq (n \cdot \U_{\max})^h$, it follows that after $k = \Theta(h \cdot (\log n )\cdot \log(n \cdot \U_{\max} ))$ iterations we have reduced $n_a$ to $0$ for every arc which is to say that for any path $P$ between $S$ and $T$ we have that there is some arc $a \in P$ it holds that $\sum_{i} (f_i)_a = \U_a$. Since $\tilde{f}_a = \sum_{i} (f_i)_a$, we conclude that $\tilde{f}$ is blocking.
\end{proof}

Next, we observe that any blocking flow is near-optimal.
\begin{lemma}\label{lem:blockingGivesApx}
Any $\alpha$-blocking $S$-$T$ flow in an $h$-layer $S$-$T$ DAG is $\left(\frac{\alpha}{h}\right)$-approximate.
\end{lemma}
\begin{proof}
Let $f$ be our $\alpha$-blocking flow and let $D$ be the input graph. Let $f^*$ be the optimal $S$-$T$ flow in the input DAG and let $ \sum_{P} f_P$ be it's flow decomposition into path flows where each $P$ is a directed path from $S$ to $T$ and $(f_P)_a$ is $1$ if $a \in P$ and $0$ otherwise. 

Since $f$ is blocking, for each such path there is some arc, $a_P$ where $f_{a_P} \geq \alpha \cdot \U_{a_P} \geq \alpha \cdot f_{a_P}^*$. Let $A' = \{a_{P} : P \text{ in flow decomposition of } f^*\}$ be the union of all such blocked arcs. Thus, $\st(f^*) \leq \sum_{a \in A'} f^*_a \leq \sum_{a \in A'} \frac{f_{a}}{\alpha}$. However, since $D$ is $h$-layered, by an averaging argument we have that there must be some $j$ such that $f(\delta^+(V_j) \cap A') \geq \frac{1}{h} \sum_{a \in A'} f_{a}$ where $V_j$ is the $j$th layer of our digraph. On the other hand, $\st(f) \geq f(\delta^+(V_j)) \geq f(\delta^+(V_j) \cap A')$ and so we conclude that 
\begin{align*}
    \st(f) &\geq f(\delta^+(V_j) \cap A')\\
    & \geq \frac{1}{h} \cdot \sum_{a \in A'} f_{a}\\
    & \geq \frac{\alpha}{h} \cdot \st(f^*),
\end{align*}
showing that $f$ is $\left(\frac{\alpha}{h}\right)$-approximate as desired.
\end{proof}

We conclude that the iterated path count flow is near-optimal and efficiently computable; our CONGEST algorithm will make use of sparse neighborhood covers to deal with potentially large diameter graphs.

\begin{lemma}\label{lem:iterPathCount}
Let $D$ be a capacitated $h$-layer $S$-$T$ DAG with diameter at most $\tilde{O}(h)$. Then one can deterministically compute a (possibly non-integral) flow $\tilde{f}$:
\begin{enumerate}
    \item In parallel that is $\Omega\left(\frac{1}{h} \right)$-approximate in time $\tilde{O}(h^2 )$ with $m$ processors;
    \item In CONGEST that is $\tilde{\Omega}\left(\frac{1}{h} \right)$-approximate in time $\tilde{O}\left(h^4 \right)$.
\end{enumerate}
\end{lemma}
\begin{proof}
Combining \Cref{lem:pathCountFlowBlocking} and \Cref{lem:blockingGivesApx} shows that the iterated path count flow is an $S$-$T$ flow that is $\Omega(\frac{1}{h})$-approximate.

For our parallel algorithm, we simply return the iterated path count flow. The iterated path count flow is simply a sum of $k = \Theta(h \cdot (\log n) \cdot  \log(n \cdot \U_{\max} )$-many path count flows. Thus, it suffices to argue that we can compute path count flows in $O(h)$ parallel time with $m$ processors. By \Cref{lem:computePathCounts} we can compute $n_a$ for every $a$ in these times and so to then compute the corresponding path count flows we need only compute $\max_a n_a$ which is trivial to do in parallel in the stated time.

For our CONGEST algorithm we do something similar but must make use of sparse neighborhood covers, because we cannot outright compute $\max_{a} n_a$ as the diameter of $D$ might be very large.  Specifically, we do the following. Apply \Cref{lem:sparseNeigh} to compute an $s$-sparse $h$-neighborhood cover with diameter $\tilde{O}(h)$ and partition $\mcV_1, \mcV_2, \ldots, \mcV_s$ for $s = \tilde{O}(1)$. Then we iterate through each of these partitions for $i = 1, 2, \ldots, s$. For each part $V_i^{(j)} \in \mcV_i$, we let $\tilde{f}_i^{(j)}$ be the iterated path count flow of $D[V_i^{(j)}]$ with source set $S \cap V_i^{(j)}$ and sink set $T \cap V_i^{(j)}$. We let $\tilde{f}_i := \sum_{j} f_i^{[j]}$ be the path count flows associated with the $i$th partition and return as our solution the average path count flow across partitions; namely we return
\begin{align*}
    \tilde{f} = \frac{1}{s} \cdot \sum_i \tilde{f}_i.
\end{align*}
This flow is an $S$-$T$ flow since it is a convex combination of $S$-$T$ flows. We now argue that this flow is $\tilde{\Omega}(\frac{1}{h})$-optimal. Let $\hat{f}_i^{[j]}$ be the optimal flow on $D[V_i^{(j)}]$ with source set $S \cap V_i^{(j)}$ and sink set $T \cap V_i^{(j)}$. As our path count flows are $\Omega(\frac{1}{h})$-approximate, we know that
\begin{align*}
\st(\tilde{f}_i^{[j]}) \geq \tilde{\Omega}\left(\frac{1}{h}\right) \cdot \st(\hat{f}_i^{[j]}).    
\end{align*}
Moreover, since every $h$-neighborhood is contained in one of the $V_i^{[j]}$, it follows that $\sum_{i, j} \st(\hat{f}_i^{[j]}) \geq \st(f^*)$ where $f^*$ is the optimal $S$-$T$ flow on $D$ with source set $S$ and sink set $T$. Thus, we conclude that
\begin{align*}
    \st(\tilde{f}) &= \frac{1}{s} \cdot \sum_{i,j} \tilde{f}_i^{[j]}\\
    & \geq  \Omega\left(\frac{1}{h}\right) \cdot \frac{1}{s} \cdot \sum_{i,j} \hat{f}_i^{[j]}\\
    & \geq \tilde{\Omega}\left(\frac{1}{h}\right) \cdot \st(f^*).
\end{align*}

Lastly, we argue the running time of our CONGEST algorithm. We describe how to compute $\tilde{f}_i$ for a fixed $i$. Again, $\tilde{f}_i$ on each part is simply a sum of $k = \tilde{\Theta}(h)$-many path count flows. To compute one of these path count flows we first compute the path counts $\{n_a\}_a$ on each part by applying \Cref{lem:computePathCounts} which takes $\tilde{O}(h^2)$ time. Next, we compute $\max_a n_a$ in $\tilde{O}(h)$ time by appealing to \Cref{lem:sparseNeigh} and the fact that $\max_a n_a \leq O(n^h)$. Thus, computing each $\tilde{f}_i$ takes time $\tilde{O}(h^3)$ and since there are $\tilde{O}(h)$ of these, overall this takes $\tilde{O}(h^4)$ time.
\end{proof}

\subsection{Deterministic Rounding of Flows in $h$-Layer DAGs}\label{sec:flowRounding}
In the previous section we showed how to construct our iterated path count flows and that they were near-optimal but possibly fractional. In this section, we give the flow rounding algorithm that we will use to round our iterated path count flows to be integral. Specifically, in this section we show the following flow rounding algorithm.
\begin{restatable}{lemma}{flowRound}
\label{lem:detRounding}
There is a deterministic algorithm which, given a capacitated $h$-layer $S$-$T$ DAG $D$, $\eps = \Omega(\frac{1}{\poly(n)})$ and (possibly fractional) flow $f$, computes an integral $S$-$T$ flow $\hat{f}$ in:
\begin{enumerate}
    \item Parallel time $\tilde{O}(h)$ with $m$ processors;
    \item CONGEST time $\tilde{O}(\frac{1}{\eps^5} \cdot h^5 \cdot (\rho_{CC})^{10})$.
\end{enumerate}
Furthermore, $\st(\hat{f}) \geq (1-\eps) \cdot \st(f)$.
\end{restatable}
Parts of the above parallel result are implied by the work of \citet{cohen1995approximate} while the CONGEST result is entirely new.

\subsubsection{Turning Flows on $(1-\eps)$-Near Eulerian Partitions}

As discussed earlier, our rounding will round our flow from the least to most significant bit. To round the input flow on a particular bit we will consider the graph induced by the arcs which set this bit to $1$. We then compute an oriented near-Eulerian partition of these edges and ``turn'' flow along each cycle and path consistently with its orientation. We will always turn flow so as to not increase the deficit of our flow.

We now formalize how we use our $(1-\eps)$-near Eulerian partitions to update our flow. Given a path or cycle $H$, our flow update will carefully choose a subset of arcs of $H$ along which to increase flow (denoted $H^+$) and decrease flow along all other arcs of $H$. Specifically, let $H$ be an oriented cycle or path of a graph produced by forgetting about the directions in a digraph $D  = (V,A)$. Then $H^+$ is illustrated in \Cref{fig:HPlus} and defined as follows:
\begin{itemize}
    \item Suppose $H$ is an oriented cycle. Then, we let $H^+$ be all arcs of $D$ in this cycle that point in the same direction as their orientation. 
    \item Suppose $H = (s_H = v_0, v_1, v_2, \ldots)$ is an oriented path. We let $H^+$ be all arcs of $D$ in this path that point in the same direction as the one arc in $D$ incident to $s_H$ (i.e.\ the designated source of the path). That is either $(v_{0}, v_1)$ or $(v_{1}, v_0)$ are in $D$. In the former case we let $H^+$ be all arcs in $D$ of the form $(v_{i}, v_{i+1})$ for some $i$. In the latter case we let $H^+$ be all arcs in $D$ of the form $(v_{i+1}, v_{i})$ for some $i$.
\end{itemize}
\begin{figure}
    \centering
    \begin{subfigure}[b]{0.3\textwidth}
        \centering
        \includegraphics[width=\textwidth,trim=0mm 0mm 0mm 0mm, clip]{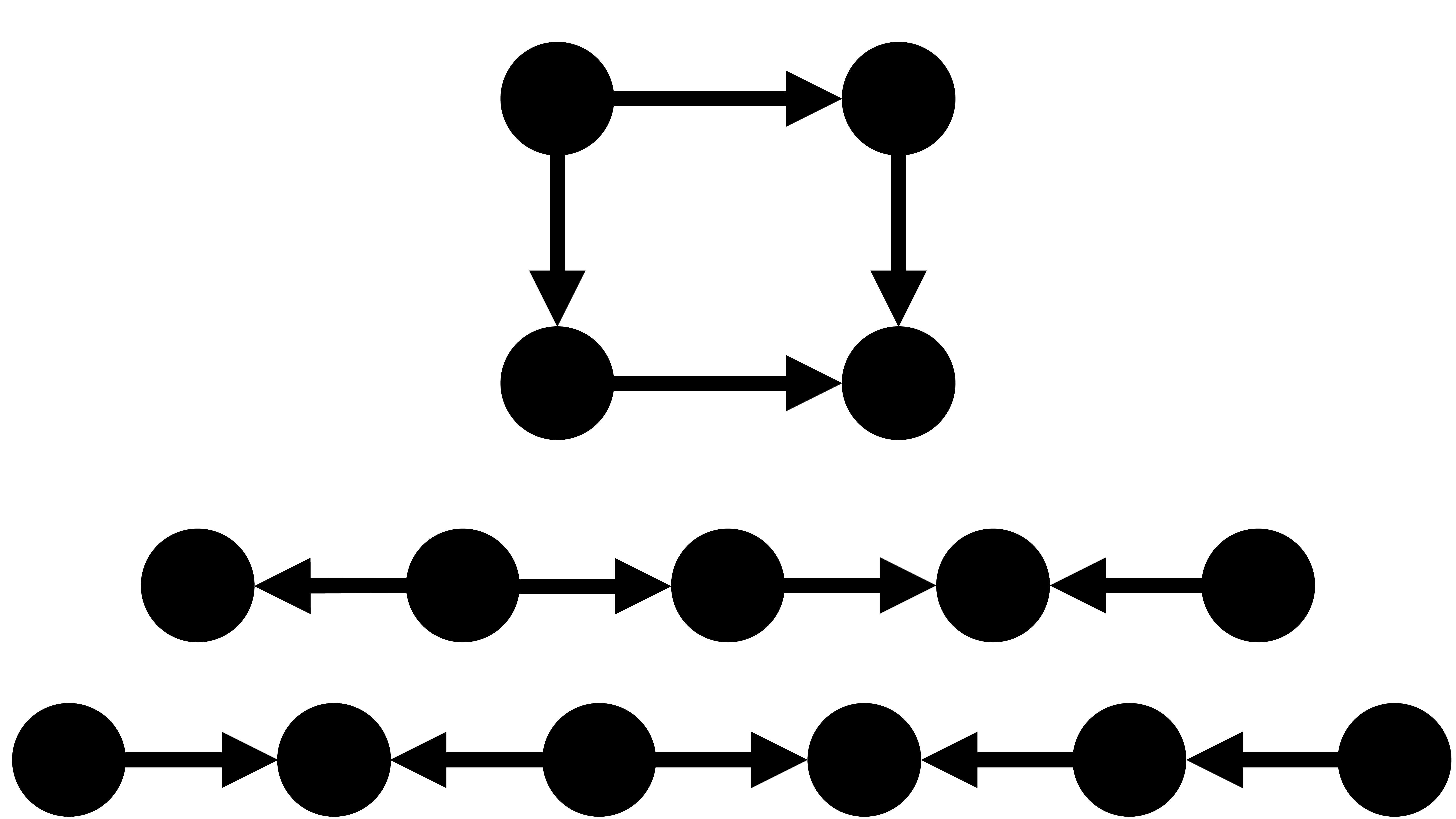}
        \caption{Near Eulerian Partition $\mathcal{H}$.}\label{sfig:HPlus1}
    \end{subfigure}    \hfill
    \begin{subfigure}[b]{0.3\textwidth}
        \centering
        \includegraphics[width=\textwidth,trim=0mm 0mm 0mm 0mm, clip]{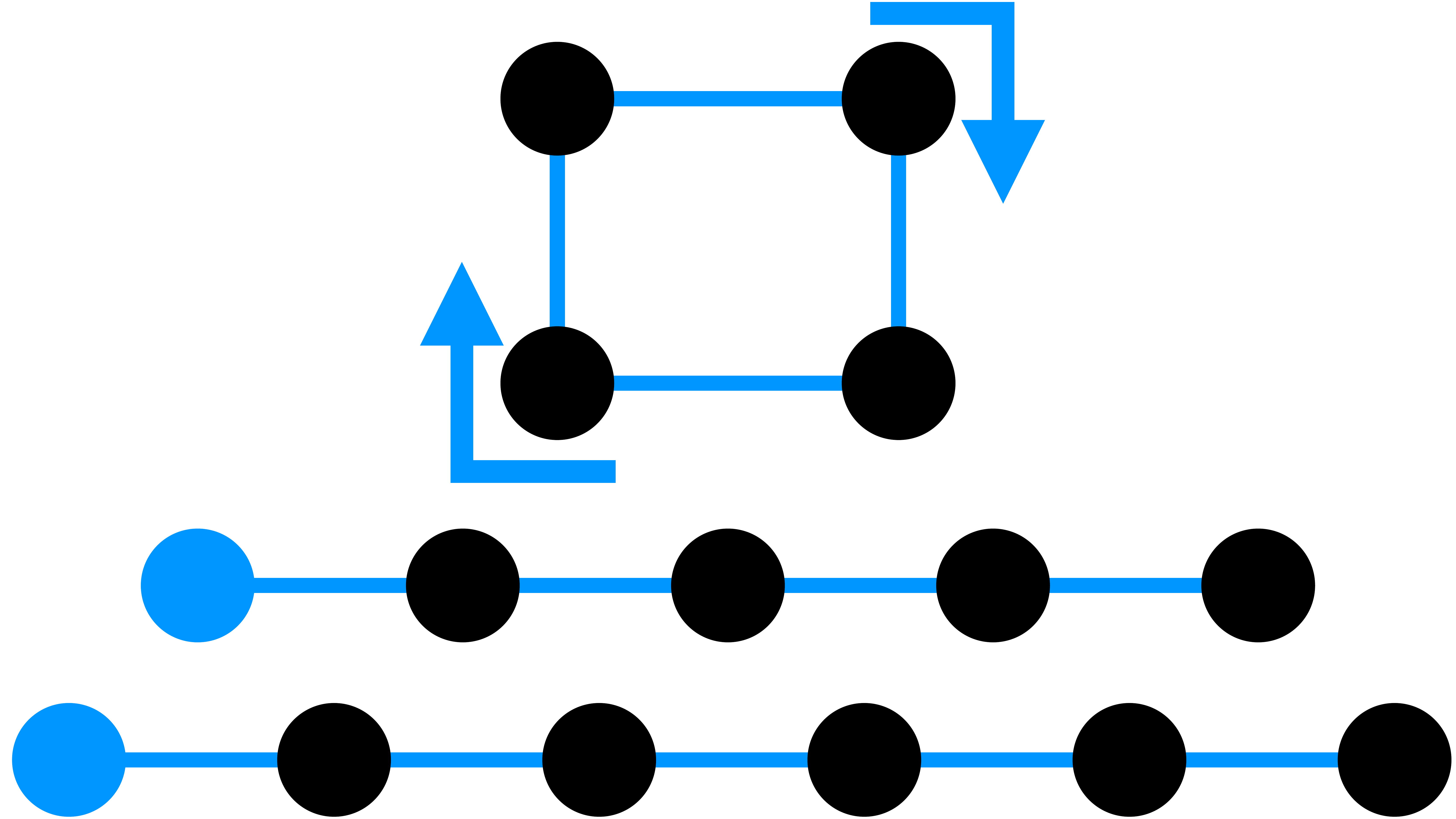}
        \caption{Orientation of $\mathcal{H}$.}\label{sfig:HPlus2}
    \end{subfigure}    \hfill
    \begin{subfigure}[b]{0.3\textwidth}
        \centering
        \includegraphics[width=\textwidth,trim=0mm 0mm 0mm 0mm, clip]{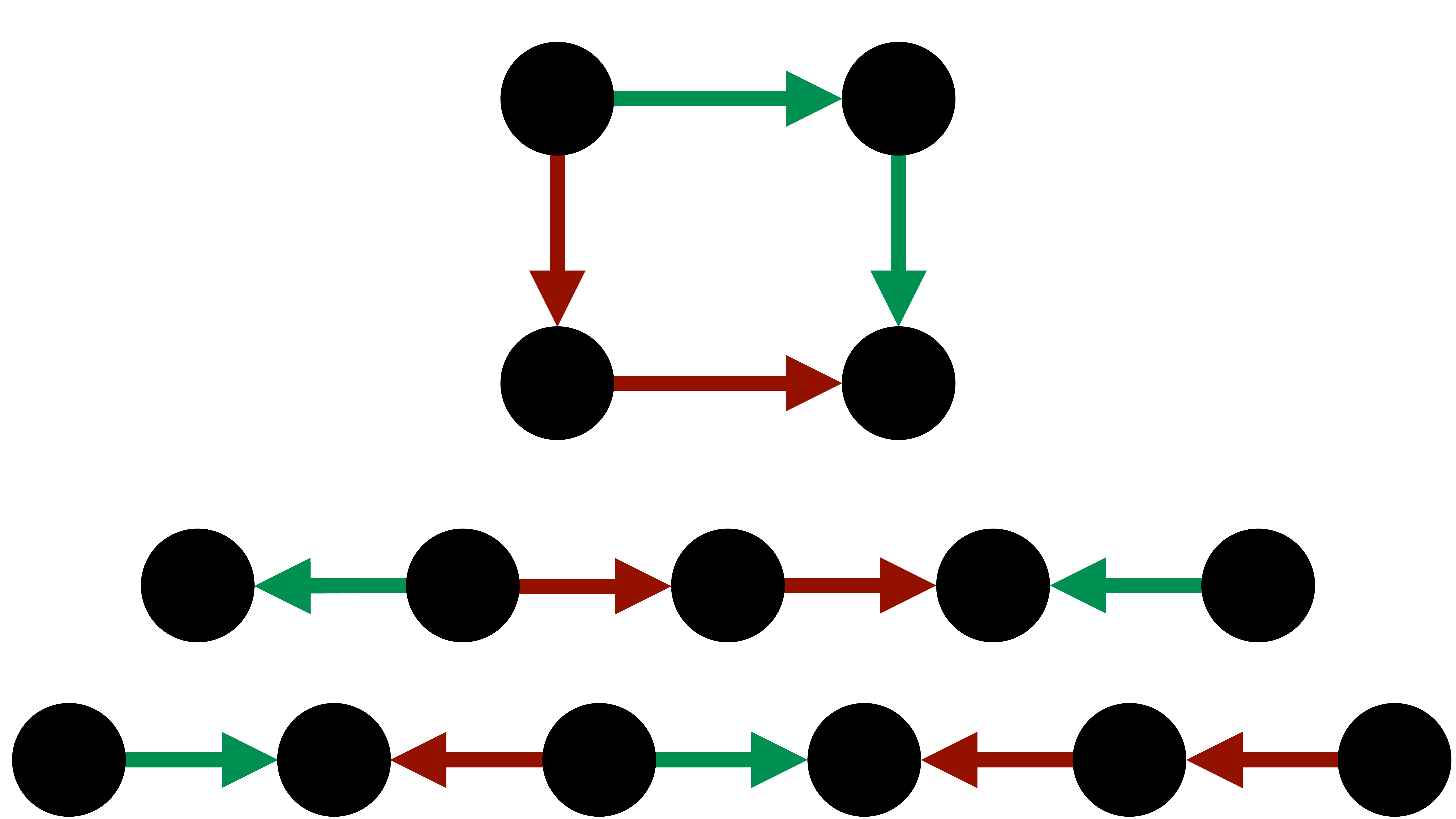}
        \caption{$H^+$ for each $H \in \mathcal{H}$.}\label{sfig:HPlus3}
    \end{subfigure}
    \caption{An illustration of a near Eulerian partition $\mcH$ and $H^+$ for each $H \in \mcH$. \ref{sfig:HPlus1} gives $\mcH$ which consists of one cycle and two paths. \ref{sfig:HPlus2} gives the orientation of $\mcH$ where the source of each path is in blue. \ref{sfig:HPlus3} gives $H^+$ (in green) and $H \setminus H^+$ (in red) for each $H \in \mcH$.}\label{fig:HPlus}
\end{figure}

With our definition of $H^+$ in hand, we now define our flow updates as follows.
\begin{definition}[$(1-\eps)$-Near Eulerian Partition Flow Update]\label{dfn:flowUpdate}
Let $f$ be a flow in a capacitated DAG for which $f_a \in \{0, c\}$ for every $a \in A$ for some $c$ and let $\mcH$ be an oriented $(1-\eps)$-near Eulerian partition of $\supp(f)$ after forgetting about edge directions. Then if $H  \in \mcH$, we define the flow $f_{H}$ on arc $a$:
\begin{align*}
(f_H)_a := 
\begin{cases}
2c & \text{if $a \in H^+$}\\
0 & \text{otherwise}
\end{cases}
\end{align*}

Likewise, we define the flow corresponding to $(f, \mcH)$ as
\begin{align*}
    f_{\mcH} := \sum_{H \in \mcH}f_H.
\end{align*}
\end{definition}

The following shows that our flow update will indeed zero out the value of each bit on each edge while incurring a negligible deficit.
\begin{lemma}\label{lem:flowRoundLem}
Let $f$ be a flow in a capacitated DAG $D$ with specified source and sink vertices $S$ and $T$ where $f_a \in \{0, c\}$ for every $a \in A$ for some $c$. Let $\mcH$ be an oriented $(1-\eps)$-near Eulerian partition of $\supp(f)$ after forgetting about edge directions. Then $f_{\mcH}$ (as defined in \Cref{dfn:flowUpdate}) satisfies:
\begin{enumerate}
    \item $(f_{\mcH})_a \in \{0, 2c\}$ for every $a \in A$;
    \item $\deficit(f_{\mcH}) \leq \deficit(f) + 2\epsilon \cdot \sum_{a}f_a$.
\end{enumerate}
\end{lemma}
\begin{proof}
$(f_{\mcH})_a \in \{0, 2c\}$ holds by the definition of $f_{\mcH}$ and the fact that the elements of $\mcH$ are edge-disjoint.

We next argue that $\deficit(f') \leq \deficit(f) + 2\epsilon \cdot \sum_{a}f_a$. The basic idea is that each edge in the support of $f$ which does not appear in $A[\mcH]$ contributes its value to the deficit but any way of turning a cycle in $\mcH$ leaves the deficit unchanged and the way we chose to turn paths also leaves the deficit unchanged.

We let $f'$ be $f$ projected onto the arcs in $A[\mcH]$. That is, on arc $a$ the flow $f'$ takes value
\begin{align*}
    f'_a =:
    \begin{cases}
    f_a & \text{if $a \in A[\mcH]$}\\
    0 & \text{otherwise}
    \end{cases}
\end{align*}

We have that $\deficit(f') \leq \deficit(f) + 2\epsilon \cdot \sum_{a}f_a$ since each arc $a \not \in A[\mcH]$ increases the deficit of $f'$ by at most $2f_a$ and, from \Cref{def:Euler part}, there are at most $\eps$-fraction of arcs not in $A[\mcH]$. Thus, to show our claim it suffices to argue that $\deficit(f_{\mcH}) \leq \deficit(f')$. For a given vertex $v$, we let $n_i(v)$ be the number of elements of $\mcH$ in which $v$ has in-degree $2$. Similarly, we let $n_o(v)$ be the number of elements of $\mcH$ for which $v$ has out-degree $2$. Lastly, we let $s(v)$ be the indicator of whether $v$ is the source of some path in $\mcH$ and $t(v)$ be the indicator of whether $v$ is the sink of a path in $\mcH$. Thus, we have
\begin{align*}
    \deficit(f', v) = 2c \cdot  |n_i(v) - n_o(v)| + c \cdot (s(v) + t(v))
\end{align*}
and so
\begin{align*}
    \deficit(f') & = \sum_v 2c \cdot |n_i(v) - n_o(v)| + c \cdot (s(v) + t(v))\\
    &= 2c|\mcP| + \sum_v 2c \cdot |n_i(v) - n_o(v)|
\end{align*}

On the other hand, we have
\begin{align*}
    \deficit(f_{\mcH}, v) &\leq 2c \cdot |n_i - n_o| + 2c \cdot t(v)
\end{align*}
and so
\begin{align*}
    \deficit(f_{\mcH}) & \leq \sum_v 2c \cdot |n_i(v) - n_o(v)| + 2c \cdot t(v)\\
    &= 2c|\mcP| + \sum_v 2c \cdot |n_i(v) - n_o(v)|
\end{align*}
showing $\deficit(f_{\mcH}) \leq \deficit(f')$ as required.
\end{proof}

\subsubsection{Extracting Integral $S$-$T$ Subflows}

The last piece of our rounding deals with how to fix the damage that the accumulating deficit incurs. Specifically, as we round each bit we discard some edges, increasing our deficit. This means that after rounding all bits we are left with some (small) deficit. In this section we show how to delete flows that originate or end at vertices not in $S$ or $T$, thereby reducing the value of our flow by the deficit but guaranteeing that we are left with a legitimate $S$-$T$ flow.

\begin{lemma}\label{lem:maxSubflow}
Let $\hat{f}$ be an integral (not necessarily $S$-$T$) flow on an $h$-layer $S$-$T$ DAG. Then one can compute an $S$-$T$ integral flow $f'$ which is a subflow of $\hat{f}$ and satisfies $\st(f') \geq \st(\hat{f}) - \deficit(\hat{f})$ in:
\begin{enumerate}
    \item Parallel time $O(h)$ with $m$ processors;
    \item CONGEST time $\tilde{O}(h)$.
\end{enumerate}
\end{lemma}
\begin{proof}
Our algorithm will simply delete out flow that originates not in $S$ or ends at vertices not in $T$. More formally, we do the following. We initialize our flow $f'$ to $\hat{f}$. Let $S = V_1, V_2, \ldots, V_{h+1} = T$ be the vertices in each layer of our input $S$-$T$ DAG $D = (V,A)$. Recall that we defined a flow $\hat{f}$ as an arbitrary function on the arcs so that $\hat{f}_a \leq \U_a$ for every $a$. The basic idea of our algorithm is to first push all ``positive'' deficit from left to right and then to push all ``negative'' deficit from right to left. The deficit will be non-increasing under both of these processes.

More formally, we push positive deficit as follows. For $i = 2, 3, \ldots h$ we do the following. For each $v \in V_i$, let $$\deficit^+(v) := \max\left(0, \sum_{a \in \delta^+(v)} f'_a - \sum_{a \in \delta^-(v)}f'_a \right)$$ be the positive deficit of $v$. Then, we reduce $\sum_{a \in \delta^+(v)} f'_a$ to be equal to $\sum_{a \in \delta^-(v)} f'_a$ by arbitrarily (integrally) reducing $f'_a$ for some subset of $a\in \delta^+(v)$. 

It is easy to see by induction that at this point we have $\deficit^+(v) = 0$ for all $v \not \in S \cup T$. Likewise, we have that $\sum_{v \not \in S \cup T} \deficit^+(v)$ is non-increasing each time we iterate the above. Thus, if $\deficit^+$ is the initial value of $\sum_{v \not \in S \cup T} \deficit^+(v)$ then in the last iteration of the above we may decrease the flow into $T$ by at most $\deficit(\hat{f})$.

Next, we do the same thing symmetrically to reduce the negative deficits. For $i = h, h-1, \ldots, 2$ we do the following for each $v \in V_i$. Let $$\deficit^-(v) := \max\left(0, \sum_{a \in \delta^-(v)} f'_a - \sum_{a \in \delta^+(v)}f'_a \right)$$ be the negative deficit of $v$. Then, we reduce $\sum_{a \in \delta^-(v)} f'_a$ to be equal to $\sum_{a \in \delta^+(v)} f'_a$ by arbitrarily (integrally) reducing $f'_a$ for some subset of $a\in \delta^-(v)$. Notice that this does not increase $\deficit^+(v)$ for any $v \not \in S \cup T$.

Symmetrically to the positive deficit case, it is easy to see that at the end of this process we have reduced $\deficit^-(v)$ to $0$ for every $v \not \in S \cup T$ while reducing the flow out of $S$ by at most $\deficit(\hat{f})$. 

Thus, at the end of this process we have an $S$-$T$ integral flow $f'$ whose value is at least $\st(\hat{f})-\deficit(\hat{f})$. Implementing the above in the stated running times is trivial; the only caveat is that updating a flow in CONGEST requires updating it for both endpoints but since the flow is integral and we reduce it integrally, this can be done along a single arc in time $O(\log \U_{\max}) = \tilde{O}(1)$ by assumption.
\end{proof}

\subsubsection{Flow Rounding Algorithm}
Having defined the flow update we use for each $(1-\eps)$-near Eulerian partition and how to extract a legitimate $S$-$T$ flow from the resulting rounding, we conclude with our algorithm for rounding flows from least to most significant bit. Our algorithm is given in \Cref{alg:detFlow} and illustrated in \Cref{fig:flowRound}.

\begin{figure}
    \centering
    \begin{subfigure}[b]{0.3\textwidth}
        \centering
        \includegraphics[width=\textwidth,trim=0mm 0mm 255mm 0mm, clip]{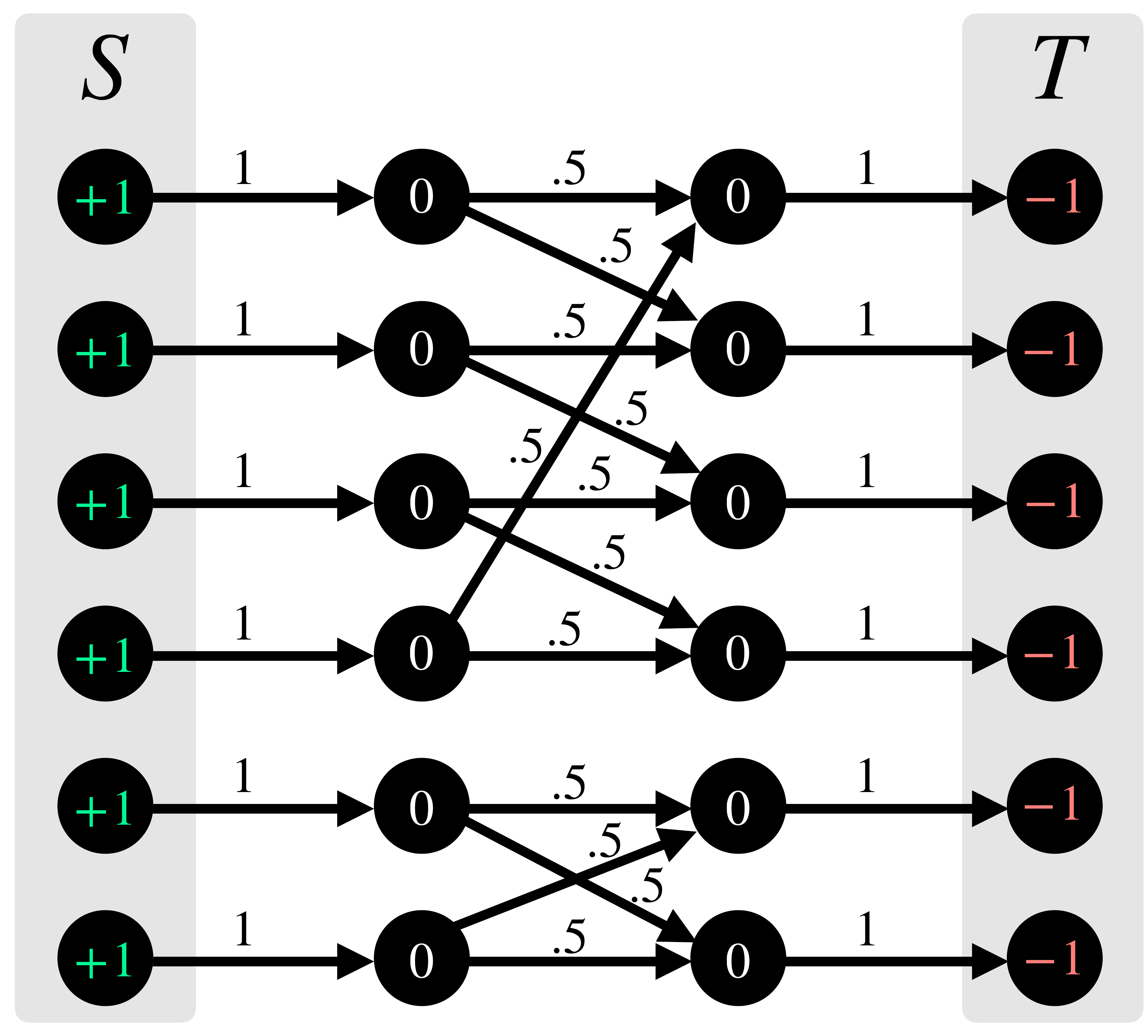}
        \caption{Fractional flow on $D$.}\label{sfig:flowRound1}
    \end{subfigure}    \hfill
    \begin{subfigure}[b]{0.3\textwidth}
        \centering
        \includegraphics[width=\textwidth,trim=0mm 0mm 255mm 0mm, clip]{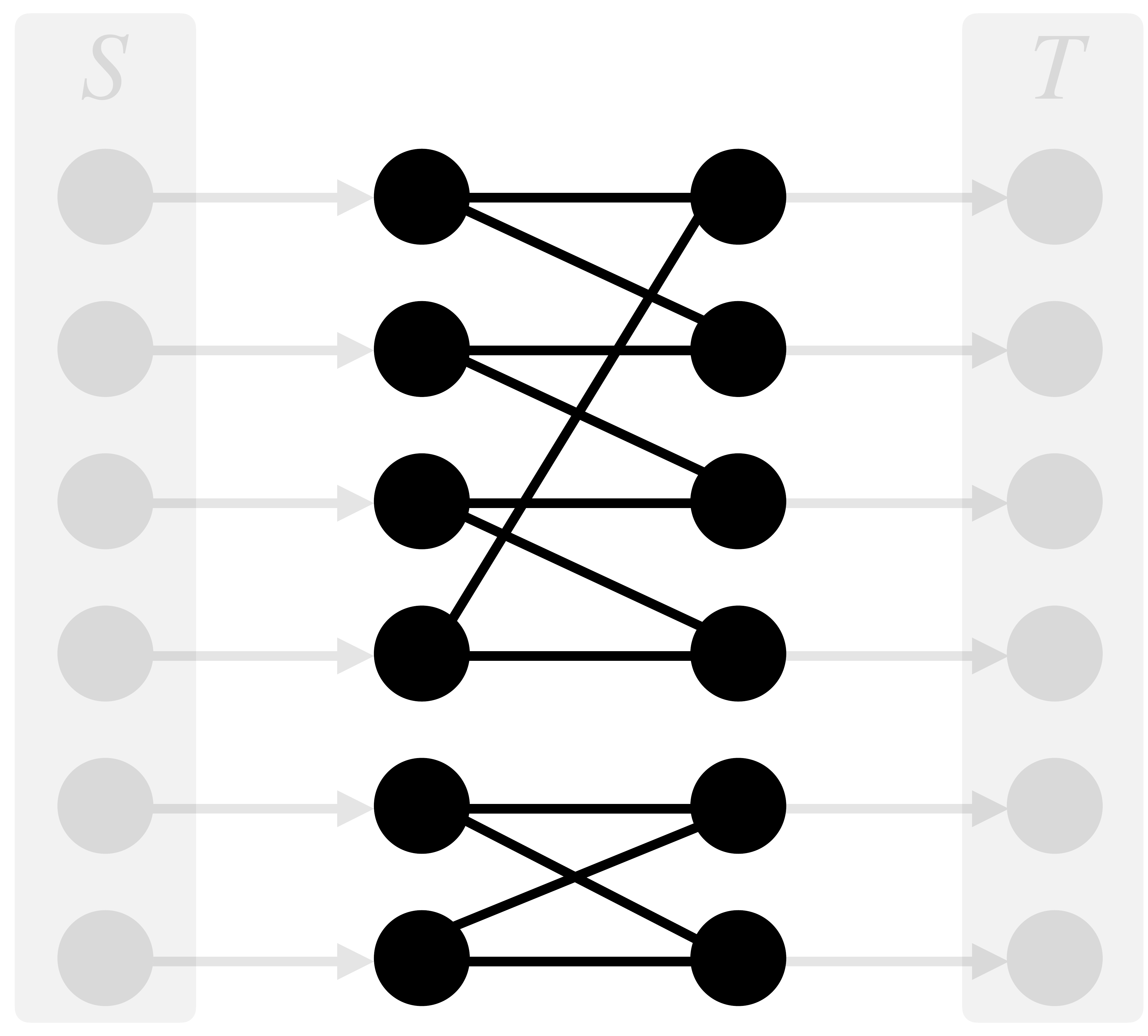}
        \caption{$D^{(2)}$.}\label{sfig:flowRound2}
    \end{subfigure}    \hfill
    \begin{subfigure}[b]{0.3\textwidth}
        \centering
        \includegraphics[width=\textwidth,trim=0mm 0mm 255mm 0mm, clip]{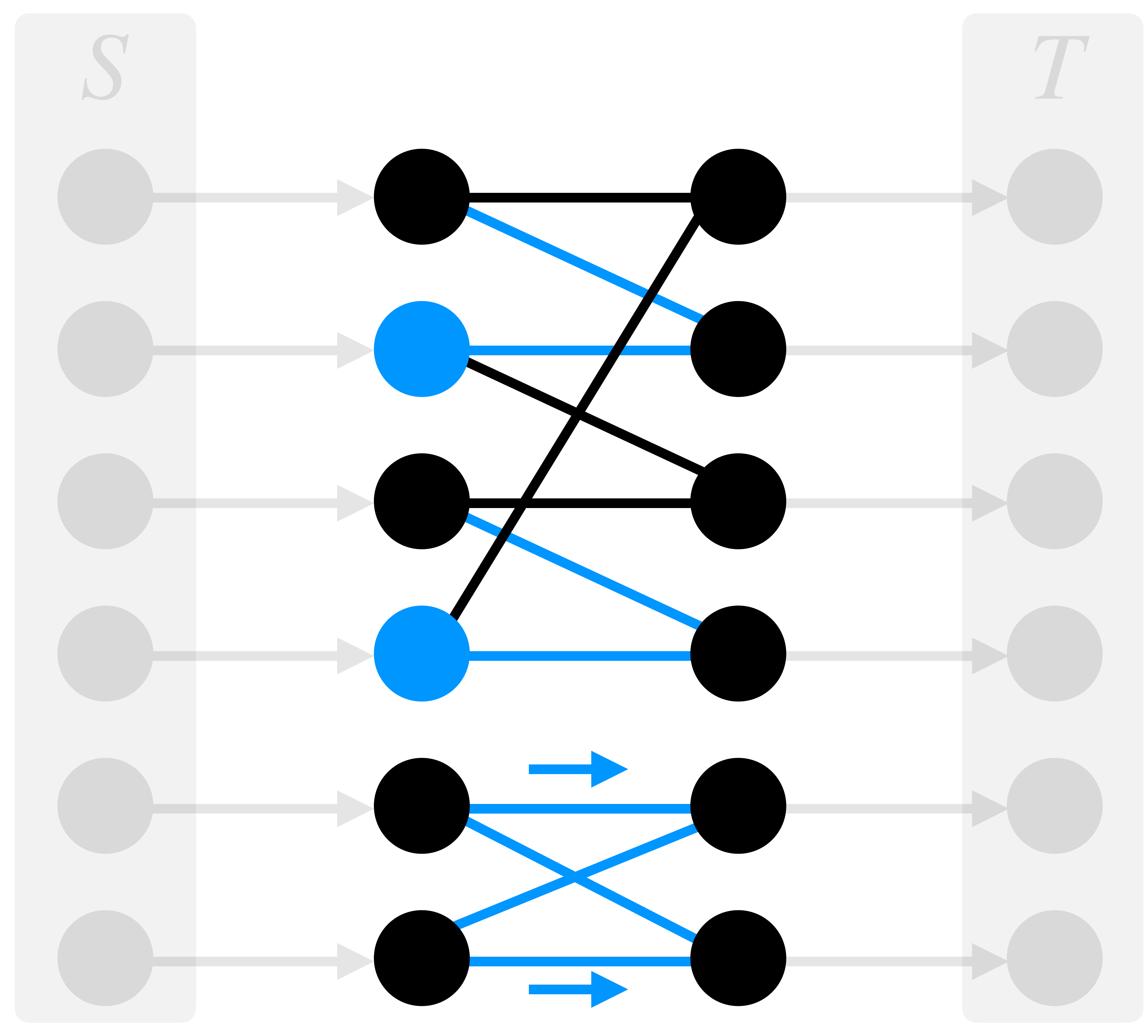}
        \caption{Near Eulerian partition $\mcH$.}\label{sfig:flowRound3}
    \end{subfigure}
    \vspace{5mm}\hfill\\
    \begin{subfigure}[b]{0.3\textwidth}
        \centering
        \includegraphics[width=\textwidth,trim=0mm 0mm 255mm 0mm, clip]{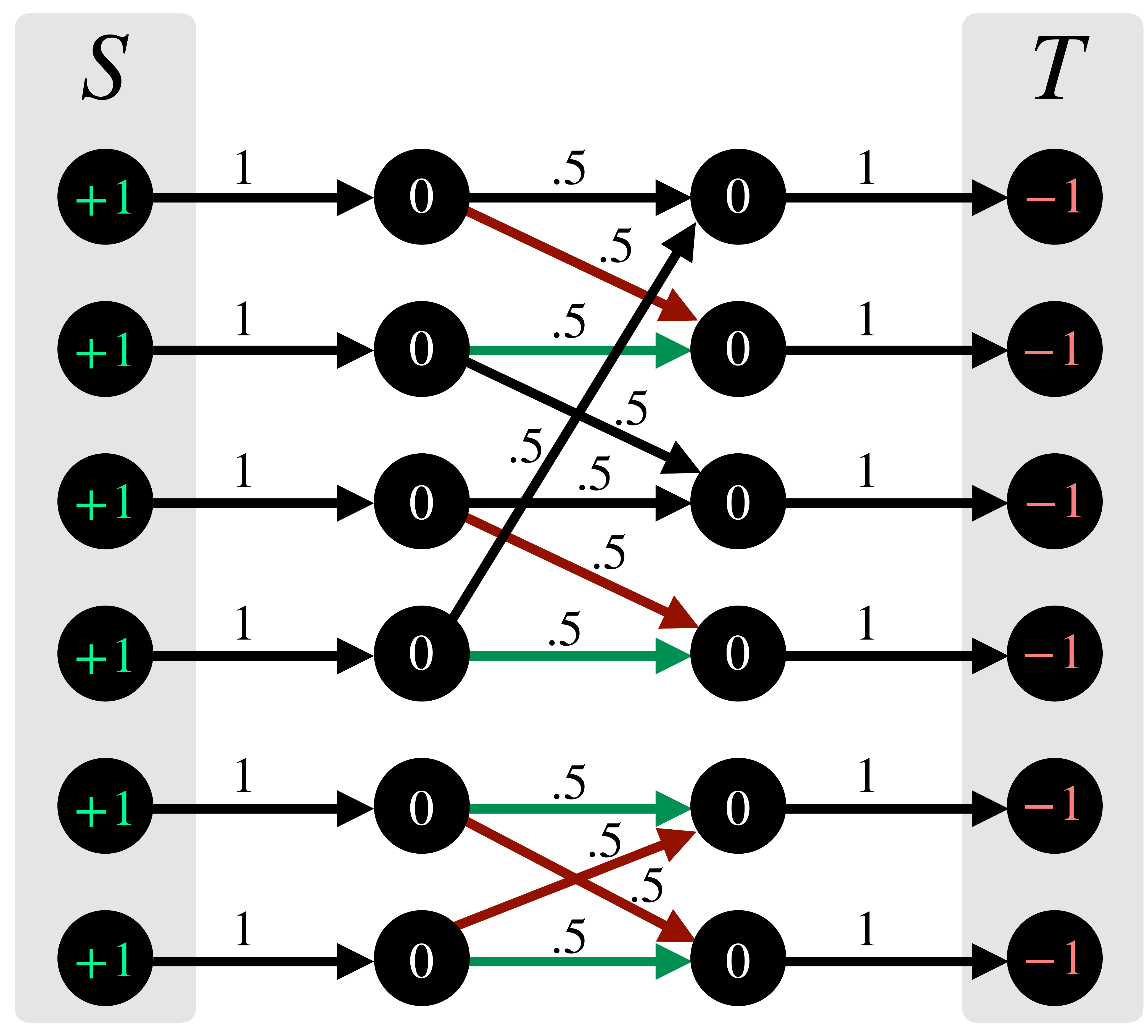}
        \caption{Flow update using $\mcH$.}\label{sfig:flowRound4}
    \end{subfigure} \hfill
    \begin{subfigure}[b]{0.3\textwidth}
        \centering
        \includegraphics[width=\textwidth,trim=0mm 0mm 255mm 0mm, clip]{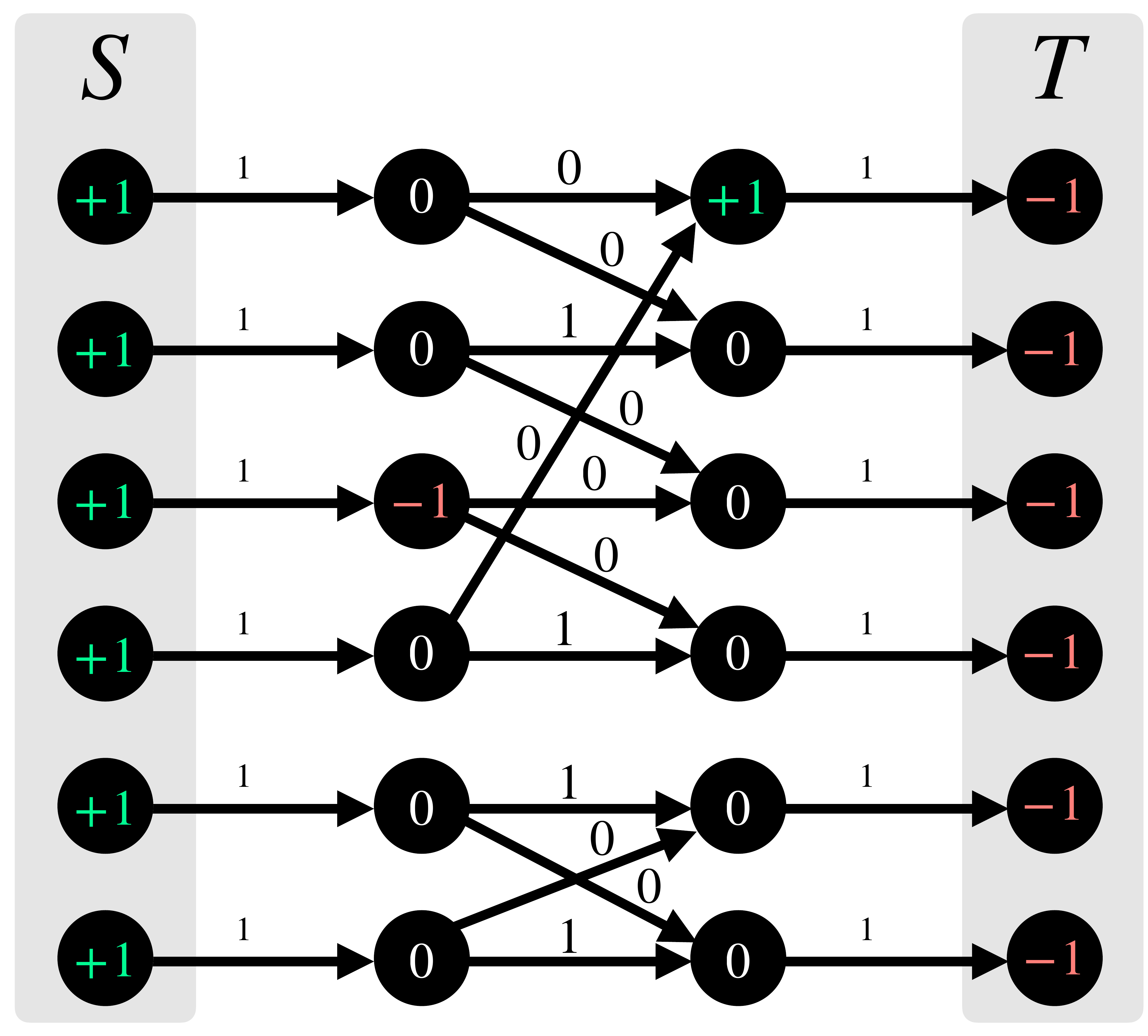}
        \caption{Flow after flow update.}\label{sfig:flowRound5}
    \end{subfigure} \hfill
    \begin{subfigure}[b]{0.3\textwidth}
        \centering
        \includegraphics[width=\textwidth,trim=0mm 0mm 255mm 0mm, clip]{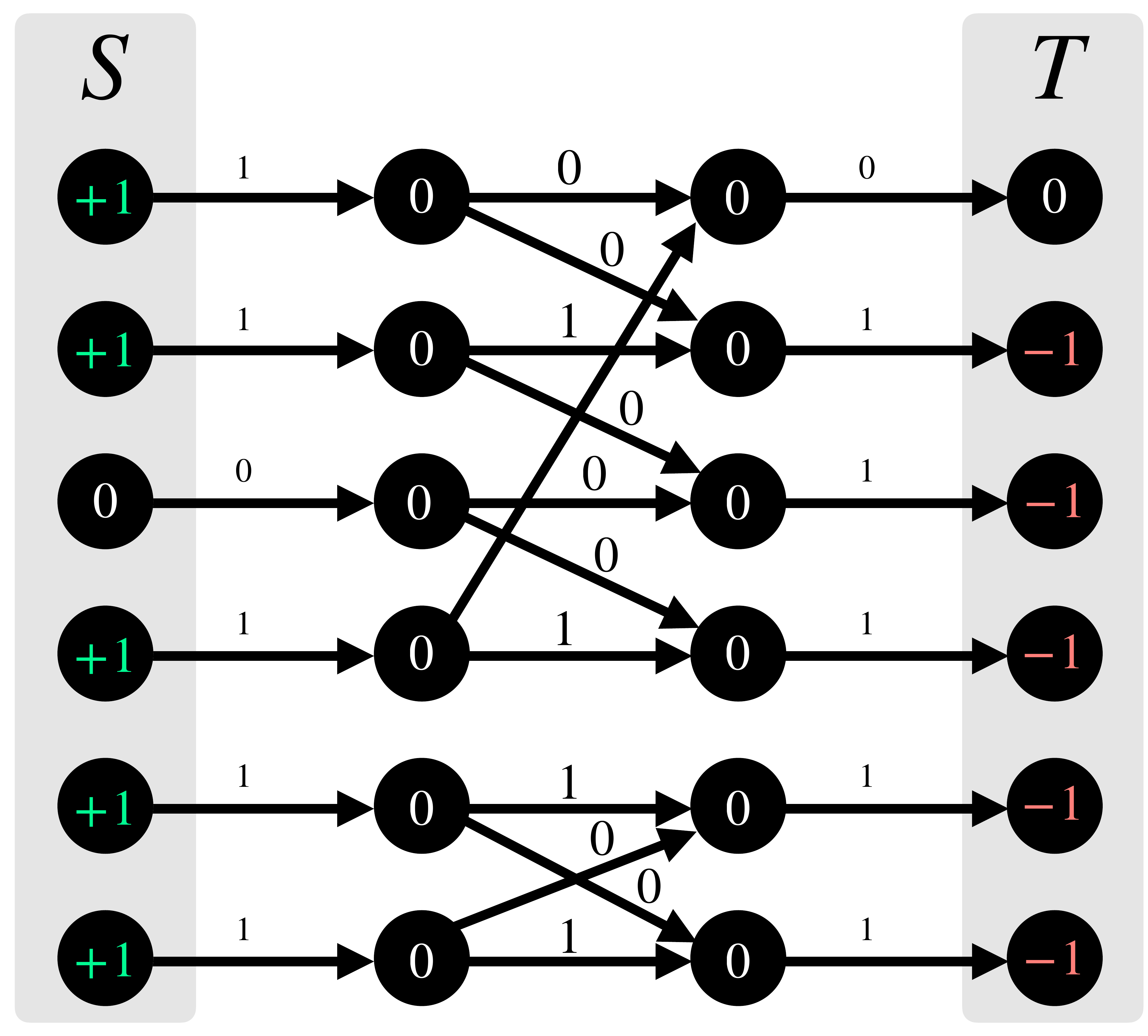}
        \caption{Returned integral $S$-$T$ flow.}\label{sfig:flowRound6}
    \end{subfigure}
    \caption{An example of our flow rounding algorithm on digraph $D$ with unit capacities. \ref{sfig:flowRound1} gives the input flow where arcs are labelled with their flow and vertices are labelled with their deficit. \ref{sfig:flowRound2} gives $D^{(2)}$, the graph induced by all arcs with flow value $.5$. \ref{sfig:flowRound3} gives our oriented near Eulerian partition of $D^{(2)}$ (in blue). \ref{sfig:flowRound4} shows how we update our flow based on the near Eulerian partition. \ref{sfig:flowRound5} gives the result of this flow update; notice that some vertices not in $S$ and $T$ have non-zero deficit. \ref{sfig:flowRound6} gives the $S$-$T$ subflow we return where only vertices in $S$ and $T$ have non-zero deficit.}\label{fig:flowRound}
\end{figure}

\begin{algorithm}
    \caption{Deterministic Flow Rounding}
    \label{alg:detFlow}
    \begin{algorithmic}[0] 
            \State \textbf{Input:} $h$-layer DAG $D$, $S$-$T$ flow $f = \sum_{i=0} f^{(i)}$ where $(f^{(i)})_a \in \{0, 2^{\log(\U_{\max}) - i}\}$ for every $a$, $i$.
            \State \textbf{Output:} integral $S$-$T$ flow $\hat{f}$.
            \State $\hat{f} \gets \sum_{i=0}^{k}f^{(i)}$ for $k = \Theta(\log n + \log (\U_{\max}))$.\Comment{Truncate lower order bits of input flow}
            \For{$i = k, \ldots, \log(\U_{\max})$}
                \State Let $\hat{f} = \sum_j \hat{f}^{(j)}$ be the bitwise flow decomposition of $\hat{f}$ (defined in \Cref{sec:notation}) and let $D^{(i)}$ be \\\qquad\qquad the undirected graph induced by the support of $f^{(i)}$.
                \State Compute an oriented $(1-\eps')$-near Eulerian partition $\mcH$ of $D^{(i)}$ (using \Cref{lem:nearEuler} with $\epsilon' = 0$ \\\qquad\qquad for the parallel algorithm and $\epsilon' = \Theta\left(\frac{\eps}{h \cdot \log n}\right)$ for the CONGEST algorithm).
                \State $\hat{f} \gets \hat{f}_\mcH^{(i)} + \sum_{j < i}  \hat{f}^{(j)}$ (as defined in \Cref{dfn:flowUpdate}). \Comment{Turn flow along $\mcH$}
            \EndFor
            \State Let $\hat{f}$ be an $S$-$T$ subflow of $\hat{f}$ (compute using \Cref{lem:maxSubflow}).
            \State \Return $\hat{f}$.
    \end{algorithmic}
\end{algorithm}


We conclude that the above rounding algorithm rounds with negligible loss in the value.
\flowRound*
\begin{proof}
We use \Cref{alg:detFlow}.

We first argue that the above algorithm returns an integral flow. Notice that by the fact that we initialize $\hat{f}$ to $\sum_{i=0}^k$ it follows that for $j > k$ on every $a$ we have $\hat{f}^{(j)}_a = 0$ just before the first iteration of our algorithm. Thus, to argue that the returned flow is integral it suffices to argue that if $\hat{f}^{(j)}$ is the $j$th bit flow of $\hat{f}$ just after the $i$th iteration then for $j \leq i$ we have $f^{(j)}_a = 0$ for every $a$. However, notice that, by \Cref{lem:flowRoundLem}, after we update $\hat{f}$ each $\hat{f}^{(i)}_a$ value is either doubled or set to $0$, meaning that $\hat{f}^{(i)}_a = 0$ after this update.

Next, we argue that $\st(\hat{f}) \geq (1-\eps) \cdot \st(f)$. By \Cref{lem:maxSubflow} it suffices to argue that just before we compute our  $S$-$T$ subflow of $\hat{f}$ we have $\deficit(\hat{f}) \leq \eps \cdot \st(f)$. We may set the constant in $k = \Theta(\log n + \log(\U_{\max}))$ to be appropriately large so that when we initialize $\hat{f}$ we reduce the flow value on each arc by at most $\frac{1}{\poly(n)}$. It follows that at this point $\deficit(\hat{f}) \leq \frac{2}{\poly(n)} \sum_a f_a$. Similarly, by \Cref{lem:flowRoundLem} in the $i$th iteration of our algorithm we increase the deficit of $\hat{f}$ by at most $2\epsilon' \sum_a \hat{f}^{(i)}_a \leq 2\epsilon' \sum_a f_a$.

For our parallel algorithm, since we have $\epsilon' = 0$, it immediately then follows that $\deficit(\hat{f}) \leq \frac{2}{\poly(n)} \sum_a f_a \leq \eps \cdot \st(f)$ by our assumption that $\eps = \Omega(\frac{1}{\poly(n)})$. For our CONGEST algorithm we choose $\eps' = \Theta(\frac{\eps}{h \log n})$ for some appropriately small constant. Since we have $\Theta(\log n)$ iterations it follows that after all of our iterations (but before we compute an $S$-$T$ subflow) it holds that $\deficit(\hat{f}) \leq \frac{\eps}{h} \cdot \sum_a f_a \leq \eps \cdot \st(f)$ where the last inequality follows from the fact that our flow is $h$-length.


Lastly, we argue that the algorithm achieves the stated running times. The above algorithm runs for $k = \Theta(\log n)$ iterations. The computation in each iteration is dominated by computing a $(1-\eps')$-near Eulerian partition. For our parallel algorithm, computing each $(1-\eps')$-near Eulerian partition takes time at most $\tilde{O}(1)$ with $m$ processors by \Cref{lem:nearEuler}. For our CONGEST algorithm computing each $(1-\eps')$-near Eulerian partition takes time at most $\tilde{O}(\frac{1}{\eps^5} \cdot h^5 \cdot (\rho_{CC})^{10})$ by \Cref{lem:nearEuler}. Lastly, we must compute an $S$-$T$ subflow of $\hat{f}$ which by \Cref{lem:maxSubflow} takes $O(h)$ parallel time with $m$ processors or $\tilde{O}(h)$ CONGEST time.
\end{proof}

\subsection{Deterministic Blocking Integral Flows}
Having shown that the iterated path count flow is near-optimal and fractional but that we can efficiently round fractional flows to be integral, we conclude with our algorithm to compute a blocking integral flow by repeatedly rounding iterated path count flows.

\detMax*
\begin{proof}
We repeatedly compute the iterated path count flow, round it to be integral, reduce capacities appropriately and repeat. We will return flow $f$ initialized to $0$ on all arcs.

Specifically, we repeat the following $\tilde{\Theta}(h)$ times. Apply \Cref{lem:iterPathCount} to compute a $\tilde{\Omega}(1/h)$-approximate (possibly fractional) flow $\tilde{f}$. Next, apply \Cref{lem:detRounding} with $\eps = .5$ to round this to an integral flow $\hat{f}$ where $\st(\hat{f}) \geq \frac{1}{2}\st(\tilde{f})$. Next, we update $f$ to $f + \hat{f}$ and for each arc $a$ we reduce $\U_a$ by $\hat{f}_a$. 

After each time we iterate the above $\tilde{\Theta}(h)$ times we must reduce the value of the optimal solution by at least a multiplicative $\frac{1}{2}$ since otherwise $f$ would be a flow with value greater than the max $S$-$T$ flow in the graph at the beginning of these iterations. Since the optimal solution is at most $m \cdot \U_{\max}$, it follows that we need only iterate the above $\tilde{\Theta}(h)$ times until the value of the optimal $S$-$T$ flow is $0$ which is to say that $f$ is a blocking flow.

By \Cref{lem:iterPathCount} and \Cref{lem:detRounding} each of the above iterations takes parallel time $\tilde{O}(h^2)$ with $m$ processors and CONGEST time $\tilde{O}(h^5 \cdot (\rho_{CC})^{10})$, giving the stated running times.
\end{proof}

\section{Sparse Decompositions of Acyclic Flows}\label{sec:sparseDecomp}

In this section we show that any flow $\hat{f}$ whose support induces an $h$-layer DAG can be decomposed sparsely into an $h$-length flow $f$. In particular, it can be decomposed into an $h$-length flow that sends flow along at most $m$ paths. We will use this result to sparsify our lightest path blockers and, by extension, the flows that we compute.

The basic idea of our algorithm is as follows. Given an $h$-layer DAG with sources $S$ and sinks $T$, we will sweep from $S$ to $T$ layer-by-layer and greedily build the support of $f$. Specifically, while considering a vertex $v$ in a particular layer we will have inductively constructed some number of paths from $S$ to $v$ each with some associated $f$ flow value. What we would like to do is simply forward the flow of each of these paths along arcs in $\delta^+(v)$. However, it may not be possible to do this because the amount of flow supported on each of these arcs by $\hat{f}$ might be much smaller than the $f$ flow of any one of the paths into $v$ sends. In this case we appropriately duplicate paths into $v$ to more finely divide up the flow that $f$ sends into $v$ so that this flow can actually be forwarded along arcs of $\delta^+(v)$. The challenge then becomes to argue that we do not need to duplicate paths too many times; we can bound the number of times we must duplicate a path in this way and therefore the size of the support of $f$ by uniquely charging each duplication to the moment the $\hat{f}$ flow is fully decomposed along some arc.

We formalize this algorithm in \Cref{alg:sparseDecomp}; recall that for an $h$-length flow $f : \mcP_h(S,T) \to \mathbb{R}_+$ we let $f(a) = \sum_{a \ni P}f_P$ and for a (non-length-constrained) flow $\hat{f} : A \to \mathbb{R}_+$ we let $\hat{f}_a$ be the flow value along $a$.

\begin{algorithm}
    \caption{Sparse $h$-Length Flow Decomposition}
    \label{alg:sparseDecomp}
    \begin{algorithmic}[0] 
            \State \textbf{Input:} $h$-layer $S$-$T$ digraph $D = (V = V_1 \sqcup V_2 \sqcup \ldots \sqcup V_{h+1}, A)$ and $S$-$T$ flow $\hat{f}$ where $\supp(\hat{f}) = A$.
            \State \textbf{Output:} An $h$-length $S$-$T$ flow $f$.
            \State Initialize $\mcP$: for each $a \in \delta^+(S)$ add to $\mcP$ the path $(a)$.
            \State Initialize $f$: $f_{P} \gets \hat{f}_a$ if $P = (a)$ for some $a \in \delta^+(S)$ and $f_P \gets 0$ otherwise.
            \State Initialize $i = 2$.
            \While{$i \leq h+1$}:
                \If{$\exists P \in \mcP$ with an endpoint in $V_i$}
                    \State Let $v \in V_i$ be an endpoint of $P$ and let $a$ be any arc in $\delta^+(v)$ with $\hat{f}_a - f(a) > 0$.
                    \State Let $x_a := \min(f_P, \hat{f}_a - f(a))$ be the flow of $P$ that can be sent along $a$.
                    \State Add $P \oplus a$ to $\mcP$ and set $f_{P \oplus a} \gets x_a$ and $f_P \gets f_P - x_a$.
                    \If{$f_P = 0$}
                        \State Remove $P$ from $\mcP$.
                    \EndIf
                \Else
                    \State $i \gets i+1$.
                \EndIf
            \EndWhile
            \State \Return $f$.
    \end{algorithmic}
\end{algorithm}

The following gives the formal properties of \Cref{alg:sparseDecomp}; note that in the following $|\supp(f)|$ is the cardinality of a collection of paths.
\begin{restatable}[Sparse Decomposition of Acyclic Flows]{thm}{sparseDecomp}\label{thm:sparseDecomp}
Given $h$-layer digraph $D=(V,A)$ with unit lengths and $h \geq 1$, source and sink vertices $S, T \subseteq V$ and an $S$-$T$ flow $\hat{f}$ where $A = \supp(\hat{f})$, \Cref{alg:sparseDecomp} computes an $h$-length $S$-$T$ flow $f$ where:
\begin{itemize}
    \item \textbf{$f$ Decomposes $\hat{f}$:} $f(a) = \hat{f}_a$ for every $a \in A$;
    \item \textbf{$f$ is Sparse:} $|\supp(f)| \leq m$;
\end{itemize}
and can be implemented in deterministic parallel time $\tilde{O}(h)$ with $m$ processors.
\end{restatable}

\begin{proof}
    Let $S = V_1 \sqcup V_2 \sqcup \ldots \sqcup V_{h+1} = T$ be the $h+1$ layers of $D$. Likewise, let $V_{\leq i} := \bigcup_{j \leq i} V_j$ and let $V_{> i} := \bigcup_{j > i} V_j$ be symmetric. A standard argument by induction and the flow conservation of $f$ shows that after the $i$th iteration we have:
    \begin{enumerate}
        \item Every path in $\mcP$ has one endpoint in $S$ and one endpoint in $V_{>i}$;\label{eq:1}
        \item For every $a$ with a source in $V_{\leq i}$ we have $\hat{f}(a) = \hat{f}_a$.\label{eq:2}
    \end{enumerate}
    Observe that every $S$ to $T$ path has length at most $h$ and so by \ref{eq:1} we know that $f$ is indeed an $h$-length $S$ to $T$ flow. Furthermore by \ref{eq:2} we have that $f(a) = \hat{f}_a$ for every $a \in A$.
    
    To see $|\supp(f)| \leq m$ consider one iteration $i$ of our while loop. It suffices to argue that at the end of this iteration we have that $|\mcP|$ is at most the number of arcs with a source in $V_{\leq i}$. Observe that as long as $f_P \neq x_a$ (before updating $f_P$) we have that $|\mcP|$ does not grow and so it suffices to argue that the number of iterations of our while loop where $f_P = x_a$ is at most this number of arcs. To this end, observe that when $f_P = x_a$ this means that at the end of this loop $\hat{f}_a - f(a) = 0$ and so $a$ will never again be considered by our while loop. Thus, we can map this increase in the size of $\mcP$ to $a$ where overall this mapping is injective.
    
    Lastly, it remains to argue that our algorithm can be implemented in parallel with the appropriate runtime. We will argue that one iteration of one loop of the while loop can be implemented in $\tilde{O}(1)$ parallel time with $m$ processors, giving the claim. Fix an iteration $i$ and a $v \in V_i$. We initailize $\mcP' \gets \emptyset$.
    
    The main challenge here is as follows. There might be an arc of $\delta^+(v)$ that must have $\Omega(m)$ paths' flow forwarded along it; this suggests we cannot have a single processor responsible for each arc of $\delta^+(v)$ as they would have to forward too many paths. Symmetrically, there might be a path of $\mcP$ which must be duplicated and have its flow split among $\Omega(m)$ different arcs of $\delta^+(v)$; this suggests we cannot have a single processor responsible for each path of $\mcP$. In the former case we would rather have a processor for each path and in the latter case we would like to have a processor for each arc in $\delta^+(v)$. Thus, the main trick we use here is to essentially have two types of processors: each processor is either responsible for a path $P \in \mcP$ \emph{or} an arc in $\delta^+(v)$.
    
    More formally, we do the following for each $v \in V_i$ in parallel. Let $P_1, P_2, \ldots$ be an arbitrary ordering of paths in $\mcP$ with $v$ as an endpoint and let 
    \begin{align*}
        p_j^- := \sum_{x \leq j} f_{P_x}
    \end{align*}
    be the $j$th prefix sum according to this ordering of $f$ flow. Let $a_1, a_2, \ldots$ be an arbitrary ordering of $\delta^+(v)$ and let 
    \begin{align*}
        p_j^+ := \sum_{x \leq j} \hat{f}_{a_x}
    \end{align*}
    be the $j$th prefix sum according to this ordering of $\hat{f}$ flow. Lastly, let $p_1, p_2, \ldots$ be  $\{p_j^-\}_j \bigcup \{p_j^+\}_j$ sorted in ascending order. 
    
    Then for each $j$ we forward $p_{j} - p_{j-1}$ flow from an appropriate path along an appropriate arc. Namely, by construction there must be some $k$ and some $l$ such that 
    \begin{align*}
        [p_{j-1},p_j] \subseteq [p_{k-1}^-, p_{k}^-] \text{ and }[p_{j-1}, p_{j}] \subseteq [p_{l-1}^+, p_{l}^+]
    \end{align*}
    and so we continue $p_{j} - p_{j-1}$ of $P_k$'s flow along arc $a_l$: add $P_k \oplus a_l$ to $\mcP'$ and set $f_{P \oplus a_l} \gets p_{j} - p_{j-1}$. At the end of this iteration we update by setting $f_{P} = 0$ for each $P \in \mcP$ and then set $\mcP \gets \mcP'$. See \Cref{fig:matchFlow} for an illustration.
    
    \begin{figure}
    \centering
    \begin{subfigure}[b]{0.27\textwidth}
        \centering
        \includegraphics[width=\textwidth,trim=200mm 40mm 200mm 0mm, clip]{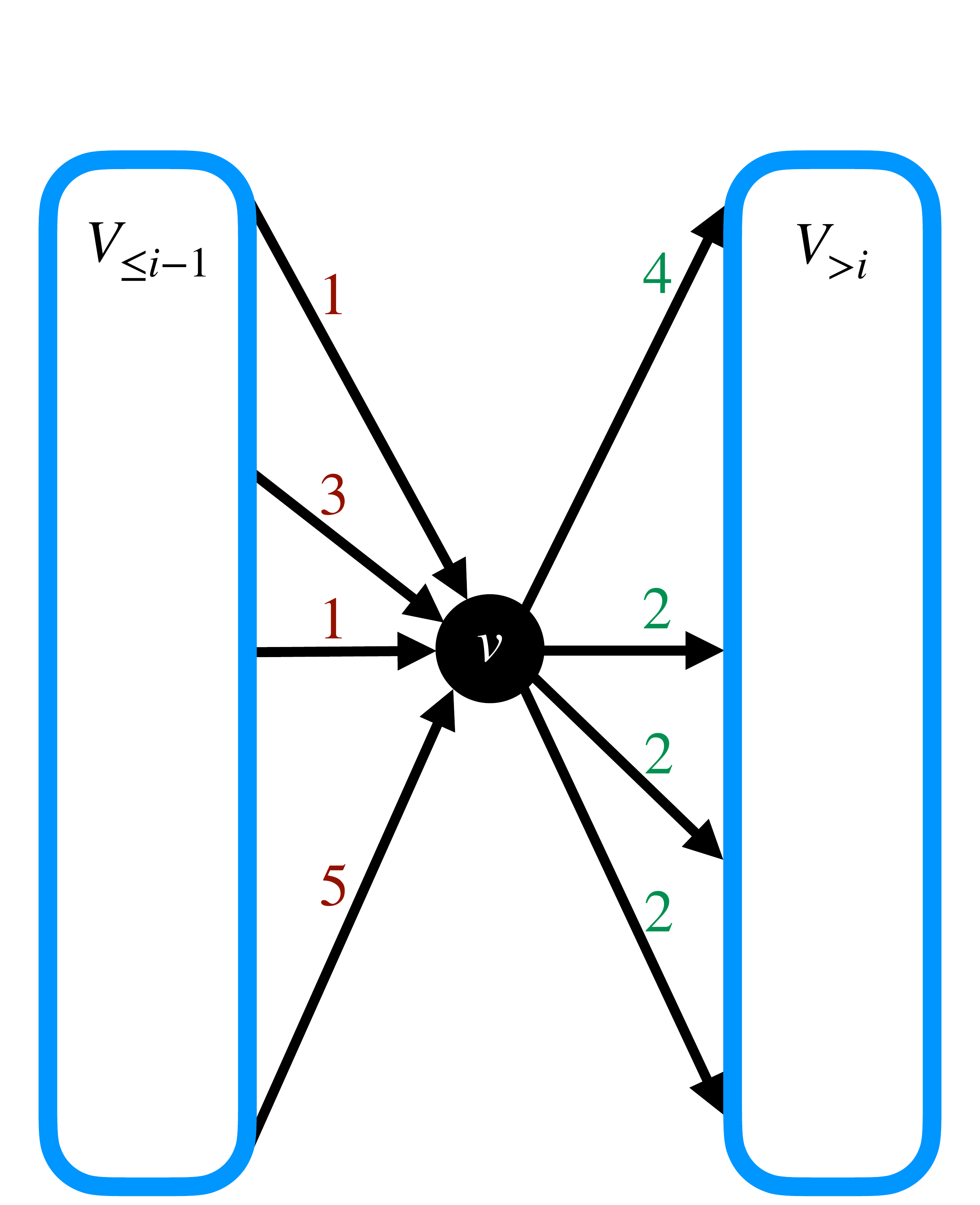}
        \caption{In $f$ and out $\hat{f}$ flow.}\label{sfig:match1}
    \end{subfigure}    \hfill
    \begin{subfigure}[b]{0.32\textwidth}
        \centering
        \includegraphics[width=\textwidth,trim=200mm 40mm 120mm 0mm, clip]{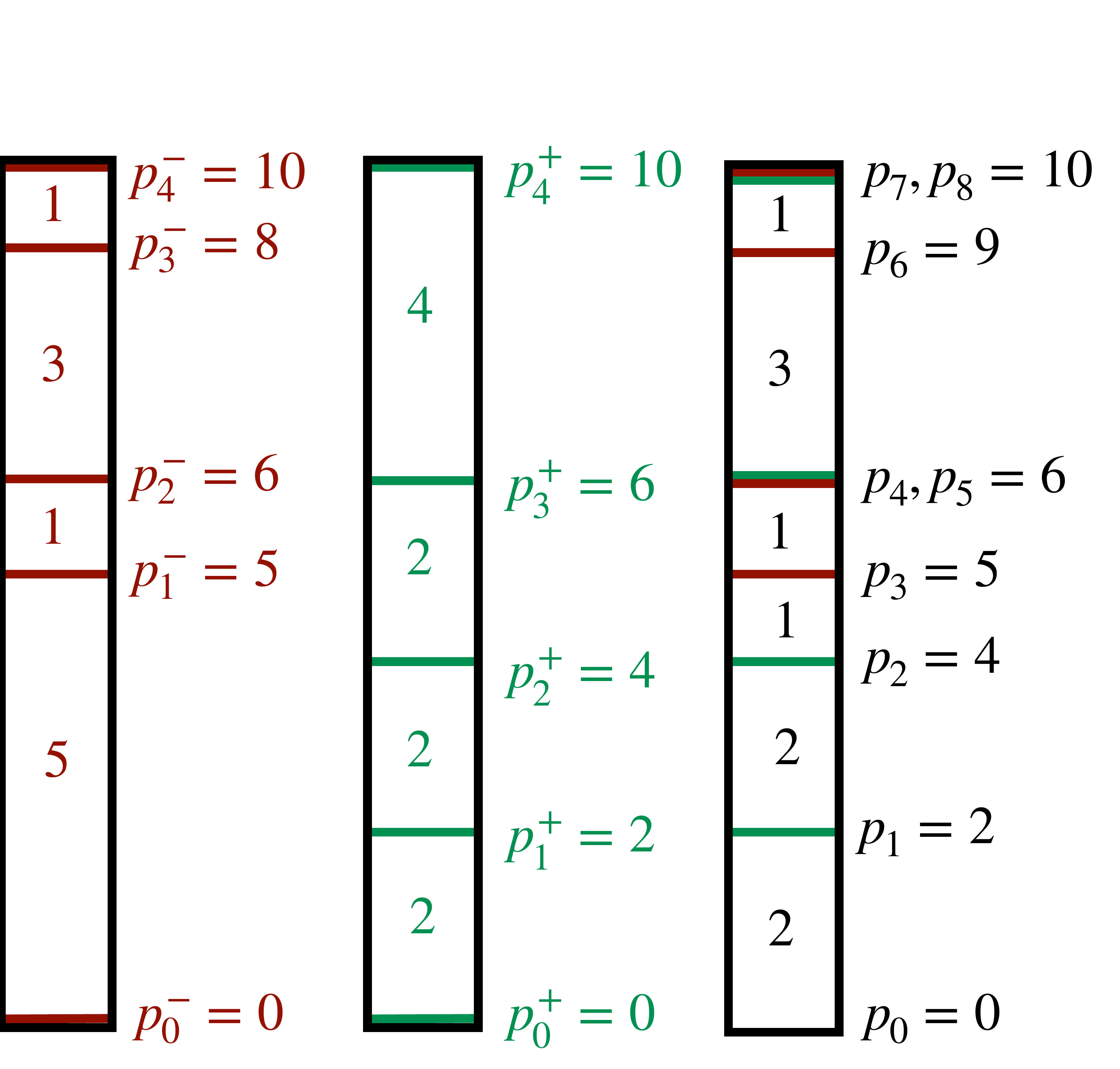}
        \caption{Prefix sums considered.}\label{sfig:match2}
    \end{subfigure}   \hfill
    \begin{subfigure}[b]{0.27\textwidth}
        \centering
        \includegraphics[width=\textwidth,trim=200mm 40mm 200mm 0mm, clip]{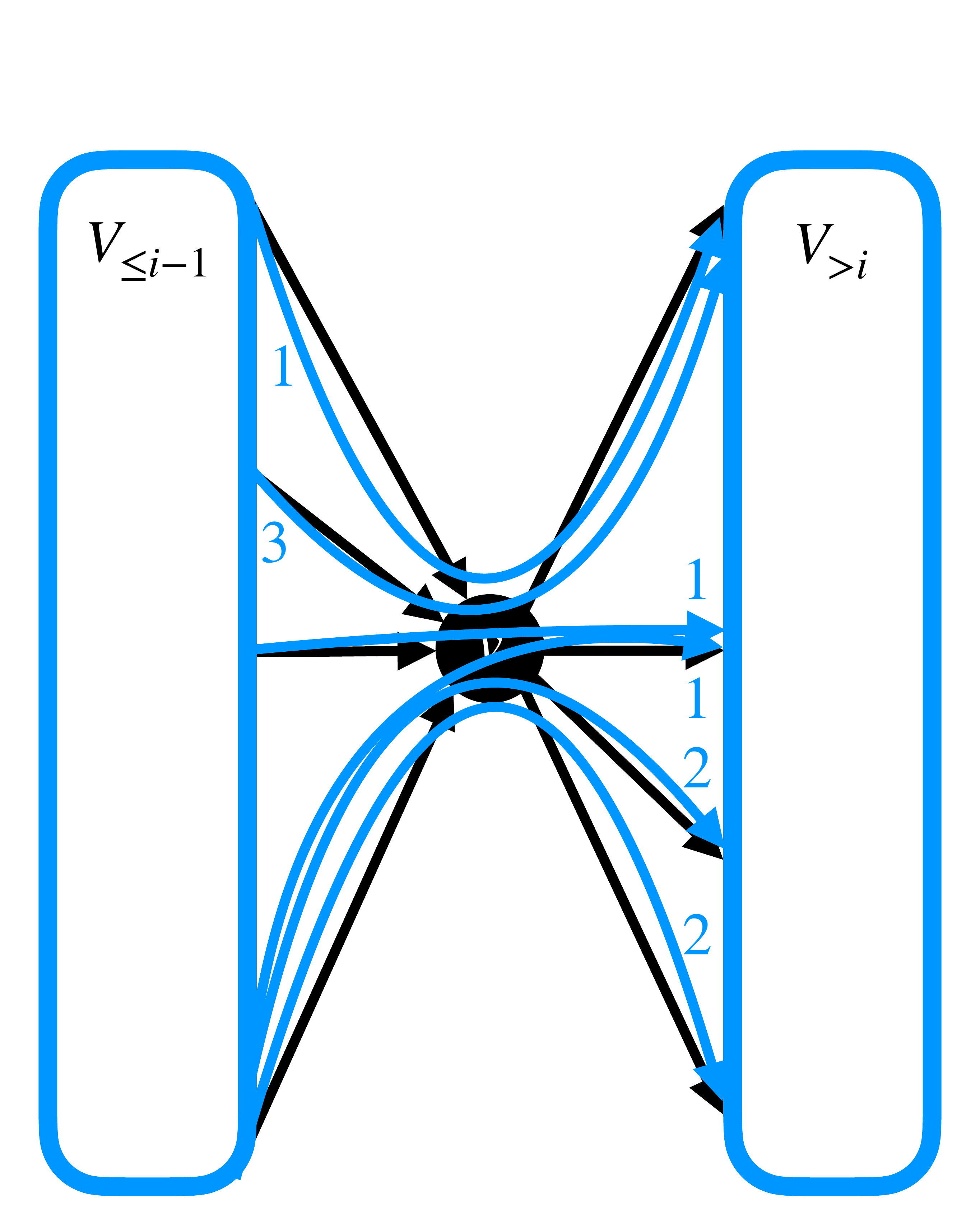}
        \caption{Forwarded flow.}\label{sfig:match3}
    \end{subfigure}
    \caption{An illustration of how we forward flow from paths incoming into $v$ along outgoing arcs. Here, we imagine that $P_1, P_2, \ldots$ are labelled bottom-up with their $f$ value in red and arcs $a_1, a_2, \ldots$ are labelled in bottom-up order by their $\hat{f}$ value in green.}\label{fig:matchFlow}
\end{figure}
    
    It is easy to verify that this is equivalent to one iteration of our algorithm's while loop where we always choose the $P$ which is earliest in the order $P_1, P_2, P_3, \ldots$ and always choose the $a \in \delta^+(v)$ satisfying $\hat{f}_a - f(a)$ that is earliest in the ordering $a_1, a_2, \ldots$. 
    
    Furthermore, we have that the total number of intervals of the form $p_j$ that we must consider across all vertices in layer $V_i$ is at most $m$ since the size of $\mcP$ is at most $|\delta^+(V_{\leq i-1})|$, each path in $\mcP$ contributes at most one such interval and each arc in $\delta^+(V_i)$ also contributes at most one such interval. Thus, computing $\{p_j^-\}_j$, $\{p_j^+\}_j$ and $\{p_j\}_j$ reduce to prefix sums and sorting at most $m$ numbers at most $O(1)$ times which are well-known to be doable in $\tilde{O}(1)$ parallel time with $m$ processors \cite{cole1989faster,blelloch2021introduction}. Lastly, observe that for a given $j$, identifying the $P_k$ and $a_l$ to extend $P_k$ along and  send $p_j - p_{j-1}$ flow across reduces to a binary search of $\{p_j^+\}_j$ and $\{p_j^-\}_j$.
\end{proof}

\section{$h$-Length $(1+\epsilon)$-Lightest Path Blockers}\label{sec:shortestPathBlockers}

In this section we show how to efficiently compute our main subroutine for our multiplicative-weights-type algorithm; what we call $h$-length $(1+\epsilon)$-lightest path blockers. We will use the blocking integral flow primitives of \Cref{sec:randMaxPaths} for our randomized algorithm and that of \Cref{sec:detMaxPaths} for our deterministic algorithm. Likewise, we will use the sparsification procedure from \Cref{sec:sparseDecomp} (as formalized by \Cref{thm:sparseDecomp}) to guarantee that these $(1+\epsilon)$-lightest path blockers are sparse.

Our $(1+\epsilon)$-lightest path blockers are defined below. In what follows, $\lambda$ is intuitively a guess of $d_{\w}^{(h)}(S,T)$. Also, in the following recall that if $f$ is an $h$-length flow then $f$ assigns flow values to entire paths (rather than just arcs as a non-length-constrained flow does). As such the support of $f$, $\supp(f)$, is a collection of paths. However, as mentioned earlier, for an $h$-length flow $f$, we will use $f(a)$ as shorthand for $\sum_{P \ni a} f_P$.

\begin{definition}[$h$-length $(1+\epsilon)$-Lightest Path Blockers]\label{dfn:alphaPathBlocker}
Let $G = (V,E)$ be a graph with lengths $\l$, weights $\w$ and capacities $\U$. Fix $\epsilon >0$, $h \geq 1$, $\lambda \leq d^{(h)}_{\w}(S,T)$ and $S, T \subseteq V$. Let $f$ be an $h$-length integral $S$-$T$ flow. $f$ is an $h$-length $(1+\epsilon)$-lightest path blocker if:
\begin{enumerate}
    \item \textbf{Near-Lightest:} $P \in \supp(f)$ has weight at most $(1 + 2\epsilon) \cdot \lambda$;
    \item \textbf{Near-Lightest Path Blocking:} If $P' \in \mcP_h(S,T)$ has weight at most $(1+\epsilon) \cdot \lambda$ then there is some $a \in P'$ where $f(a) = \U_a$.
\end{enumerate}
\end{definition}
Our main theorem in this section shows how to compute $(1 + \epsilon)$-lightest path blockers efficiently.

\begin{restatable}{thm}{pathBlockerAlg}
\label{thm:pathBlockerAlg}
Given digraph $D = (V,A)$ with lengths $\l$, weights $\w$, capacities $\U$, length constraint $h \geq 1$,  $\eps >0$, $S, T \subseteq V$ and $\lambda \leq d_{\w}^{(h)}(S,T)$, one can compute $h$-length $(1+\epsilon)$-lightest path blocker $f$ in:
\begin{enumerate}
    \item Deterministic parallel time $\tilde{O}(\frac{1}{\eps^5} \cdot h^{16})$ with $m$ processors where $|\supp(f)| \leq \tilde{O}(\frac{h^9}{\eps^3} \cdot |A|)$;
    \item Randomized CONGEST time $\tilde{O}(\frac{1}{\eps^5} \cdot h^{16})$ with high probability;
    \item Deterministic CONGEST time $\tilde{O}\left(\frac{1}{\eps^5} \cdot h^{16}  + \frac{1}{\eps^3} \cdot h^{15} \cdot (\rho_{CC})^{10} \right)$.
\end{enumerate}
\end{restatable}

The main idea for computing these objects is to reduce finding them to computing a series of blocking flows in a carefully constructed ``length-weight expanded DAG.'' In particular, by rounding arc weights up to multiples of $\frac{\eps}{h}\lambda$ we can essentially discretize the space of weights. Since each path has at most $h$ arcs, it follows that this increases the weight of a path by at most only $\lambda\eps$. This discretization allows us to construct DAGs from which we may extract blocking flows which we then project back into $D$ and then ``decongest'' so as to ensure they are feasible flows.


\subsection{Length-Weight Expanded DAG}\label{sec:lengthExpanded}

We now formally define the length-weight-expanded DAGs on which we compute blocking integral flows. Roughly, the length-weight expanded graph will create many copies of vertices and organize them into a grid where moving further down in rows corresponds to increases in length and moving further along in columns corresponds to increases in weight.

Let $D = (V, A)$ be a digraph with specified source and sink vertices $S$ and $T$, lengths $\l$, weights $\w$, capacities $\U$ and a parameter $\lambda \leq d_{\w}^{(h)}(S,T)$. We let $\tilde{\w}$ be $\w$ but rounded up to the nearest multiple of $\frac{\eps}{h} \cdot \lambda$. That is, for each $a \in A$ we have
\begin{align*}
    \tilde{\w}_a = \frac{\eps \cdot  \lambda}{h }  \cdot \left\lceil \w_a \cdot \frac{h}{\eps \cdot \lambda} \right\rceil
\end{align*}
See \Cref{fig:gridPre} for an illustration of $\tilde{w}$.

\begin{figure}
    \centering
    \begin{subfigure}[b]{0.45\textwidth}
        \centering
        \includegraphics[width=\textwidth,trim=0mm 0mm 0mm 0mm, clip]{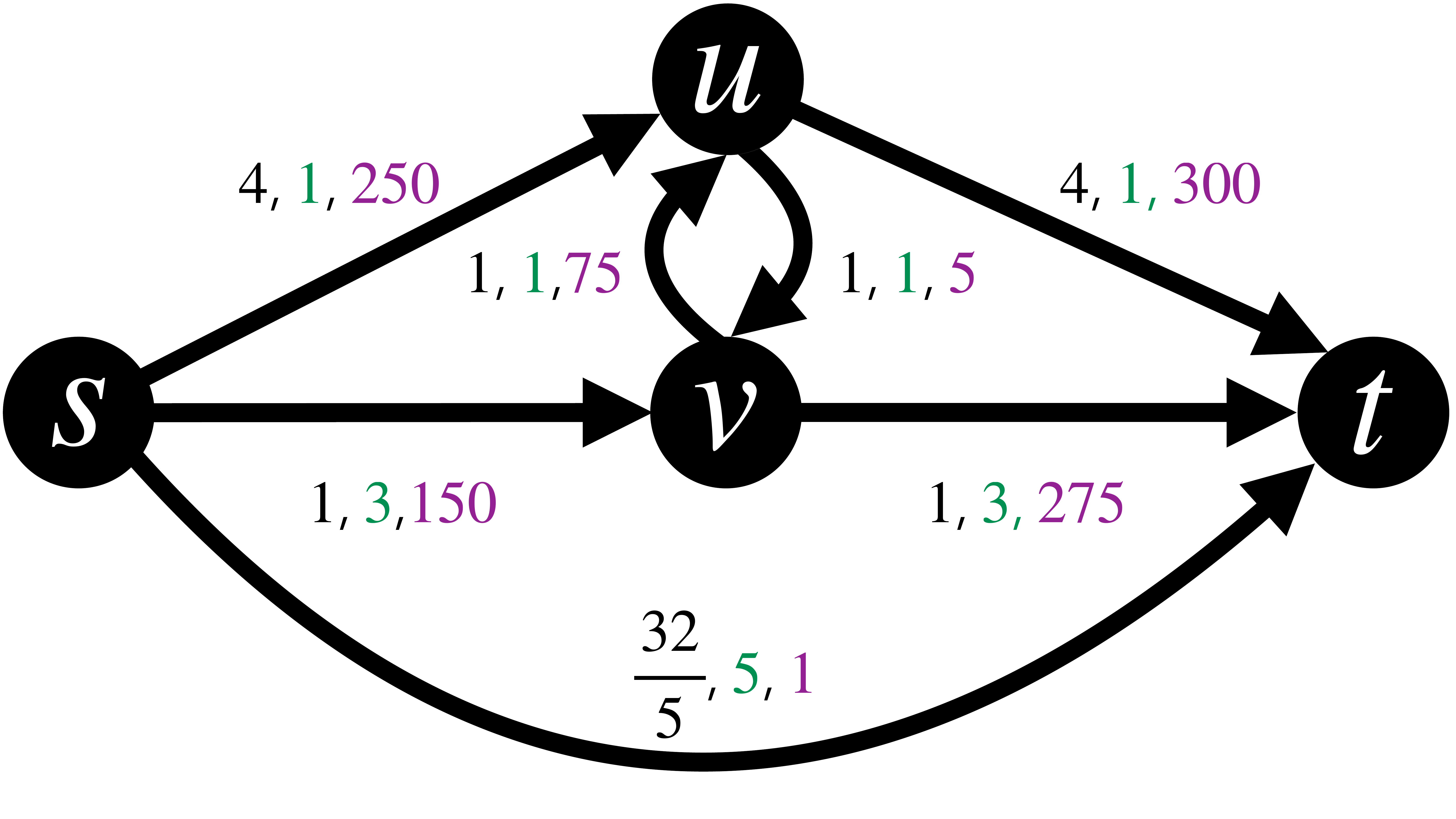}
        \caption{Digraph $D$ with $w$.}\label{sfig:gridPre1}
    \end{subfigure}    \hfill
    \begin{subfigure}[b]{0.45\textwidth}
        \centering
        \includegraphics[width=\textwidth,trim=0mm 0mm 0mm 0mm, clip]{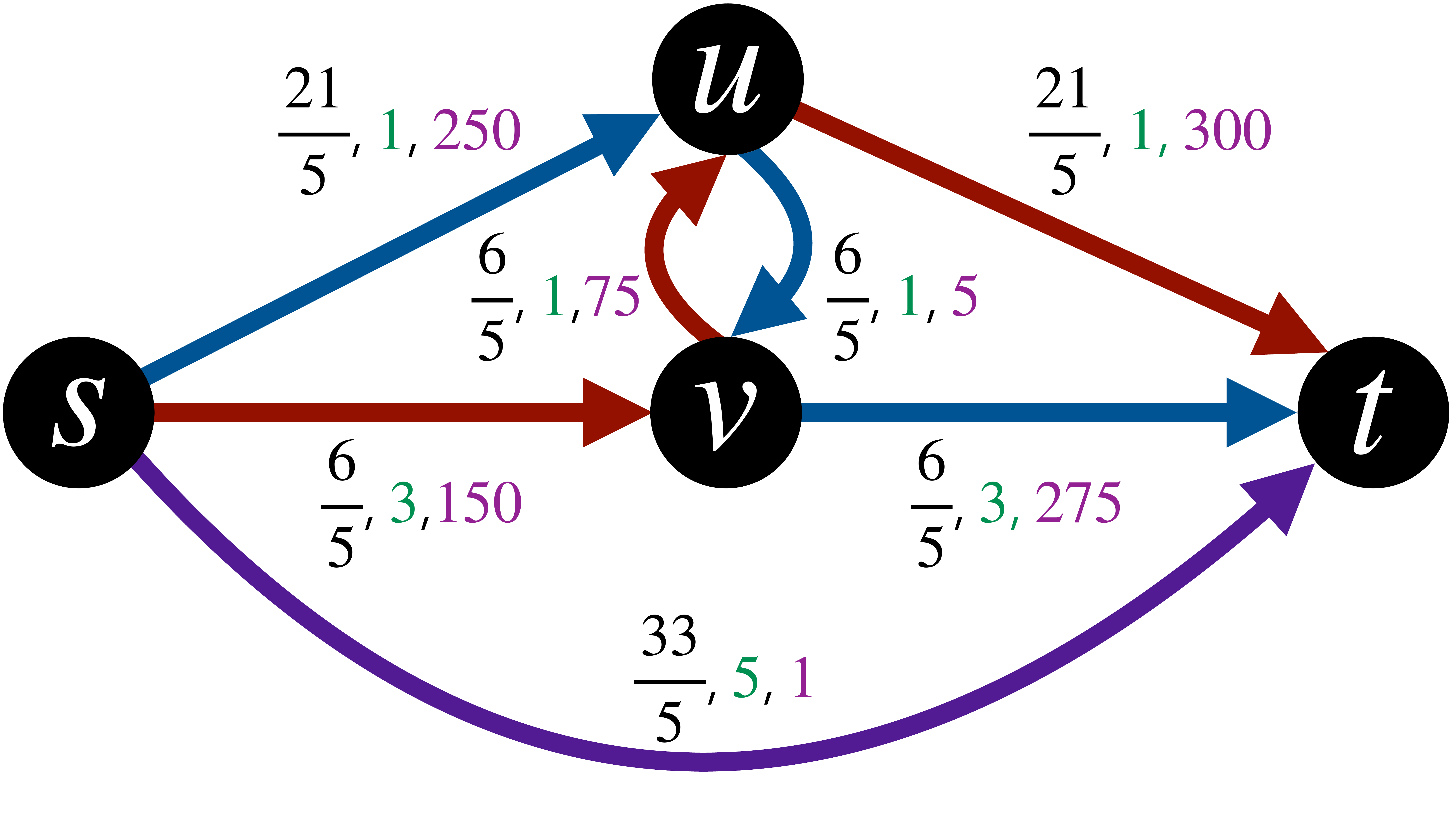}
        \caption{Digraph $D$ with $\tilde{w}$.}\label{sfig:gridPre2}
    \end{subfigure}   
    \caption{An illustration of how we round weights according to $\eps$, $\lambda$ and $h$. Here $h = 5$, $\lambda = 6$ and $\eps = .5$ and so we round to multiples of $\frac{\eps}{h} \lambda = \frac{3}{5}$. \ref{sfig:gridPre1} gives our input DAG where each arc is labeled with its weight, then length, then capacity and \ref{sfig:gridPre2} gives the weights after we round them where we color each lightest $5$-length path from $s$ to $t$.}\label{fig:gridPre}
\end{figure}

Next, we define the length-weight expanded DAG $D^{(h,\lambda)}  = (V', A')$ with capacities $\U'$. See \Cref{fig:grid} for an illustration of $D^{(h,\lambda)}$.
\begin{figure}
    \centering
        \includegraphics[width=\textwidth,trim=0mm 80mm 0mm 0mm, clip]{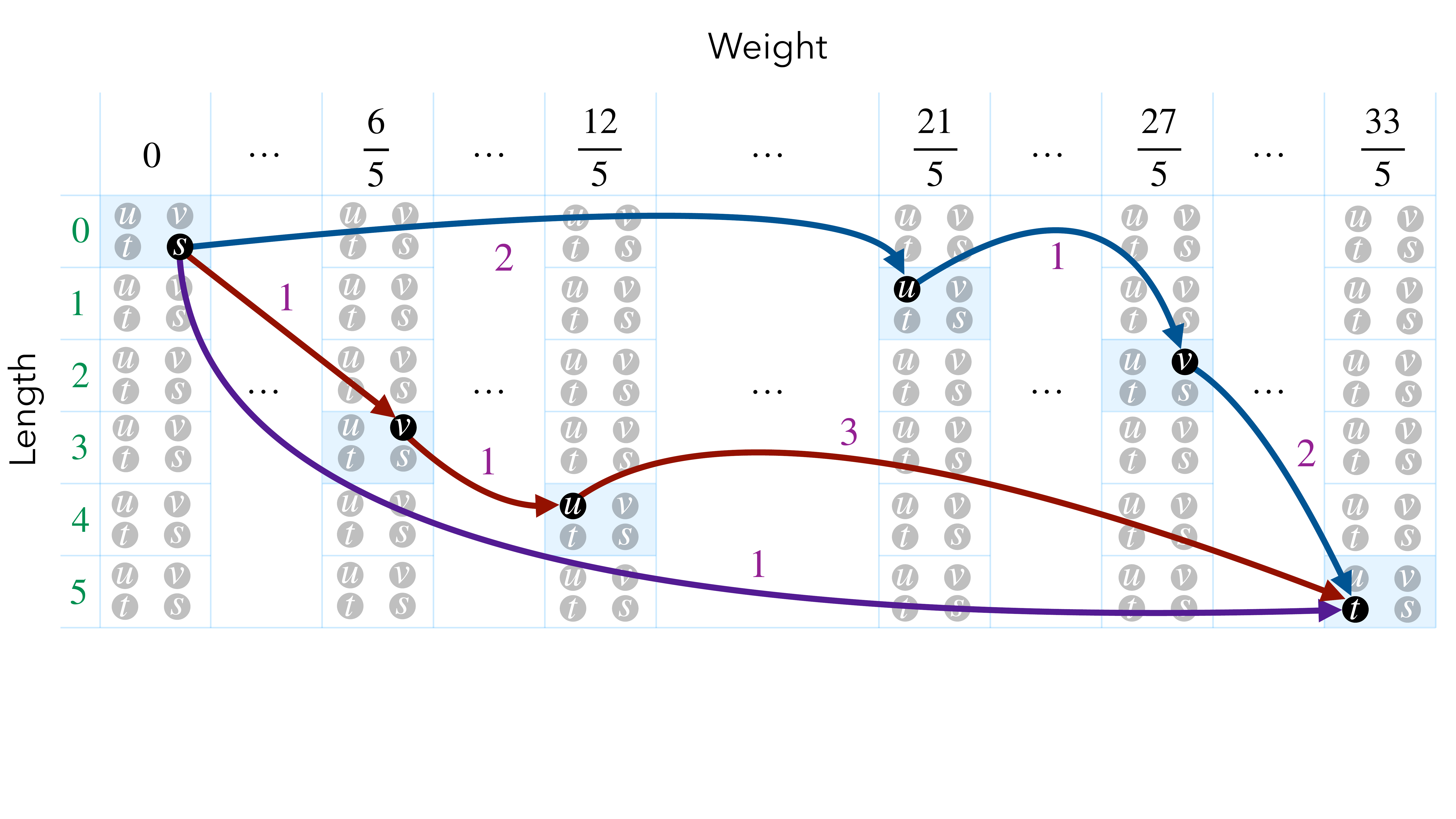}
    \caption{An illustration of $D^{(h,\lambda)}$ where $D$ and the parameters we use are given by \Cref{fig:gridPre}, $\kappa = 100$, $S = \{s\}$ and $T = \{t\}$. Copy $v(x,h')$ of vertex $v$ is in the $(x,h')$th grid cell and each arc is labelled with its capacity. We only illustrate the subgraph between $s(0, 0)$ and $t(\frac{33}{5}, 5)$. Each path is colored according to the path in \Cref{sfig:gridPre2} of which it is a copy. Notice that the graph induced by all $5$-length lightest paths in \Cref{sfig:gridPre2} is not a DAG but $D^{(h,\lambda)}$ is.}\label{fig:grid}
\end{figure}

\begin{itemize}
    \item \textbf{Vertices:} We construct the vertices $V'$ as follows. For each each vertex $v \in V$ we make $\kappa = h \cdot (\frac{h}{\epsilon} + 2h)$ copies of $v$, where we let $v(x,h')$ be one of these vertices; here $x$ ranges over all multiples of $\frac{\eps}{h} \cdot \lambda$ up to $(1+2\eps) \cdot \lambda$ (of which there are $\frac{h}{\epsilon} + 2h$) and $h' \leq h$. Intuitively, there will be a path from a copy of a vertex $s\in S$ to a vertex $v(x,h')$ iff there is a path with exactly $x$ weight (according to $\tilde{\w}$) and $h'$-length from $s$ to $v$ in $D$.
    \item \textbf{Arcs:} We construct the arcs $A'$ as follows. For each each vertex $v \not \in T$ and each $a = (v,u) \in \delta^+(v)$ we do the following. For each copy $v(x,h')$ of $v$ we add an arc to $A'$ from $v(x,h')$ to $u(x + \tilde{w}_a, h' + \l_a)$ provided $u(x + \tilde{w}_a, h' + \l_a)$ is actually a vertex in $V'$. That is, provided $x + \tilde{w}_a \leq (1 + 2 \eps) \cdot \lambda$ and $h' + \l_a \leq h$. We say that the arc $v(x,h')$ to $u(x + \tilde{w}_a, h' + \l_a)$ in $A'$ is a copy of arc $a$. For a given $a \in A$, we let $A'(a)$ give all copies of arc $a$ that are in $A'$.
    \item \textbf{Capacities:} We construct the capacities $\U'$ as follows. For low capacity arcs we set the capacity of all copies to $1$; for high capacity arcs we evenly distribute the capacity across all copies. Specifically, suppose arc $a' \in A'$ is a copy of arc $a \in A$. Then if $0 < \U_{uv} \leq \kappa$ we let $\U'_{a'} = 1$. Otherwise, we let $\U'_{a'}$ have capacity $\lfloor \U_{a} / \kappa \rfloor$.  As we will see later in our proofs, this rebalancing of flows will guarantee that when we ``project'' a flow from $D^{(h,\lambda)}$ to $D$, the only arcs that end up overcapacitated in $D$ are arcs with capacity at most $\kappa$. This, in turn, will allow us to argue that the conflict graph on which we compute an MIS is small.
\end{itemize}
We let $V'(S)$ and $V'(T)$ be all copies of $S$ and $T$ in $D^{(h,\lambda)}$ and we delete any vertex from $D^{(h,\lambda)}$ that does not lie on a $V'(S)$ to $V'(T)$ path. This will guarantee that the resulting digraph is indeed a $V'(S)$-$V'(T)$ DAG.

Lastly, we clarify what it means for a path to have its copy in $D^{(h,\lambda)}$. Suppose $P = (a_1, a_2, \ldots)$ is a path in $D$ that visits vertices $s=v_1, v_2, \ldots, v_k = t$ in $D$ and let $\tilde{\w}_i$ and $\l_i$ be the weight (according to $\tilde{\w}$) and length of $P$ summed up to the $i$th vertex it visits. Then we let $a_i'$ be the arc from $v_i(\tilde{\w}_i, \l_i)$ to $v_{i+1}(\tilde{\w}_i + \tilde{\w}_{a_i}, \l_i + \l_{a_i})$. If $a_i'$ is in $D^{(h,\lambda)}$ for every $i$ then we call $P' = (a_1', a_2', \ldots)$ the copy of $P$ in $D^{(h,\lambda)}$. Observe that a path in $D$ has at most one copy in $D^{(h,\lambda)}$ but every path in $D^{(h,\lambda)}$ is the copy of some path in $D$.

The following summarizes the key properties of our length-weight expanded digraphs.
\begin{lemma}\label{lem:lengthWeightExpanded}
Let $D = (V,A)$ be a digraph with weights $\w$, $S, T \subseteq V$ and some $\lambda \leq d^{(h)}_{\w}(S,T)$. Let $D^{(h,\lambda)} = (V', A)$ be the length-weight expanded digraph of $D$. Then $D^{(h,\lambda)}$ is an $h$-layer $V'(S)$-$V'(T)$ DAG which satisfies
\begin{enumerate}
    \item \textbf{Few  Arc Copies}: $|A'(a)| \leq O(\frac{h^2}{\eps})$.
    \item \textbf{Forward Path Projection}: For each path $P$ in $D$ from $S$ to $T$ of weight at most $\lambda \cdot(1+\epsilon)$ according to $\w$, there is a copy of $P$ in $D^{(h,\lambda)}$ from $V'(S)$ to $V'(T)$.
    \item \textbf{Backward Path Projection}: If $P'$ is a $V'(S)$ to $V'(T)$ path in $D^{(h,\lambda)}$ then it is a copy of a path with weight at most $(1 + 2\epsilon) \cdot \lambda$ according to $\w$.
    \item \textbf{Optimal Flow Preserving}: the maximum $V'(S)$-$V'(T)$ flow on $D^{(h,\lambda)}$ has value at least $\Omega(\frac{\eps}{h^2})$ times that of the maximum $h$-length flow on $D$.
\end{enumerate}
\end{lemma}
\begin{proof}
First, we argue that $D^{(h,\lambda)}$ is indeed a DAG. To see this, observe that if $a'$ is an arc in $A'$ from $v(x_1,h_1)$ to $v(x_2, h_2)$ then by construction it must be the case that $h_1 < h_2$. It follows that $D^{(h,\lambda)}$ has no cycles and has at most $h$ layers. Next, observe that $D^{(h,\lambda)}$ is a $V'(S)$-$V'(T)$ DAG by construction since we deleted any any vertices that do not lie on a path between $V'(S)$ and $V'(T)$. Additionally, we have $|A'(a)| \leq O(\frac{h^2}{\eps})$ for every $a$ since each vertex has at most $O(\frac{h^2}{\eps})$-many copies.

Next, consider an arc $a$ with weight $\w_a$ according to $\w$ and weight $\tilde{\w}_a$ according to $\tilde{\w}$. Observe that since we are rounding arc weights up we have $\w_a \leq \tilde{\w}_a$. Combining this with the fact that we are rounding to multiples of $\frac{\epsilon}{h} \cdot \lambda$ we have that 
\begin{align}\label{eq:lengthInc}
\w_a \leq \tilde{\w}_a \leq \w_a + \frac{\epsilon}{h} \cdot \lambda
\end{align}

We next argue our forward path projection property. That is, for each $h$-length path $P$ in $D$ from $S$ to $T$ of weight at most $\lambda \cdot(1+\epsilon)$ according to $\w$, there is a copy of $P$ in $D^{(h,\lambda)}$ from $V'(S)$ to $V'(T)$. First, observe that $P$ consists of at most $h$-many edges and so applying \Cref{eq:lengthInc}, its weight according to $\tilde{\w}$ is at most $\lambda \cdot(1+\epsilon) + h \cdot \frac{\epsilon}{h} \cdot \lambda = \lambda \cdot(1+2\epsilon)$. Next, observe that since $P$ has weight at most $\lambda \cdot(1+2\epsilon)$ according to $\tilde{w}$, it must have a copy in $D^{(h,\lambda)}$. In particular, suppose $P = (a_1, a_2, \ldots)$ visits vertices $s=v_1, v_2, \ldots, v_k = t$ in $D$ and let $\tilde{\w}_i$ and $\l_i$ be the weight (according to $\tilde{\w}$) and length of $P$ up to the $i$th vertex it visits. Then $D^{(h,\lambda)}$ always includes the arc from $v_i(\tilde{\w}_i, \l_i)$ to $v_{i+1}(\tilde{\w}_i + \tilde{\w}_{a_i}, \l_i + \l_{a_i})$ since $\tilde{w}_i \leq (1 + 2 \epsilon) \lambda$ and  $\l_i \leq h$ for every $i$.

We argue our backward path projection property. That is, if $P'$ is a $V'(S)$ to $V'(T)$ path in $D^{(h,\lambda)}$ then it is a copy of a path with weight at most $(1 + 2\epsilon) \cdot \lambda$ in $D$ according to $\w$. Since each arc in $D^{(h,\lambda)}$ is a copy of some arc in $D$, we know that $P'$ is a copy of \emph{some} path in $D$. Moreover, since we let $v(x,h')$ only range over $x \in \frac{h}{\epsilon} + 2h$, it follows that the weight of this path according to $\tilde{w}$ is at most $(1 + 2 \eps) \cdot \lambda$. However, since weights according to $\tilde{w}$ are only larger than those according to $w$ by \Cref{eq:lengthInc}, it follows that $P'$ is a copy of a path with weight at most $(1 + 2 \eps) \cdot \lambda$ according to $\w$.

Lastly, to see the optimal flow preserving property notice that if $f^*$ is the optimal $h$-length flow on $D$ then by how chose the capacities of $D^{(h,\lambda)}$ we have that the flow that gives path $P'$ in $D^{(h,\lambda)}$ value $\Theta(\frac{\eps}{h^2}) \cdot f^*_P$ where $P'$ is the copy of $P$ is indeed a feasible flow in $D^{(h,\lambda)}$.
\end{proof}

\subsection{Decongesting Flows}
Part of what makes using our length-weight expanded digraph non-trivial is that when we compute a flow in it and then project this flow back into $D$, the projected flow might not respect capacities. However, this flow will only violate capacities to a bounded extent and so in this section we show how to resolve such flows at a bounded loss in the value of the flow. In the below we say that an $h$-length flow $\hat{f}$ is $\alpha$-congested if any arc $a$ where $\hat{f}(a) > \U_a$ satisfies $\hat{f}(a) \leq \alpha$.
\begin{lemma}\label{lem:decongest}
There is a deterministic algorithm that, given a digraph $D = (V,A)$ with capacities $\U$, a length constraint $h \geq 1$, $S, T \subseteq V$ and an $h$-length $\alpha$-congested $S$-$T$ integral flow $\hat{f}$, computes an $S$-$T$ $h$-length integral flow $f$ where $\st(f) \geq \frac{1}{\alpha^2 h^2} \cdot \st(\hat{f})$ and $|\supp(f)| \leq |\supp(f')|$ in:
\begin{enumerate}
    \item Deterministic parallel time $\tilde{O}(\alpha^2 \cdot h)$ with $m$ processors;
    \item Deterministic CONGEST time $\tilde{O}(\alpha^3 \cdot h^3)$.
\end{enumerate}
\end{lemma}
\begin{proof}
The basic idea is to consider the conflict graph induced by our flow paths and then to compute a approximate maximum-weighted independent set among these flow paths where flow paths are weighted according to their flow value.

Specifically, construct our conflict graph $G' = (V', E')$ of $\supp(\hat{f})$ as follows. $V' = \supp(\hat{f})$ has a vertex for each path in the support of $\hat{f}$.  We say that $P_1$ and $P_2$ in $\supp(\hat{f})$ \emph{conflict} if there is some arc $a$ in both $P_1$ and $P_2$ such that $\hat{f}(a) > \U_a$. Then we add edge $\{P_1, P_2\}$ to $E'$ iff $P_1$ and $P_2$ conflict.

Observe that since each path in $\supp(\hat{f})$ consists of at most $h$ arcs and since $\hat{f}$ is $\alpha$-congested, we know that the maximum degree in $G'$ is at most $h \cdot \alpha$. 

We then apply \Cref{thm:maxIS} to $G'$ to compute a $\frac{1}{h \alpha}$-approximate maximum independent set in $G'$ in deterministic CONGEST time $\tilde{O}(h \alpha)$ with the node weight of $P \in \supp(\hat{f})$ as $\hat{f}_P$. Let $\mcI$ be this independent set and let $f = \sum_{P \in \mcI} \hat{f}_P$ be the flow corresponding to this set. We return $f$.

We trivially have $|\supp(f)| \leq |\supp(f')|$ by construction of $f$.

We next argue that $\st(f) \geq \frac{\st(\hat{f})}{\alpha^2 h^2}$
Since the total node weight in $G'$ is $\st(\hat{f})$ and the maximum degree in $G'$ is $\alpha \cdot h$, it follows that the maximum independent set in $G'$ has node weight at least $\frac{\st(\hat{f})}{\alpha h}$. Since $\mcI$ is $\frac{1}{\alpha h}$-approximate, we conclude that $f$ has $\st(f) \geq \frac{\st(\hat{f})}{\alpha^2 h^2}$.

Lastly, we argue that we achieve the claimed running times. Notice that the total number of vertices in $G'$ is at most $m \cdot \alpha$ because each congested arc $a$ where $\U_a <\hat{f}(a) \le \alpha$ is contained in at most $\alpha$ integral flow paths.  Hence, we can simulate any CONGEST algorithm in $G'$ with at most $\alpha$ overhead. \Cref{thm:maxIS} tells us that we can compute $\mcI$ in time at most $\tilde{O}(\alpha \cdot h)$ in $G'$, giving our parallel running time. 

It remains to describe how to simulate $G'$ in $D$ in CONGEST. 
We keep the following invariant: if a node $P_{1}$ in $G'$ receives a message, we make sure that all vertices $v\in P_{1}$ in $G$ receive the same message too. Because of this, any vertex $v\in P_{1}$ in $G$ can determine what $P_{1}$ as a node in $G'$ will do next. Let us assume that each message in $G'$ from $P_{1}$ to $P_{2}$ is of the form $(\texttt{msg},P_{1},P_{2})$. To simulate sending $(\texttt{msg},P_{1},P_{2})$ in $G$, a vertex $v_{1}\in P_{1}$ first forwards $(\texttt{msg},P_{1},P_{2})$ through $P_{1}$ to make sure that every node in $P_{1}$ gets this message. Let $v_{2}\in P_{1}\cap P_{2}$ be a common vertex in both $P_{1}$ and $P_{2}$. Then, $v_{2}$ forwards $(\texttt{msg},P_{1},P_{2})$ through $P_{2}$. After we are done simulating all messages sent in $G'$, our invariant is maintained. 

Now, we analyze the overhead of simulating one round of $G'$ in $G$. The dilation for simulating sending each message in $G'$ is clearly $O(h)$. Next, we analyze the congestion. Each arc $a$ is contained in at most $\max\{U_{a},\alpha\}\le\alpha U_{a}$ paths. For each such path $P$, there are at most $\alpha h$ messages needed to sent through $a$ because the maximum degree in $G'$ at most $\alpha h$. Therefore, the congestion is at most $\frac{\alpha U_{a}\cdot\alpha h}{U_{a}}=\alpha^{2}h$. Note that, here (and nowhere else in this work) we rely on the fact that we may send $O(\U_a)$ messages over an arc $a$ with capacity $\U_a$ in one round of CONGEST.

To conclude, the deterministic simulation overhead is at most dilation
times congestion which is at most $O(h)\cdot \alpha^{2}h=O(\alpha^{2}h^{2})$. Combining this simulation with the $\tilde{O}(\alpha \cdot h)$ running time of our approximate maximum independent set algorithm gives our CONGEST running time.
%
\end{proof}

\subsection{Computing $h$-Length $(1+\eps)$-Lightest Path Blockers }\label{sec:computeShortestPathBlockers}
Having described our length-weight expanded DAGs, their properties and how to decongest flows that we compute using them, we now use these primitives to build our $h$-length $(1+\epsilon)$-lightest path blockers. Again, the basic idea is to compute the length-weight expanded DAG $D^{(h,\lambda)}$, compute blocking flows in  $D^{(h,\lambda)}$, project these back into $D$, decongest the resulting flows and then repeat. \Cref{alg:pathBlocker} gives our algorithm. We prove its properties below.

\begin{algorithm}
    \caption{$(1+\epsilon)$-Lightest Path Blocker}
    \label{alg:pathBlocker}
    \begin{algorithmic}[0] 
            \State \textbf{Input:} $D = (V, A)$ with weights $\w$, lengths $\l$, capacities $\U$, $h \geq 1$, $S,T \subseteq V$, $\lambda > 0$ and  $\eps > 0$.
            \State \textbf{Output:} $h$-length $(1+\eps)$-lightest path blocker $f$.
            \State Initialize solution $f$ to be $0$ on all paths.
            \State Let $D^{(h,\lambda)} = (V', A')$ be the length-weight expanded digraph of $D$ with capacities $\hat{\U} = \U$
            \For{$\tilde{\Theta}(\frac{h^7}{\eps^2})$ repetitions}
                \State \textbf{Blocking Flows}: Let $f'$ be a blocking integral flow in $D^{(h,\lambda)}$ with capacities $\hat{U}$ (compute \\ \qquad \qquad  using \Cref{lem:randMax}  with randomness or \Cref{lem:detMax} deterministically).
                \State \textbf{Sparsify Flow}: Sparsify $f'$ so that $|\supp(f')| \leq |A'|$ (only for parallel alg.,\ use \Cref{thm:sparseDecomp})
                \State \textbf{Project Into $D$}: Let $\tilde{f}$ be the $h$-length flow that gives path $P$ value $f'_{P'}$ where $P'$ is the\\ \qquad \qquad copy of $P$ in $D^{(h,\lambda)}$.
                \State \textbf{Decongest Flow}: Let $\hat{f}$ be the result of decongesting $\tilde{f}$ with \Cref{lem:decongest}.
                \State For each copy $a' \in A'$ of $a \in A$ update capacities as $\hat{U}_{a'} = \hat{U}_{a'} - \hat{f}(a)$.
                \State Update $f = f + \hat{f}$.
            \EndFor
            \State \Return $f$.
    \end{algorithmic}
\end{algorithm}

\pathBlockerAlg*
\begin{proof}
We first argue that $f$ is a $h$-length $(1+\eps)$-lightest path blocker (\Cref{dfn:alphaPathBlocker}). $f$ is an integral $h$-length $S$-$T$ flow by construction. Moreover, the support of $f$ is near-lightest by the backward path projection property of $D^{(h, \lambda)}$, as stated in \Cref{lem:lengthWeightExpanded}. 

Also, notice that by the guarantees of \Cref{thm:sparseDecomp} and \Cref{lem:decongest} and the fact that projecting into $D$ does not increase the support size (i.e.\ $|\supp(\tilde{f})| \leq |\supp(f')|$) tells us that $|\supp(\hat{f})| \leq |A'|$ for each $\hat{f}$ (where, as a reminder, $A'$ is the arcs of $D^{(h, \lambda)}$) and so we know that
\begin{align*}
    |\supp(f)| \leq \tilde{O}\left(\frac{h^7}{\eps^2} \cdot |A'| \right)
\end{align*}
for our parallel algorithm. Applying the fact that $|A'| \leq O(\frac{h^2}{\eps}) \cdot |A|$ by \Cref{lem:lengthWeightExpanded} gives our bound on the support of $f$.

Thus, it remains to argue the near-lightest path blocking property of $f$ and, in particular that if $P \in \mcP_h(S,T)$ is a path in $D$ and $P$ has weight at most $(1 + \eps) \cdot \lambda$ according to $\w$ then there is some $a \in P$ where $f(a) = \U_a$. Towards this, observe that by the forward path projection property as stated in \Cref{lem:lengthWeightExpanded}, such a path $P$ has copy in $D^{(h, \lambda)}$. By how we construct $f$, it follows that to show $f(a) = \U_a$ for some $a$, it suffices to show that $\hat{\U}_a = 0$ by the end of our algorithm. To show that such an $a$ exists, it suffices to show that the maximum flow in $D^{(h, \lambda)}$ under the capacities $\hat{\U}$ is $0$ by the end of our algorithm. 

We do so now. Our strategy will be to show that we have implicitly computed a flow on $D^{(h,\lambda)}$ of near-optimal value and so after just a few iterations it must be the case that the optimal flow on $D^{(h,\lambda)}$ is reduced to $0$.

Consider a fixed iteration of our algorithm and let $\OPT^{(h, \lambda)}$ be the value of the maximum $V'(S)$-$V'(T)$ flow on $D^{(h,\lambda)}$. Since $f'$ is a blocking flow in $D^{(h,\lambda)}$ and $D^{(h,\lambda)}$ is an $h$-layer DAG by \Cref{lem:lengthWeightExpanded}, it follows from \Cref{lem:blockingGivesApx} that 
\begin{align}
    \st(f') \geq \frac{1}{h} \cdot \OPT^{(h, \lambda)}.\label{eq:ab}
\end{align}

Continuing, we claim that $\tilde{f}$ is an $O(\frac{h^2}{\eps})$-congested flow. In particular, any arc $a$ with capacity in $D$ greater than $O(\frac{h^2}{\eps})$ is such that the sum of its capacities across copies in $D^{(h,\lambda)}$ is at most $\hat{U}_a$. Thus, such an arc is never overcongested by $\tilde{f}$. Any arc with capacity less than $O(\frac{h^2}{\eps})$ in $D'$ has up to $O(\frac{h^2}{\eps})$ copies in $D^{(h,\lambda)}$ each of which has capacity $1$; thus, such an arc may have flow value up to $O(\frac{h^2}{\eps})$ in $\tilde{f}$.  Thus, by $\st(\tilde{f}) = \st(f')$ and  this bound on the congestedness of $\tilde{f}$, we have from \Cref{lem:decongest} that
\begin{align}
    \st(\hat{f}) &\geq \frac{\eps^2}{h^6} \cdot \st(\tilde{f}) \nonumber\\
    & = \frac{\eps^2}{h^6} \cdot \st(f').\label{eq:cd}
\end{align}

Combining \Cref{eq:ab} and \Cref{eq:cd}, we get 
\begin{align}
    \st(\hat{f}) \geq \frac{\eps^2}{h^7} \cdot \OPT^{(h, \lambda)}. \label{eq:ef}
\end{align}

Lastly, let $f''$ be $\hat{f}$ projected back into $D^{(h, \lambda)}$. That is, if arc $a'$ is a copy of arc $a$ then $f''$ assigns to $a'$ the flow value $\sum_{P \ni a} \hat{f}_P$. Observe that by construction of $\hat{f}$, we know that $f''$ is a $V'(S)$-$V'(T)$ flow in $D^{(h, \lambda)}$ of value $\st(f'') = \st(\hat{f})$. Thus, applying this and \Cref{eq:ef} we get

\begin{align*}
    \st(f'') \geq \frac{\eps^2}{h^7} \cdot \OPT^{(h, \lambda)}.
\end{align*}

Since we decrement the value of $\hat{U}_a$ by $f''_a$ in each iteration, it follows that after $\tilde{O}(\frac{h^7}{\eps^2})$ many repetitions of \Cref{alg:pathBlocker},  we must decrease the value of the optimal flow in $D^{(h, \lambda)}$ by at least a constant fraction since otherwise we would have computed a flow with value greater than that of the optimal flow. Since initially $\OPT^{(h, \lambda)} \leq \poly(n)$, we get that after $\tilde{O}(\frac{h^7}{\eps^2})$-many repetitions we have reduced the value of the optimal flow to $0$ on $D^{(h, \lambda)}$, therefore showing that $f$ satisfies the near-lightest path blocking property.

It remains to show our running times. The computation in each of our iterations is dominated by constructing the length-expanded digraph $D^{(h,\lambda)}$, computing our maximal integral flow $f^{(h)}$ in $D^{(h,\lambda)}$ and decongesting our flow.
\begin{itemize}
    \item  We can construct $D^{(h,\lambda)}$ by e.g.\ Bellman-Ford for $\tilde{O}(h)$ rounds for a total running time of $\tilde{O}(h)$ in either CONGEST or parallel. Likewise projecting flows back from $D^{(h,\lambda)}$ is trivial. 
    \item  It is is easy to simulate $D^{(h,\lambda)}$ in either CONGEST or in parallel with an overhead of $O(\frac{h^2}{\eps})$ since this is a bound on the number copies of each vertex.
    
    With randomization, by \Cref{lem:randMax} computing $f'$ takes time $\tilde{O}(h^3)$ in parallel with $m$ processors or $\tilde{O}(h^4)$ in CONGEST on $D^{(h,\lambda)}$ and so $\tilde{O}(\frac{h^5}{\eps})$ in parallel or $\tilde{O}(\frac{h^6}{\eps})$ in CONGEST on $D$. 
    
    For our deterministic algorithm, by \Cref{lem:detMax} doing so takes $\tilde{O}(h^3)$ in parallel with $m$ processors and CONGEST time $\tilde{O}(h^6 \cdot (\rho_{CC})^{10})$ on $D^{(h,\lambda)}$ and so $\tilde{O}(\frac{1}{\eps} \cdot h^5)$ parallel time on $D$ or $\tilde{O}(\frac{1}{\eps} \cdot h^8 \cdot (\rho_{CC})^{10})$ CONGEST time on $D$.
    
    \item Lastly, decongesting our flow by \Cref{lem:decongest} and the fact that $\tilde{f}$ is  $O(\frac{h^2}{\eps})$-congested takes deterministic parallel time $\tilde{O}(\frac{h^5}{\eps^2})$ and deterministic CONGEST time $\tilde{O}(\frac{h^9}{\eps^3})$.
\end{itemize}
Combining these running times with our $\tilde{O}(\frac{h^7}{\eps^2})$-many repetitions gives the stated running times.
\end{proof}

\section{Computing Length-Constrained Flows and Moving Cuts}\label{sec:MW}

Having shown how to compute an $h$-length $(1+\epsilon)$-lightest path blocker, we now use a series of these as batches to which we apply multiplicative-weights-type updates. The result is our algorithm which returns both a length-constrained flow and a (nearly) certifying moving cut.


\begin{algorithm}
    \caption{Length-Constrained Flows and Moving Cuts}
    \label{alg:mw}
    \begin{algorithmic}[0] 
            \State \textbf{Input:} digraph $D = (V,A)$ with lengths $\l$, capacities $\U$, $h \geq 1$, $S, T \subseteq V$ and $\eps \in (0,1)$.
            \State \textbf{Output:} $(1 \pm \eps)$-approximate $h$-length flow $f$ and moving cut $\w$.
            \State Let $\epsilon_0 = \frac{\epsilon}{6}$, let $\zeta = \frac{1+2 \eps_0}{\eps_0} + 1$ and let $\eta = \frac{\eps_0}{(1 + \eps_0) \cdot \zeta} \cdot \frac{1}{\log m}$.
            \State Initialize $\w_a \gets \left(\frac{1}{m}\right)^{\zeta}$ for all $a \in A$.
            \State Initialize $\lambda \gets  \left(\frac{1}{m}\right)^{\zeta}$.
            \State Initialize $f_P \gets 0$ for all $P \in \mcP_h(S,T)$.
            \While{$\lambda < 1$}:
                \For{$\Theta\left(\frac{h \log_{1+\epsilon_0} n}{ \epsilon_0} \right)$ iterations:}
                    \State Compute $h$-length $(1+\epsilon_0)$-lightest path blocker $\hat{f}$ (using \Cref{thm:pathBlockerAlg} with current $\lambda$).
                    \State \textbf{Length-Constrained Flow (Primal) Update:} $f \gets f + \eta \cdot \hat{f}$.
                    \State \textbf{Moving Cut (Dual) Update:} $\w_a \gets (1+\epsilon_0)^{\hat{f}(a)/ \U_a} \cdot \w_a$ for every $a \in A$.
                \EndFor
                \State $\lambda \gets (1 + \eps_0) \cdot \lambda$
            \EndWhile
            \State \Return $(f,\w)$.
    \end{algorithmic}
\end{algorithm}

As a reminder for an $h$-length flow $f$, we let $f(a) := \sum_{P \ni a} f_P$. Throughout our analysis we will refer to the innermost loop of \Cref{alg:mw} as one ``iteration.'' We begin by observing that $\lambda$ always lower bounds $d_{\w}^{(h)}(S,T)$ in our algorithm.
\begin{lemma}\label{lem:lambLowerBound}
At the beginning of each iteration of \Cref{alg:mw} we have $\lambda \leq d^{(h)}_{\w}(S,T)$
\end{lemma}
\begin{proof}
Our proof is by induction. The statement trivially holds at the beginning of our algorithm.

Let $\lambda_i$ be the value of $\lambda$ at the beginning of the $i$th iteration. We argue that if $d^{(h)}_{\w}(S,T) = \lambda_i$ then after $\Theta\left(\frac{h \log_{1+\epsilon_0} n}{ \epsilon_0} \right)$ additional iterations we must have $d^{(h)}_{\w}(S,T) \geq (1+\eps_0) \cdot \lambda_i$. Let $\lambda_i' = (1 + \epsilon_0) \cdot \lambda$ be $\lambda$ after these iterations. Let $\hat{f}_j$ be our lightest path blocker in the $j$th iteration.

Assume for the sake of contradiction that $d_{\w}^{(h)}(S,T) < \lambda_i'$ after $i + \Theta\left(\frac{h \log_{1+\epsilon_0} n}{ \epsilon_0} \right)$ iterations. It follows that there is some path $P \in \mcP_h(S,T)$ with weight at most $\lambda_i'$ after $i + \Theta\left(\frac{h \log_{1+\epsilon_0} n}{ \epsilon_0} \right)$ many iterations. However, notice that by definition of an $h$-length $(1+\epsilon_0)$-lightest path blocker (\Cref{dfn:alphaPathBlocker}), we know that for every $j \in \left[i, i + \Theta\left(\frac{h \log_{1+\epsilon_0} n}{ \epsilon_0}\right)\right]$ there is some $a \in P$ for which $\hat{f}_j(a) = \U_a$. By averaging, it follows that there is some single arc $a \in P$ for which $\hat{f}_j(a) = \U_a$ for at least $\Theta\left(\frac{\log_{1+\epsilon_0} n}{\epsilon_0}\right)$ of these $j \in [i, i + \Theta\left(\frac{h \log_{1+\epsilon_0} n}{ \epsilon_0}\right)]$. Since every such arc starts with dual value $(\frac{1}{m})^\zeta$ and multiplicatively increases by a $(1+\epsilon_0)$ factor in each of these updates, such an arc after $i + \Theta\left(\frac{h \log_{1+\epsilon_0} n}{ \epsilon_0} \right)$ many iterations must have $\w_a$ value at least $(\frac{1}{m})^\zeta \cdot (1+\epsilon_0)^{\Theta\left(\frac{\log_{1+\epsilon_0} n}{ \epsilon_0}\right)} \geq n^{2}$ for an appropriately large hidden constant in our $\Theta$. However, by assumption, the weight of $P$ is at most $\lambda_i'$ after $i + \Theta\left(\frac{h \log_{1+\epsilon_0} n}{ \epsilon_0} \right)$ iterations and this is at most $2$ since $\lambda_i < 1$ since otherwise our algorithm would have halted. But $2 < n^{2}$ and so we have arrived at a contradiction.

Repeatedly applying the fact that $\lambda_i' = (1+ \epsilon_0) \lambda_i$ gives that $\lambda$ is always a lower bound on $d^{(h)}_{\w}(S,T)$.
\end{proof}

We next prove the feasibility of our solution.
\begin{lemma}\label{lem:ana1}
The pair $(f, \w)$ returned by \Cref{alg:mw} are feasible for \ref{LP:hopConFlow} and \ref{LP:movingCut} respectively.
\end{lemma}
\begin{proof}
First, observe that by \Cref{lem:lambLowerBound} we know that $\lambda$ is always a lower bound on $d_{\w}^{(h)}(S,T)$ and so since we only return once $\lambda > 1$, the $\w$ we return is always feasible.

To see that $f$ is feasible it suffices to argue that for each arc $a$, the number of times a path containing $a$ has its primal value increased is at most $\frac{\U_a}{\eta}$. Notice that each time we increase the primal value on a path containing arc $a$ by $\eta$ we increase the dual value of this edge by a multiplicative $(1+\epsilon_0)^{1/\U_a}$. Since the weight of our arcs according to $\w$ start at $(\frac{1}{m})^{\zeta}$, it follows that if we increase the primal value of $k$ paths incident to arc $a$ then $\w_a = (1+\epsilon_0)^{k / \U_a} \cdot (\frac{1}{m})^{\zeta}$. On the other hand, by assumption when we increase the dual value of an arc $a$ it must be the case that $\w_a < 1$ since otherwise $d_{\w}^{(h)}(S,T) \geq 1$, contradicting the fact that $\lambda$ always lower bounds $d_{\w}^{(h)}(S,T)$. It follows that $(1+\epsilon_0)^{k/\U_a} \cdot (\frac{1}{m})^{\zeta} \leq 1$ and so applying the fact that $\ln(1+ \epsilon_0) \geq \frac{\epsilon_0}{1+\epsilon_0}$ for $\epsilon_0 > -1$ and our definition of $\zeta$ and $\eta$ we get
\begin{align*}
    k & \leq \frac{\zeta \cdot (1+\eps_0)}{\eps_0} \cdot \U_a \log m  \\
    & = \frac{\U_a}{\eta}
\end{align*}
as desired.
\end{proof}

We next prove the near-optimality of our solution.
\begin{lemma}\label{lem:ana2}
The pair $(f, \w)$ returned by \Cref{alg:mw} satisfies $(1-\epsilon)\sum_{a} \w_a \leq \sum_P f_P$.
\end{lemma}
\begin{proof}
Fix an iteration $i$ of the above while loop and let $\hat{f}$ be our lightest path blocker in this iteration. Let $k_i$ be $\st(\hat{f})$, let $\lambda_i$ be $\lambda$ at the start of this iteration and let $D_i := \sum_{a} \w_a$ be our total dual value at the start of this iteration. Notice that $\frac{1}{\lambda_i} \cdot \w$ is dual feasible and has cost $\frac{D_i}{\lambda_i}$ by \Cref{lem:lambLowerBound}. If $\beta$ is the optimal dual value then by optimality it follows that $\beta \leq \frac{D_i}{\lambda_i}$, giving us the upper bound on $\lambda_i$ of $\frac{D_i}{\beta}$. By how we update our dual, our bound on $\lambda_i$ and $(1+x)^{r} \leq 1 + xr$ for any $x \geq 0$ and $r \in (0,1)$ we have that
\begin{align*}
   D_{i+1} &= \sum_a (1+ \epsilon_0)^{\hat{f}(a)/\U_a} \cdot \w_a \cdot \U_a\\ &
   \leq \sum_a \left(1+ \frac{\epsilon_0 \hat{f}(a)}{\U_a} \right) \cdot \w_a \cdot \U_a \\
   &= D_i  + \epsilon_0 \sum_a  \hat{f}(a) \w_a\\
   & \leq D_i + \epsilon_0 (1 + 2 \eps_0) \cdot k_i \lambda_i\\
    &\leq D_i \left(1+\frac{(1+2\eps_0)\eps_0 \cdot k_i}{\beta} \right)\\
    & \leq D_i \cdot \exp\left(\frac{(1+2\eps_0)\eps_0 \cdot k_i}{\beta} \right).
\end{align*}
Let $T-1$ be the index of the last iteration of our algorithm; notice that $D_T$ is the value of $\w$ in our returned solution. Let $K := \sum_i k_i$. Then, repeatedly applying this recurrence gives us
\begin{align*}
    D_T &\leq D_0 \cdot \exp\left(\frac{(1+2\eps_0)\eps_0 \cdot K}{\beta} \right) \\
    &= \left(\frac{1}{m} \right)^{\zeta-1} \exp\left( \frac{(1+2\eps_0)\eps_0 \cdot K}{\beta} \right)
\end{align*}

On the other hand, we know that $\w$ is dual feasible when we return it, so it must be the case that $D_T \geq 1$. Combining this with the above upper bound on $D_T$ gives us $1 \leq \left(\frac{1}{m} \right)^{\zeta} \exp\left(\frac{(1+2\eps_0)\eps_0 \cdot K}{\beta} \right)$. Solving for $K$ and using our definition of $\zeta$ gives us
\begin{align*}
    \beta \log m \cdot \frac{\zeta-1}{(1 + 2 \eps_0) \cdot \eps_0}  &\leq K\\
    \beta \log m \cdot \frac{1}{\eps_0^2}  &\leq K.
\end{align*}
However, notice that $K\eta$ is the primal value of our solution so using our choice of $\eta$ and rewriting this inequality in terms of $K \eta$ by multiplying by $\eta = \frac{\eps_0}{(1 + \eps_0) \cdot \zeta} \cdot \frac{1}{\log m}$ and applying our definition of $\zeta = \frac{1+2 \eps_0}{\eps_0} + 1$ gives us
\begin{align}
    \frac{\beta}{\eps_0 \cdot (1+\eps_0) \cdot \zeta}&\leq K \eta \nonumber\\
    \frac{\beta}{(1+\eps_0)(1+3\eps_0)} &\leq K \eta.\label{eq:firstP}
\end{align}
Moreover, by our choice of $\eps_0 = \frac{\eps}{6}$ and the fact that $\frac{1}{1+x+x^2} \geq 1-x$ for $x \in (0,1)$ we get
\begin{align}
    1 - \eps &\leq \frac{1}{1 + \eps + \eps^2} \nonumber\\
    &\leq \frac{1}{(1+\frac{1}{2}\eps)^2} \nonumber\\
    &\leq \frac{1}{(1+3\eps_0)^2} \nonumber\\
    &\leq \frac{1}{(1+\eps_0)(1+3\eps_0)}. \label{eq:secondP}
\end{align}
Combining \Cref{eq:firstP} and \Cref{eq:secondP} we conclude that 
\begin{align*}
    (1-\eps) \cdot \beta \leq K \eta. 
\end{align*}\qedhere
\end{proof}

We conclude with our main theorem by proving that we need only iterate our algorithm $\tilde{O}\left(\frac{h}{\epsilon^4} \right)$ times.
\mainThm*
\begin{proof}
We use \Cref{alg:mw}. By \Cref{lem:ana1} and \Cref{lem:ana2} we know that our solution is feasible and $(1\pm \epsilon)$-optimal so it only remains to argue the runtime of our algorithm and that the returned flow decomposes in the stated way.

We argue that we must only run for $O\left(\frac{h \log^2 n}{\epsilon^4} \right)$ total iterations. Since $\lambda$ increases by a multiplicative $(1+\epsilon_0)$ after every $\Theta\left(\frac{h \log n}{\epsilon_0^2} \right)$ iterations and starts at at least $\left(\frac{1}{m}\right)^{\Theta(1/\eps_0)}$, it follows by \Cref{lem:lambLowerBound} that after $y \cdot \Theta \left(\frac{h \log n}{\epsilon_0^2} \right)$ total iterations the $h$-length distance between $S$ and $T$ is at least $(1 + \epsilon_0)^y \cdot \left(\frac{1}{m}\right)^{\Theta(1/\eps_0)}$. Thus, for $y \geq \Omega\left(\frac{\log_{1+\epsilon_0}m}{\epsilon_0}\right) = \Omega\left(\frac{\log n}{\epsilon_0^2} \right)$ we have that $S$ and $T$ are at least $1$ apart in $h$-length distance. Consequently, our algorithm must run for at most $O\left(\frac{h \log^2 n}{\epsilon_0^4}\right) = O\left(\frac{h \log^2 n}{\epsilon^4} \right)$ many iterations.

Our running time is immediate from the the bound of $O\left(\frac{h \log^2 n}{\epsilon^4} \right)$ on the number of iterations of the while loop and the running times given in \Cref{thm:pathBlockerAlg} for computing our $h$-length $(1+\epsilon_0)$-lightest path blocker.

Lastly, the flow decomposes in the stated way because we have at most $O\left(\frac{h \log^2 n}{\epsilon^4} \right)$ iterations and each $f_j$ is an integral $S$-$T$ flow by \Cref{thm:pathBlockerAlg}. Thus, our final solution is $\eta \cdot \sum_{j=1}^k f_j$ and $k = \tilde{O}\left(\frac{h}{\epsilon^4} \right)$. Likewise we have $|\supp(f)| \leq \tilde{O}(\frac{h^{10}}{\eps^7})$ for our parallel algorithm since we have $O\left(\frac{h \log^2 n}{\epsilon^4} \right)$ iterations and the fact that \Cref{thm:pathBlockerAlg} guarantees each $(1+\eps_0)$-lightest path blocker has support size at most $\tilde{O}(\frac{h^9}{\eps^3} \cdot |A|)$.
\end{proof}

\section{Application: Maximal and Maximum Disjoint Paths}\label{sec:disjointPaths}

In this section we show that our main theorem (\Cref{thm:main}) almost immediately gives deterministic CONGEST algorithms for many varieties of maximal disjoint path problems as well as essentially-optimal algorithms for many maximum disjoint path problems. In \Cref{sec:variants} we give the variants we study. In \Cref{sec:reduce} we observe that it suffices to solve the arc-disjoint directed variants of these problems. Lastly, we give our results for maximal and maximum disjoint path problems in \Cref{sec:maximalPaths} and \Cref{sec:maximumPaths} respectively where we observe in \Cref{sec:hardness} that our algorithms for the latter are essentially optimal.

\subsection{Maximal and Maximum Disjoint Path Variants}\label{sec:variants}

We consider the following maximal disjoint path variants.

\newcommand{\maximalVertexDisjointUndirected}{
\begin{quote}
    \textbf{Maximal Vertex-Disjoint Paths}: Given graph $G = (V, E)$, length constraint $h \geq 1$ and two disjoint sets $S, T \subseteq V$, find a collection of $h$-length vertex-disjoint $S$ to $T$ paths $\mcP$ such that any $h$-length $S$ to $T$ path shares a vertex with at least one path in $\mcP$.
\end{quote}
}

\newcommand{\maximalVertexDisjointDirected}{
\begin{quote}
    \textbf{Maximal Vertex-Disjoint Directed Paths}: Given digraph $D = (V, A)$, length constraint $h \geq 1$ and two disjoint sets $S, T \subseteq V$, find a collection of $h$-length vertex-disjoint $S$ to $T$ paths $\mcP$ such that any $h$-length $S$ to $T$ path shares a vertex with at least one path in $\mcP$.
\end{quote}
}

\newcommand{\maximalArcDisjointDirected}{
\begin{quote}
    \textbf{Maximal Arc-Disjoint Directed Paths}: Given digraph $D = (V, A)$, length constraint $h \geq 1$ and two disjoint sets $S, T \subseteq V$, find a collection of $h$-length arc-disjoint $S$ to $T$ paths $\mcP$ such that any $h$-length $S$ to $T$ path shares an arc with at least one path in $\mcP$.
\end{quote}
}

\maximalVertexDisjointUndirected

\maximalEdgeDisjointUndirected

\maximalVertexDisjointDirected

\maximalArcDisjointDirected

As discussed in \Cref{sec:contributions}, the existence of efficient deterministic algorithms for the above problems (specifically the maximal vertex-disjoint paths problem) in CONGEST was stated as an open question by \citet{chang2020deterministic} and the lack of these algorithms is a major barrier to simple deterministic constructions of expander decompositions.

\newcommand{\maximumVertexDisjointUndirected}{
\begin{quote}
    \textbf{Maximum Vertex-Disjoint Paths}: Given graph $G = (V, E)$, length constraint $h \geq 1$ and disjoint sets $S, T \subseteq V$, find a max cardinality collection of $h$-length vertex-disjoint $S$ to $T$ paths.
\end{quote}
}

\newcommand{\maximumVertexDisjointDirected}{
\begin{quote}
    \textbf{Maximum Vertex-Disjoint Directed Paths}: Given digraph $D = (V, A)$, length constraint $h \geq 1$ and disjoint sets $S, T \subseteq V$, find a max cardinality collection of $h$-length vertex-disjoint $S$ to $T$ paths.
\end{quote}
}

\newcommand{\maximumArcDisjointDirected}{
\begin{quote}
    \textbf{Maximum Arc-Disjoint Directed Paths}: Given digraph $D = (V, A)$, length constraint $h \geq 1$ and disjoint sets $S, T \subseteq V$, find a max cardinality collection of $h$-length arc-disjoint $S$ to $T$ paths.
\end{quote}
}

We consider the following maximum disjoint path variants.

\maximumVertexDisjointUndirected

\maximumEdgeDisjointUndirected

\maximumVertexDisjointDirected

\maximumArcDisjointDirected

\subsection{Reducing Among Variants}\label{sec:reduce}

We begin by observing that the arc-disjoint directed paths problem is the hardest of the above variants and so it will suffice to solve this problem. The reductions we use are illustrated in \Cref{fig:red}.
\begin{figure}
    \centering
    \begin{subfigure}[b]{0.32\textwidth}
        \centering
        \includegraphics[width=\textwidth,trim=50mm 100mm 0mm 100mm, clip]{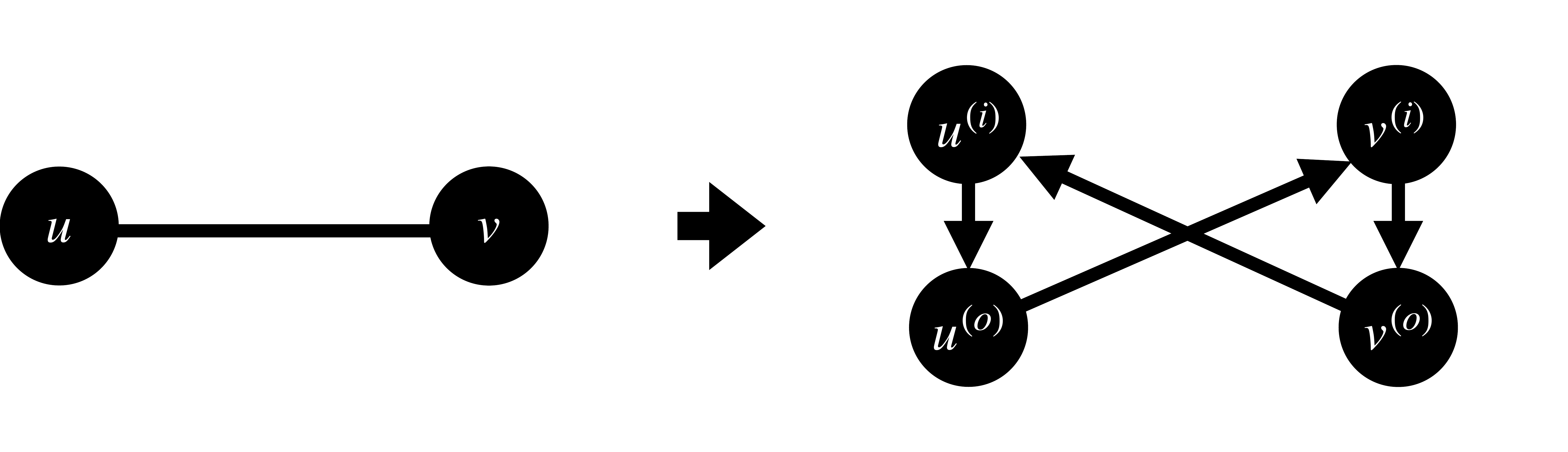}
        \caption{Vertex-disjoint paths.}\label{sfig:red1}
    \end{subfigure}    \hfill
    \begin{subfigure}[b]{0.32\textwidth}
        \centering
        \includegraphics[width=\textwidth,trim=50mm 100mm 0mm 100mm, clip]{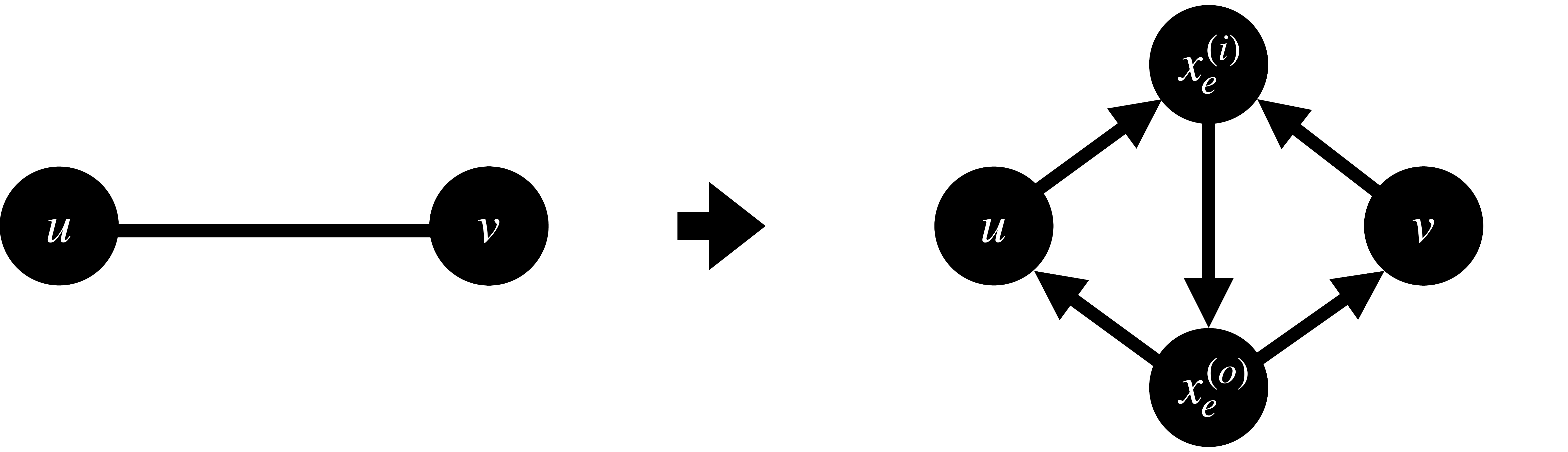}
        \caption{Edge-disjoint paths.}\label{sfig:red2}
    \end{subfigure}    \hfill
    \begin{subfigure}[b]{0.32\textwidth}
        \centering
        \includegraphics[width=\textwidth,trim=50mm 100mm 0mm 100mm, clip]{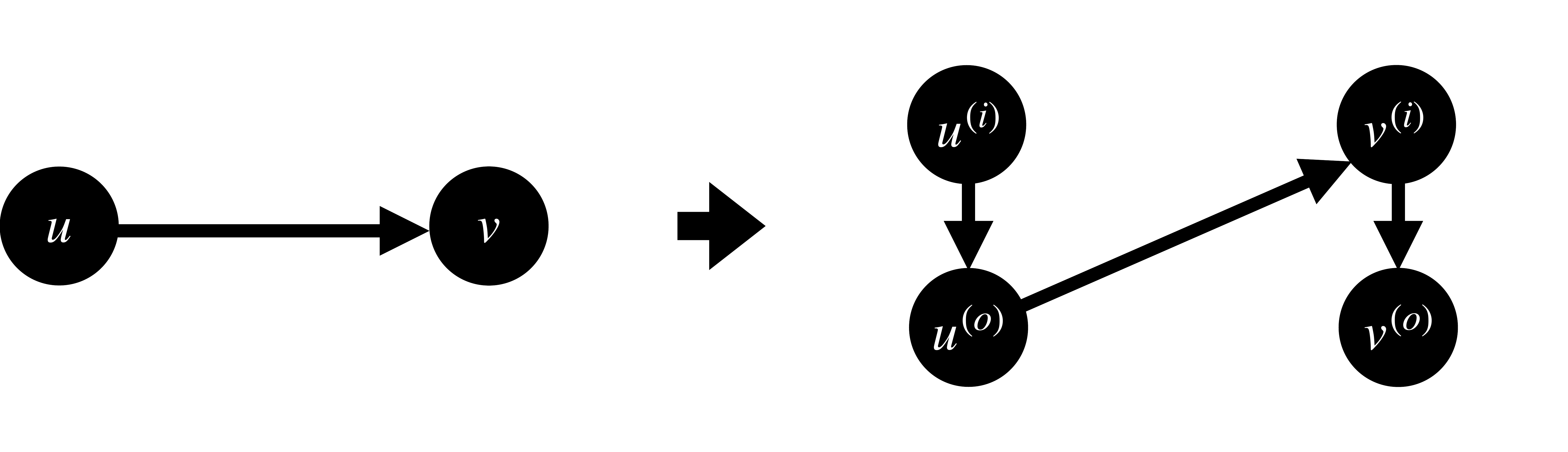}
        \caption{Vertex-disjoint directed paths.}\label{sfig:red3}
    \end{subfigure}
    \caption{Illustration of our reduction on a single edge or arc between $u$ and $v$ for reducing maximal or maximum vertex-disjoint paths, edge-disjoint paths or vertex-disjoint directed paths to arc-disjoint directed paths.}\label{fig:red}
\end{figure}

\begin{lemma}\label{lem:reduceMaximalPaths}
If there is a deterministic algorithm for maximal arc-disjoint directed paths in CONGEST running in time $T$ then there are deterministic CONGEST algorithms for maximal vertex-disjoint paths, edge-disjoint paths and vertex-disjoint directed paths all running in time $O(T)$. 

Likewise, if there is a deterministic (resp. randomized) parallel with $m$ processors or CONGEST algorithm for maximum arc-disjoint directed paths in CONGEST running in time $T$ with approximation ratio $\tilde{O}(h)$ then there are deterministic (resp. randomized) parallel with $m$ processors and CONGEST algorithms for maximum vertex-disjoint paths, edge-disjoint paths and vertex-disjoint directed paths all running in time $O(T)$ with approximation ratio $\tilde{O}(h)$.
\end{lemma}
\begin{proof}
We reduce each of maximal vertex-disjoint paths, maximal edge-disjoint paths and maximal vertex-disjoint directed paths to maximal arc-disjoint directed paths and do the same for the maximum variants of these problems.

\paragraph{Reducing from maximal/maximum vertex-disjoint paths.} Consider an instance of maximal or maximum vertex-disjoint paths on graph $G = (V, E)$ with length constraint $h$ and vertex sets $S$ and $T$. We create a digraph $D = (V', A)$ as follows:
\begin{itemize}
    \item \textbf{Vertices:} $V'$ is constructed as follows: for each $v \in V$ we add to $V'$ vertex $v^{(i)}$ and $v^{(o)}$. 
    \item \textbf{Arcs:} For each $v \in V$ we add an arc from $v^{(i)}$ to $v^{(o)}$. Furthermore, for each $e = \{u,v\} \in E$ we add to $A$ the arcs $(u^{(o)}, v^{(i)})$ and $(v^{(o)}, u^{(i)})$.
\end{itemize}
A collection of arc-disjoint paths in $D$ from $S' = \{s^{(i)} : s \in S \}$ to $T' = \{t^{(o)} : t \in T \}$ with length constraint $2h-1$ uniquely corresponds to an equal cardinality collection of $S$-$T$ vertex-disjoint paths in $G$ with length constraint $h$. Thus, an $\tilde{O}(h)$ approximation on $D$ for the maximum $S'$-$T'$ arc-disjoint directed paths problem gives an $\tilde{O}(h)$ approximation for the maximum vertex-disjoint paths problem on $G$. Likewise, a maximal collection of arc-disjoint $S'$-$T'$ paths on $D$ with length constraint $2h-1$ corresponds to a maximal collection of vertex-disjoint $S$-$T$ paths with length constraint $h$. Lastly, a $T$-time CONGEST algorithm on $D$ can be simulated on $G$ in time $O(T)$ since each $v \in V$ can simulate $v^{(o)}$ and $v^{(i)}$. 

\paragraph{Reducing from maximal/maximum edge-disjoint paths.} Consider an instance of maximal or maximum edge-disjoint paths on graph $G = (V, E)$ with length constraint $h$ and vertex sets $S$ and $T$. We create a digraph $D = (V', A)$ as follows:
\begin{itemize}
    \item \textbf{Vertices:} $V'$ consists of $V$ along with two vertices for each edge $e$, namely $x_{e}^{(i)}$ and $v_{e}^{(o)}$ for each $e \in E$.
    \item \textbf{Arcs:} For each $e \in \{u,v\} \in E$ we add to $A$ an arc from $x_{e}^{(i)}$ to $x_{e}^{(o)}$ as well as an arc from $u$ and $v$ to $x_{e}^{(i)}$ and an arc from $x_{e}^{(o)}$ to $u$ and $v$.
\end{itemize}
A collection of arc-disjoint $S$-$T$ paths in $D$ with length constraint $3h$ uniquely corresponds to an equal cardinality collection of $S$-$T$ edge-disjoint paths in $G$ with length constraint $h$. Thus, an $\tilde{O}(h)$ approximation on $D$ for the maximum $S$-$T$ arc-disjoint directed paths problem gives an $\tilde{O}(h)$ approximation for the maximum edge-disjoint paths problem on $G$. Likewise, a maximal collection of arc-disjoint $S$-$T$ paths on $D$ with length constraint $3h$ corresponds to a maximal collection of edge-disjoint $S$-$T$ paths with length constraint $h$ on $G$. Lastly, a $T$-time CONGEST algorithm on $D$ can be simulated on $G$ in time $O(T)$ since the endpoints of $e \in E$ can simulate $x_e^{(i)}$ and $x_e^{(o)}$ with constant overhead.

\paragraph{Reducing from maximal/maximum vertex-disjoint directed paths.}

Consider an instance of maximal or maximum vertex-disjoint directed paths on graph $D = (V, A)$ with length constraint $h$ and vertex sets $S$ and $T$. We create a digraph $D' = (V', A')$ as follows:
\begin{itemize}
    \item \textbf{Vertices:} $V'$ consists of vertices $v^{(o)}$ and $v^{(i)}$ for each $v \in V$.
    \item \textbf{Arcs:} For each $v \in V$ we add to $A'$ the arc $(v^{(i)}, v^{(o)})$. For each arc $a = (u,v) \in A$ we add to $A'$ the arc $(u^{(o)}, v^{(i)})$.
\end{itemize}
A collection of arc-disjoint paths in $D'$ from $S' = \{s^{(i)} : s \in S \}$ to $T' = \{t^{(o)} : t \in T \}$ with length constraint $2h-1$ uniquely corresponds to an equal cardinality collection of $S$-$T$ vertex-disjoint  paths in $D$ with length constraint $h$. Thus, an $\tilde{O}(h)$ approximation on $D'$ for the maximum $S'$-$T'$ arc-disjoint directed paths problem gives an $\tilde{O}(h)$ approximation for the maximum $S$-$T$ vertex-disjoint directed paths problem on $D$. Likewise, a maximal collection of arc-disjoint $S'$-$T'$ paths on $D'$ with length constraint $2h-1$ corresponds to a maximal collection of vertex-disjoint $S$-$T$ paths with length constraint $h$ on $D$. Lastly, a $T$-time CONGEST algorithm on $D'$ can be simulated on $D$ in time $T$ each $v \in V$ can simulate $v^{(i)}$ and $v^{(o)}$.
\end{proof}

\subsection{Maximal Disjoint Path Algorithms}\label{sec:maximalPaths}

We now observe that our length-constrained flow algorithms allow us to solve maximal arc-disjoint directed paths and therefore all of the above variants efficiently.
\begin{theorem}
There are deterministic CONGEST algorithms for maximal vertex-disjoint paths,  edge-disjoint paths,  vertex-disjoint directed paths and arc-disjoint directed paths running in time $\tilde{O}\left(h^{18}  + h^{17} \cdot (\rho_{CC})^{10} \right)$.
\end{theorem}
\begin{proof}
By \Cref{lem:reduceMaximalPaths}, it suffices to show that maximal arc-disjoint directed paths can be solved in time $\tilde{O}\left(h^{18}  + h^{17} \cdot (\rho_{CC})^{10} \right)$. We proceed to do so on digraph $D$ with length constraint $h$ and vertex sets $S$ and $T$ for the rest of this proof.


Specifically, we repeat the following until no path between $S$ and $T$ consists of $h$ or fewer edges. Apply \Cref{thm:main} to compute a $(1-\epsilon)$-approximate $h$-length $S$-$T$ flow $f$ in $D$ for $\epsilon = .5$ (any constant would suffice) with unit capacities. By the properties of $f$ as guaranteed by \Cref{thm:main}, we have that $f = \eta \cdot \sum_{j=1}^k f_j$ for $\eta = \tilde\Theta(1)$ and $k = \tilde{O}\left(h \right)$ where each $f_j$ is an integral flow. For each vertex $v$ we let $f_j^{(v)}$ be $f_j$ restricted to its flow paths out of $v$ and let $f_{j^*}^{(v)} := \argmax_{f_j^{(v)}} \st(f_j^{(v)})$. Then, we let $f_{j^*} := \sum_v f_{j^*}^{(v)}$ (notice that we cannot simply define $f_{j^*}$ as $\argmax_{f_j} \st(f_j)$ since we cannot compute $\st(f_j)$ efficiently in CONGEST because $D$ may have diameter much larger than $h$). Observe that since $f_{j^*}$ is integral and $h$-length, it exactly corresponds to an arc-disjoint collection of $S$-$T$ paths $\mcP'$ in $D$ each of which consists of at most $h$ edges. We add $\mcP'$ to $\mcP$, delete from $D$ any arc incident to a path of $\mcP'$ and continue to the next iteration.

As the above algorithm removes at least one path from $S$ to $T$ each time, it clearly terminates with a feasible solution for the maximal arc-disjoint directed paths problem.

Stronger, though, we claim that we need only iterate the above $\tilde{O}(h)$-many times until $S$ and $T$ are disconnected. Specifically, fix one iteration and let $\mcP^*$ be the collection of vertex-disjoint paths from $S$ to $T$ of maximum cardinality at the beginning of this iteration. By the $(1 - \epsilon)$-optimality of our flow and an averaging argument we have that $\st(f_{j^*}) \geq \tilde{\Omega}\left(\frac{1}{h}\right) \cdot |\mcP^*|$ which is to say that $|\mcP'| \geq \tilde{\Omega}\left(\frac{1}{h}\right) \cdot |\mcP^*|$. However, it follows that after $\tilde{\Theta}(h)$-many iterations for a large hidden constant we must at least halve $|P^*|$ since otherwise we would have computed a collection of vertex-disjoint $S$-$T$ paths whose cardinality is larger than the largest cardinality of any set of vertex-disjoint $S$-$T$ paths. Since initially $|P^*| \leq n$, it follows that after iterating the above $\tilde{O}(h)$-many times we have reduced $|P^*|$ to $0$ which is to say we have solved the maximal arc-disjoint directed paths problem.

Our running time is immediate from \Cref{thm:main} and the above bound we provide on the number of required iterations of $\tilde{O}(h)$ as well as the fact that each vertex can easily compute $f_{j^*}^{(v)}$ and $\mcP$ deterministically in parallel or CONGEST time $\tilde{O}(h)$ since our flows are $h$-length.
\end{proof}
Applying the fact that it is known that $\rho_{CC} \leq 2^{O(\sqrt{\log n})}$ (see \Cref{sec:cycCovers}), the above gives deterministic CONGEST algorithms running in time $\tilde{O}(\poly(h) \cdot 2^{O(\sqrt{\log n})})$. If $\rho_{CC}$ where improved to be poly-log in $n$ then we would get a $\tilde{O}(\poly(h))$ running time.

\subsection{Maximum Disjoint Path Algorithms}\label{sec:maximumPaths}
Lastly, we observe that our length-constrained flow algorithms allow us to $\tilde{O}(h)$-approximate maximum arc-disjoint directed paths and therefore all of the above variants efficiently.

\begin{theorem}\label{thm:maxPath}
    There are $\tilde{O}(h)$-approximation algorithms for maximum vertex-disjoint paths, edge-disjoint paths,  vertex-disjoint directed paths and arc-disjoint directed paths running in:
    \begin{itemize}
        \item Deterministic parallel time $\tilde{O}( h^{17})$ with $m$ processors;
        \item Randomized CONGEST time $\tilde{O}( h^{17})$ with high probability;
        \item Deterministic CONGEST time $\tilde{O}\left( h^{17}  + h^{16} \cdot (\rho_{CC})^{10} \right)$.
    \end{itemize}
\end{theorem}
\begin{proof}
By \Cref{lem:reduceMaximalPaths}, it suffices to provide a $\tilde{O}(h)$-approximate algorithm for maximum arc-disjoint directed paths with the stated running times. We do so for the rest of this proof. Let the input be digraph $D = (V, A)$ with length constraint $h \geq 1$ and disjoint sets $S, T \subseteq V$.

We apply \Cref{thm:main} to compute an $\epsilon$-approximate $h$-length constrained flow $f$ in $D$ for $\epsilon = .5$ (any constant would suffice) and capacities $\U_a = 1$ for every $a$. By the properties of $f$ as guaranteed by \Cref{thm:main}, we have that $f = \eta \cdot \sum_{j=1}^k f_j$ for $\eta = \Theta(1)$ and $k = \tilde{O}\left(h \right)$ where each $f_j$ is an integral flow. For each vertex $v$ we let $f_j^{(v)}$ be $f_j$ restricted to its flow paths out of $v$ and let $f_{j^*}^{(v)} := \argmax_{f_j^{(v)}} \st(f_j^{(v)})$. Then, we let $f_{j^*} := \sum_v f_{j^*}^{(v)}$. Observe that since $f_{j^*}$ is integral and $h$-length, it exactly corresponds to an arc-disjoint collection of paths $\mcP$ in $D$ each of which consists of at most $h$ edges. We return $\mcP$ as our solution.

Letting $\mcP^*$ be the optimal solution to the input problem we have by $k = \tilde{O}(h)$ and an averaging argument that
\begin{align*}
|\mcP| = \st(f_{j^*}) \geq \tilde{\Omega}\left(\frac{1}{h}\right) \cdot |\mcP^*|
\end{align*}
and so our solution is $\tilde{\Omega}(\frac{1}{h})$-approximate.

For our running time, observe that each vertex can easily compute $f_{j^*}^{(v)}$ and $\mcP$ deterministically in parallel or CONGEST time $\tilde{O}(h)$ since our flows are $h$-length. Thus, our running time is dominated by \Cref{thm:main}.
\end{proof}

\subsection{On the Hardness of Maximum Disjoint Paths}\label{sec:hardness}

\citet{guruswami2003near} give hardness results for a variety of length-constrained maximum disjoint path problems. In their work they state hardness of approximation result in terms of $m$, the number of edges in the graph. In the following we restate these results but in terms of $h$, the length-constraint. 
\begin{theorem}[Adaptation of Theorem 1 of \citet{guruswami2003near}]
Assume the strong exponential time hypothesis (SETH). Then there does not exist a polynomial-time $O(h)$-approximation algorithm solving the maximum arc-disjoint directed paths problem for instances where $h = \Omega(\log n)$.
\end{theorem}
\noindent Observe that it follows that assuming SETH, the parallel algorithm in \Cref{thm:maxPath} is optimal up to poly-logs.

\section{Application: Simple Distributed Expander Decompositions}\label{sec:expander}
In this section, we explain how our maximal disjoint path algorithm can significantly simplify the distributed deterministic expander decomposition of \citet{chang2020deterministic}.

The key algorithmic primitive of \citep{chang2020deterministic} in their
distributed deterministic expander decomposition is their Lemma D.8. Instead
of computing maximal bounded-hop disjoint paths, they were only be able to compute
a set of paths that are ``nearly maximal''. The formal statement
is as follows:
\begin{lemma}
[Nearly maximal disjoint paths (Lemma D.8 of \cite{chang2020deterministic}] Consider a graph
$G=(V,E)$ of maximum degree $\Delta$. Let $S\subseteq V$ and $T\subseteq V$
be two subsets. There is an $O(d^{3}\beta^{-1}\log^{2}\Delta\log n)$-round
deterministic algorithm that finds a set $P$ of $S-T$ vertex-disjoint
paths of length at most $d$, together with $a$ vertex set $B$ of
size at most $\beta|V\setminus T|<\beta|V|$, such that any $S-T$
path of length at most $d$ that is vertex-disjoint to all paths in
$P$ must contain a vertex in $B$.
\end{lemma}

The set $P$ from the lemma is nearly maximal in the sense that if $B$ is deleted
from $G$, then $P$ would be maximal. However, we can see that there
might possibly be many additional disjoint paths that go through $B$.
This set $B$ complicates all of their later algorithmic steps. 

The high-level summary of the issue is that all their flow primitives
that are based on Lemma D.8 must work with source/sink sets that are
very big only. Otherwise, the guarantee becomes meaningless or the
running time becomes very slow. 

Now, we explain in more details. Given two sets $S$ and $T$ where
$|S|\le|T|$, normally if the matching player from the cut-matching
game does not return a sparse cut, then it returns an embedding of
a matching where every vertex in $S$ is matched to some vertex in
$T$. However, in Lemma D.9 of \cite{chang2020deterministic}, the matching player based on
Lemma D.8 may return an embedding that leaves as many as $\approx\beta|V\setminus T|$
vertices in $S$ unmatched. This is called the ``left-over'' set.
We think of $\beta\ge1/n^{o(1)}$ as the round complexity of Lemma
D.8 is proportional to $\beta^{-1}$. Therefore, it is only when $|S|,|T|\ge2\beta|V|\ge|V|/n^{o(1)}$
that Lemma D.9 in \cite{chang2020deterministic} may give some meaningful guarantee, yet
this is still weaker than normal.

The same issue holds for their multi-commodity version of the matching
player (i.e.~Lemma D.11 of \cite{chang2020deterministic}). For the same reasoning, the lemma is meaningful
only when the total number of source and sink is at least $\Omega(\beta|V|)$.
The issue propagates to their important subroutine (Theorem 4.1 of \cite{chang2020deterministic}) for
computing most balanced sparse cut. The guarantee holds when only
the returned cut $C$ is such that $|C|\ge\Omega(\beta|V|)$. At the
end, they managed to obtain an deterministic expander decomposition
(just treat the edges incident to the left-over part as inter-cluster
edges at the end). However, they need to keep track of this left-over
parameter from the first basic primitive until the end result. 

In contrast, in their randomized algorithm for computing expander decomposition, this issues does not appear anyway because of the randomized maximal disjoint path algorithm. Therefore, by plugging in our deterministic maximal disjoint path algorithm into the expander decomposition of \cite{chang2020deterministic}, all these issue  will be resolved immediately.

\section{Application: $(1-\epsilon)$-Approximate Distributed Bipartite $b$-Matching}\label{sec:bmatching}
In this section we give the first efficient $(1-\epsilon)$-approximate CONGEST algorithms for maximum cardinality bipartite $b$-matching. In fact, our results are for the slightly more general edge-capacitated maximum bipartite $b$-matching problem, defined as follow.

\begin{quote}
    \textbf{Edge-Capacitated Maximum Bipartite $b$-Matching}: Given bipartite graph $G = (V, E)$, edge capacities $\U$ and function $b : V \to \mathbb{Z}_{ > 0}$ compute an integer $x_e \in [0, \U_e]$ for each $e \in E$ maximizing $\sum_e x_e$ so that for each $v \in V$ we have $\sum_{e \in \delta(v)} x_e \leq b(v)$.
\end{quote}
Notice that the case where $b(v) = 1$ for every $v$ is just the classic maximum cardinality matching problem. ``$b$-matching'' seems to refer to two different problems in the literature depending on whether edges can be chosen with multiplicity: either it is the above problem where $\U_e = 1$ for every $e \in E$ or it is the above problem where $\U_e = \max_v b_v$ for each $e \in E$. Our algorithms will work for both of these variants since they solve the above problem which generalizes both of these problems.

The following theorem summarizes our main result for bipartite $b$-matching in CONGEST. Again, recall that $\rho_{CC}$ is defined in \Cref{dfn:rhoCC} and is known to be at most $2^{O(\sqrt{\log n})}$.

\begin{theorem}
There is a deterministic $(1-\epsilon)$-approximation for edge-capacitated maximum bipartite $b$-matching running in CONGEST time $\tilde{O}\left(\frac{1}{\eps^9} + \frac{1}{\eps^7} \cdot (\rho_{CC})^{10} \right)$.
\end{theorem}
\begin{proof}
Our algorithm works in two steps. First, we reduce edge-capacitated $b$-matching to length-constrained flow and use our length constrained flow algorithm to efficiently compute a fractional flow. Then, we apply the flow rounding technology we developed in \Cref{sec:flowRounding} to round this flow to an integral flow which, in turn, corresponds to an integral $b$-matching.

More formally our algorithm is as follows. Suppose we are given an instance of edge-capacitated $b$-matching on bipartite graph $G = (V,E)$. Let $L$ and $R$ be the corresponding bipartition of vertices of $G$. We construct the following instance of length-constrained flow on digraph $D = (V',A)$ with $h=3$ as follows. Each $v \in V$ has two copies $v^{(i)}$ and $v^{(o)}$ in $V'$. We add arc $(v^{(i)}, v^{(o)})$ to $A$ with capacity $b(v)$. If $\{u, v\} \in E$ where $u \in L$ and $v \in R$ then we add arc $(u^{(o)}, v^{(i)})$ with capacity $\U_e$ to $A$. Lastly, we let $S = \{u^{(i)} : u \in L \}$, $T = \{v^{(o)} : v \in R\}$ and the length of each arc in $D$ be $1$. Next, we apply \Cref{thm:main} to compute a $(1-\eps_1)$-approximate maximum $3$-length $S$-$T$ flow $f$ on $D$ for some small $\eps_1$ to be chosen later. Since $D$ is a $3$-layer $S$-$T$ DAG we may interpret this as a (non-length-constrained) flow where the flow value on arc $a$ is $f(a)$. 

We then apply \Cref{lem:detRounding} to this non-length-constrained flow to get integral $S$-$T$ flow $f'$ satisfying $\st(f') \geq (1-\eps_2) \cdot \st(f)$ for some small $\eps_2$ to be chosen later. We return as our solution the $b$-matching which naturally corresponds to $f'$. Namely, if $e = \{u, v\}$ then since $f'$ is integral it assigns arc $(u^{(o)}, v^{(i)})$ a value in $\{0, 1, \ldots, \U_e \}$. We let $x_e$ be this value for $e = \{u, v\}$ and we return as our $b$-matching solution $\{x_e\}_e$.

$f'$ is a $(1-\eps_1)(1-\eps_2)$-approximate maximum $S$-$T$ flow. Letting $\OPT$ be the value of the optimal $b$-matching solution, it is easy to see that the maximum $S$-$T$ flow has value $\OPT$ and so the solution we return has value at least $(1-\eps_1)(1-\eps_2) \cdot \OPT$. Letting $\eps_1 = \eps_2 = \Theta(\eps)$ for an appropriately small hidden constant we get that $(1-\eps_1)(1-\eps_2) \cdot \OPT \geq (1-\eps) \cdot \OPT$.

Lastly, we argue our running time. Our running time is dominated by one call to \Cref{thm:main} with $\eps_1 = \Theta(\eps)$ which takes $\tilde{O}\left(\frac{1}{\eps^9} + \frac{1}{\eps^7} \cdot (\rho_{CC})^{10} \right)$ and one call to \Cref{lem:detRounding} with $\eps_2 = \Theta(\eps)$ which takes $\tilde{O}(\frac{1}{\eps^5} \cdot (\rho_{CC})^{10})$. Combining these running times gives the overall running time of our algorithm.
\end{proof}

\section{Application: Length-Constrained Cutmatches}\label{sec:cutMatches}
As it captures low-latency communication subject to bandwidth constraints, the problem of computing low-congestion $h$-length paths between two set of nodes $S$ and $T$ occurs often in network optimization. 

In this section we give algorithms that find a low-congestion $h$-length collection of paths between two sets of nodes and certify that there is no low-congestion way of extending the current collection of paths with a moving cut. Such a construction is called a length-constrained cutmatch. A recent work \cite{haeuplerExpander2022} uses the algorithms we give for cutmatches to give the first efficient constructions of a length-constrained version of expander decompositions. These constructions were then used to give the first distributed CONGEST algorithms for many problems including MST, $(1+\epsilon)$-min-cut and $(1+\epsilon)$-lightest paths that are guaranteed to run in sub-linear rounds as long as such algorithms exist on the input network. 

We now formalize cutmatches. In what follows, for a vertex subset $W \subseteq V$ we let $\U^+(W) = \sum_{v \in W}\sum_{a \in \delta^+(v)}\U_a$ and $\U^-(W) = \sum_{v \in W} \sum_{a \in \delta^-(v)}\U_a$. We also let $\delta^{\pm}(S,T) := \bigcup_{v \in S} \delta^+(v) \cup \bigcup_{v \in T} \delta^-(T)$. Note that throughout this section we assume that each $S$-$T$ path in the support of an $h$-length flow contains exactly one vertex from $S$ and one vertex from $T$ (this is without loss of generality since any such flow can be made to satisfy this property without changing its value).

\begin{definition}[$h$-Length Cutmatch]\label{def:cutmatch}
Given digraph $D = (V,A)$ with capacities $\U$ and lengths $\l$, an $h$-length $\phi$-sparse cutmatch of congestion $\gamma$ between disjoint node sets $S,T \subseteq V$ consists of:
\begin{itemize}
    \item An integral $h$-length $S$-$T$ flow $f$ in $D$ with capacities $\{\U_a\}_{a \in \delta^{\pm}(S,T)} \cup \{\gamma \cdot \U_a\}_{a \not \in \delta^{\pm}(S,T)}$ and lengths $\l$;
    \item A moving cut $\w$ of value $ \sum_a \w_a \cdot \U_a \leq \phi \left( \U^+(S) - \st(f)\right)$ such that $d_{\l'}(S,T) > h$ where
    \begin{align*}
            \l'_a := 
            \begin{cases}
                h+1 & \text{if $a \in \delta^\pm(S,T)$ and $f(a) = \U_a$}\\
                \l_a + h \cdot \w_a & \text{otherwise}
        \end{cases}
    \end{align*}
\end{itemize}
\end{definition}

Our main theorem of this sections shows how to efficiently compute length-constrained cutmatches.
\begin{restatable}{thm}{cutMatchTheorem}\label{lem:matchcut-old-purecut}
    Suppose we are given a digraph $D = (V,A)$ with capacities $\U$ and lengths $\l$. There is an algorithm that, given two node sets $S,T \subseteq V$, $h \geq 1$ and $\phi \leq 1$, outputs an $h$-length $\phi$-sparse cutmatch $(\hat{f}, \hat{w})$ of congestion $\gamma$ between $S$ and $T$, where $\gamma=\tilde{O}(\frac{1}{\phi})$. This algorithm runs in:
\begin{enumerate}
    \item Deterministic parallel time $\tilde{O}( h^{17})$ with $m$ processors where $|\supp(\hat{f})| \leq \tilde{O}(h^{10} \cdot |A|)$;
    \item Randomized CONGEST time $\tilde{O}(h^{17})$ with high probability.
\end{enumerate}
\end{restatable}

Before moving onto details of the algorithm, we give a high level description of how we prove the above. Our proof is based on a structural result which may be interesting in its own right. This structural result shows that one can always substantially reduce the value of an optimal length-constrained flow through one of two operations. Specifically, given an optimal flow and certifying moving cut pair $(f,\w)$, either:
\begin{enumerate}
    \item The flow $f$ covers most of the cut mass of $\w$ on arcs in $\delta^\pm(S,T)$. Consequently, if we reduce the capacities of arcs in $\delta^\pm(S,T)$ by the amount of flow that $f$ sends over them, then we substantially reduce the cost of $\w$. Since $\w$ is a feasible moving cut this, in turn, multiplicatively reduces the value of the optimal flow by a constant; or
    \item One can apply the cut $\w$ on arcs not in $\delta^\pm(S,T)$ (increasing the length of $a \not \in \delta^\pm(S,T)$ by about $h \cdot \w_a$) to reduce the optimal flow by a constant.
\end{enumerate}
This result is formalized by \Cref{lem:cutmatchTwoCases} in \Cref{sec:structural}; there, we show that this holds even for approximately-optimal length-constrained flow, moving cut pairs.

Our algorithm for length-constrained cutmatches scales the capacities of all arcs in $A \setminus \delta^\pm(S,T)$ up by about $\frac{1}{\phi}$, computes a series of flows and moving cuts and then adds to our cutmatch from the computed flow or the computed moving cut depending on which of the two operations reduces the optimal value by a constant. This allows us to compute a cutmatch because: (1) scaling back down these capacities guarantees that the moving cut we compute is sufficiently cheap; and (2) we only have to do the above $\tilde{O}(1)$-many times since each time we reduce the optimal value by a multiplicative constant which, in turn, allows us to argue our low congestion. We note, however, that the fact that we must scale capacities prevents us from using our deterministic flow algorithms. Implementing this in CONGEST requires using a sparse neighborhood cover. 

The above strategy is slightly complicated by the fact that we would like our flows in our cutmatches to be integral but each flow we compute is fractional. Crucially, however, by the properties of the flows we compute, if we compute an $O(1)$-approximate flow, then scale up this flow by $\Theta(1)$ the result becomes integral while increasing congestion by at most an $O(1)$ factor. In \Cref{sec:decongF} we show how to eliminate this congestion on arcs in $\delta^\pm(S,T)$ while only reducing the value of our flow by a bounded amount. We do not need to resolve this extra congestion on arcs not in $\delta^\pm(S,T)$ since our cutmatches may have large congestion on such arcs. We describe this in more detail in \Cref{sec:decongF}.

\subsection{Flow Mostly Covers $\delta^\pm(S,T)$ or Can Reduce Optimal by Cutting $A \setminus \delta^\pm(S,T)$}\label{sec:structural}
In this section we show the main structural result on which our cutmatch algorithm relies: given a flow, moving cut pair $(f,\w)$, either $f$ covers most of $\w$ on arcs in $\delta^\pm(S,T)$ or $\w$ can be applied to arcs in $A \setminus \delta^\pm(S,T)$ to reduce the optimal flow value by a constant.

The sense of $f$ covering $\w$ will makes use of the following notion of saturated arcs.
\begin{definition}[Saturated Arcs]
     Let $(f,\w)$ be an $h$-length $S$-$T$ flow, moving cut pair. We say that arc $a \in A$ is $c$-saturated for $c \in [0,1]$ with respect to $(f, \w)$ if
     \begin{align*}
         c \cdot \U_a \leq f(a).
     \end{align*}
\end{definition}

The following simple helper lemma shows that any near-optimal length-constrained flow, moving cut pair must be such that most of $\w$'s mass lies on saturated arcs.
\begin{lemma}[$h$-Length Flows Saturate Moving Cuts]\label{lem:satCuts}
    Let $(f,\w)$ be a $(1 \pm \eps)$-approximate $h$-length $S$-$T$ flow, moving cut pair. Fix any $c_1 \in [0,1]$ and let $c_2 \in [0,1]$ be
    \begin{align*}
        c_2 := \sum_{c_1\text{-saturated } a} \U_a \cdot \w_a \bigg/ \sum_a \U_a \cdot \w_a.
    \end{align*}
    Then $c_2 \geq 1 - \frac{\eps}{1-c_1}$.
\end{lemma}
\begin{proof}
The proof is by a simple averaging argument. Recall that the fact that $(f,w)$ is $(1 \pm \eps)$-approximate means 
\begin{align}
\left(1-\eps \right) \cdot \sum_a \U_a \cdot \w_a \leq \st(f).\label{eq:awe}
\end{align}


Next, applying the feasibility of $\w$, $f(a) \leq \U_a$ for every $a$ and the definition of $c_1$-saturated arcs and $c_2$, we have
\begin{align*}
    \st(f) &= \sum_{P} f_P \\
    &\leq \sum_{P} f_P \sum_{a \in P} \w_a\\
    &= \sum_{a} f(a) \cdot \w_a\\
    &= \sum_{c_1 \text{-saturated }a} f(a) \cdot \w_a  + \sum_{\text{not }c_1 \text{-saturated }a} f(a) \cdot \w_a\\
    &\leq \sum_{c_1 \text{-saturated }a} \U_a \cdot \w_a  + c_1 \cdot \sum_{\text{not }c_1 \text{-saturated }a} \U_a \cdot \w_a \\
    &= (c_2 + (1-c_2)\cdot c_1) \cdot \sum_{a} \U_a \cdot \w_a.
\end{align*}
and so $\st(f) < (c_2 + (1-c_2)\cdot c_1) \cdot \sum_{a} \U_a \cdot \w_a$ which when combined with \Cref{eq:awe} implies $c_2 + (1-c_2) \cdot c_1 \geq 1 - \eps$.
\end{proof}
We now show the main structural result of this section.
\begin{lemma}\label{lem:cutmatchTwoCases}
    Let $(f,\w)$ be a $(1 \pm  \eps)$-approximate $h$-length $S$-$T$ flow, moving cut pair in digraph $D = (V,A)$ with capacities $\U$ and lengths $\l$. Then either:
    \begin{enumerate}
        \item \textbf{Flow Mostly Covers Cut on $\delta^\pm(S,T)$:} $\sum_{a \in A'} \U_a \cdot \w_a \geq \left(\frac{1}{2}-3\eps \right) \cdot \sum_{a} \U_a \cdot \w_a$ where $A' = \{a \in \delta^\pm(S,T) : \text{$a$ is }\frac{1}{2}\text{-saturated by $f$}\}$;
        \item \textbf{Reduce Optimal with Moving Cut on $A \setminus \delta^\pm(S,T)$:} $\OPT_w \leq \frac{1}{2} \cdot \OPT$ where $\OPT_w$ is the maximum value of an $h$-length $S$-$T$ flow in $D$ with capacities $\U$ and lengths $\l' := \{\l_a\}_{a \in \delta^\pm(S,T)} \cup \{\l_a + \frac{1}{\eps} \cdot h \cdot \w_a\}_{a \not \in \delta^\pm(S,T)}$.
    \end{enumerate}
\end{lemma}
\begin{proof}
The proof idea is as follows. We case on whether most of the cut mass of $\w$ lies in $\delta^\pm(S,T)$ or not. If it does then by \Cref{lem:satCuts} arcs in $\delta^\pm(S,T)$ have a lot of flow over them. On the other hand, if most of the cut mass of $\w$ does not lie in $\delta^\pm(S,T)$ then increasing lengths of arcs in $A \setminus \delta^\pm(S,T)$ according to $\w$ greatly reduces the maximum $h$-length flow; in particular, any flow which is $h$-length after these increases must be incident to a lot of mass of $\w$ on arcs in $\delta^\pm(S,T)$ but by assumption this mass is bounded and so we can bound the total size of said flow. More formally, we case on whether or not $\sum_{a \in \delta^\pm(S,T)} \U_a \cdot \w_a \geq \frac{(1-\eps)^2}{2} \cdot \sum_{a} \U_a \cdot \w_a$. 

Suppose that $\sum_{a \in \delta^\pm(S,T)} \U_a \cdot \w_a \geq \frac{(1-\eps)^2}{2} \cdot \sum_{a} \U_a \cdot \w_a$. Letting $c_1 := \frac{1}{2}$ and applying \Cref{lem:satCuts} we know that
\begin{align}\label{eq:asfsf}
   \sum_{a \text{ $c_1$-saturated}} \U_a \cdot \w_a  \geq \left(1 - 2\eps\right ) \cdot \sum_a \U_a \cdot \w_a.
\end{align}
\noindent On the other hand applying \Cref{eq:asfsf} and our assumption that $\sum_{a \in \delta^\pm(S,T)} \U_a \cdot \w_a \geq \frac{(1-\eps)^2}{2} \cdot \sum_{a} \U_a \cdot \w_a$, it follows that
\begin{align*}
    \sum_{a \in A'} \U_a \cdot \w_a \geq \left( \frac{(1-\eps)^2}{2} - 2\eps\right) \cdot \sum_{a} \U_a \cdot \w_a  \geq \left(\frac{1}{2}-3\eps \right) \cdot \sum_{a} \U_a \cdot \w_a 
\end{align*}
\noindent as required.


Next, suppose that $\sum_{a \in \delta^\pm(S,T)} \U_a \cdot \w_a < \frac{(1-\eps)^2}{2} \cdot \sum_{a} \U_a \cdot \w_a$. Let $f'$ be an optimal $h$-length $S$-$T$ flow in $D$ with capacities $\U$ and lengths $\l'$. 
Observe that since $f'$ is $h$-length according to $\l'$ we know that every path $P \in \supp(f)$ satisfies $\sum_{a \in P \cap A \setminus \delta^\pm(S,T)} \w_a \leq \eps$. On the other hand, since $\w$ is a feasible moving cut we know that $\sum_{a \in P } \w_a \geq 1$ for any $P \in \supp(f)$. It follows that for each $P \in \supp(f')$ we know that
\begin{align}\label{eq:weqwq}
    \sum_{a \in P \cap \delta^\pm(S,T)} \w_a \geq 1 - \eps.
\end{align}

Thus, applying \Cref{eq:weqwq} and our assumption that $\sum_{a \in \delta^\pm(S,T)} \U_a \cdot \w_a < \frac{(1-\eps)^2}{2} \cdot \sum_{a} \U_a \cdot \w_a$ we have
\begin{align*}
    \st(f') &= \sum_{P \in \supp(f')} f_P\\
    &\leq \frac{1}{1-\eps} \sum_{P \in \supp(f')} f_P \sum_{a \in P \cap \delta^\pm(S,T)} \w_a\\
    &= \frac{1}{1-\eps} \cdot \sum_{a \in \delta^\pm(S,T)} \w_a \cdot f(a)\\
    & \leq \frac{1}{1-\eps} \cdot \sum_{a \in \delta^\pm(S,T)} \w_a \cdot \U_a\\
    & \leq \frac{1-\eps}{2} \cdot \sum_{a} \U_a \cdot \w_a.
\end{align*}

Lastly, by our assumption that $(f,w)$ is $(1 \pm \eps)$-approximate, we know that $(1-\eps) \cdot \sum_a \U_a \cdot \w_a \leq \st(f) \leq \OPT$ meaning $\sum_{a} \U_a \cdot \w_a \leq \frac{1}{1-\eps} \cdot \OPT$ and so
\begin{align*}
    \st(f') \leq \frac{1}{2} \cdot \OPT
\end{align*}
\noindent as required.
\end{proof}

\subsection{Decongesting our Flows on Arcs in $\delta^\pm(S,T)$}\label{sec:decongF}
We now introduce a helper procedure which will allow us to turn our computed fractional flows into integral flows that respect the capacities of arcs in $\delta^\pm(S,T)$ (while maybe violating the capacities of arcs not in $\delta^\pm(S,T)$).

\begin{lemma}\label{lem:decongestCutmatch}
    Suppose we are given a digraph $D = (V,A)$ with capacities $\U$ and lengths $\l$. Fix $\eta > 0$, $S, T \subseteq V$ and $h \geq 1$. Let $f = \eta \cdot \sum_j f_j$ where each $f_j$ is an integral $h$-length $S$-$T$ flow in $D$ with capacities $\U$ and lengths $\l$. Fix $A' \subseteq \delta^\pm(S,T)$ and a moving cut $\w$. Then one can compute an $h$-length $S$-$T$ flow $f'$ on $D$ with lengths $\l$ and capacities $\{\U_a\}_{a \in \delta^\pm(S,T)} \cup \{\eta \cdot \U_a\}_{a \not \in \delta^\pm(S,T)}$  such that $\sum_{a \in A'} f'(a) \cdot \w_a \geq \frac{\eta}{8} \cdot \sum_{a \in A'}f(a) \cdot \w_a$ in:
    \begin{enumerate}
        \item Deterministic parallel time $O(h)$ with $m$ processors;
        \item Deterministic CONGEST time $O(h)$.
    \end{enumerate}
\end{lemma}
\begin{proof}
    The basic idea is to scale up the input flow $f$ by $\eta$ and then remove flow paths that overcapacitate edges in $\delta^\pm(S,T)$. In particular, we first resolve over capacitated arcs in $\delta^+(S)$ by greedily choosing flows that send the most over each edge in $\delta^+(S)$ and then do the same for $\delta^-(T)$ but using the remaining flow.

     More formally, we begin by describing how to construct $f'$. We begin by dealing with over capacitated arcs in $\delta^+(S)$. For arc $a \in \delta^+(S)$ order the $f_j$ in descending order according to $f_j(a)$ as $f^{(1)}_a, f_a^{(2)}, \ldots$. Let $k_a$ be the largest integer such that $\sum_{i \leq k_a} f_a^{(i)}(a) \leq \U_a$ and if $f_j \in \{f_a^{(i)} : i \leq k_a\}$ then say that $a$ prefers flow $f_j$. Lastly, let $f_j^{\mid a}$ be $f_j$ restricted to paths going through $a$. That is, the value of $f^{\mid a}_j$ on path $P$ is
     \begin{align*}
         \left(f^{\mid a}_j \right)_P := \begin{cases}
             \left(f_j \right)_P & \text{if $P \cap a \neq \emptyset$}\\
             0 & \text{otherwise}
         \end{cases}
     \end{align*}
     Then, we resolve congestion on arcs in $\delta^+(S)$ by turning each $f_j$ into another flow $f_j'$. Specifically, let $f'_j$ be $f_j$ restricted only to paths where $f_j$ is preferred. That is, $f_j'$ is
     \begin{align*}
         f_j' := \sum_{a \in \delta^+(S) : a \text{ prefers } f_j} f^{\mid a}_j.
     \end{align*}

     We now resolve over capacitated edges in $\delta^-(T)$ in a symmetric way but using the $f_j'$ flows rather than the $f_j$ flows and taking the $\w$ values into account. Specifically, fix $a \in \delta^-(T)$ and let $x_{j,a}$ be the amount of flow sent over arc $a$ by $f_j$, scaled appropriately by $\w$. That is, we let
     \begin{align*}
         x_{a, j} := \sum_{P = (a', \ldots, a)} \left(f_j\right)_P \cdot \w_{a'}.
     \end{align*}
     Then, order the $f_j'$ in descending order according to $x_{a,j}$ as $f^{(1)}_a, f_a^{(2)}, \ldots$. Let $k_a$ be the largest integer such that $\sum_{i \leq k} f_a^{(i)}(a) \leq \U_a$ and if $f_j' \in \{f_a^{(i)} : i \leq k\}$ then say that $a$ prefers flow $f_j'$. Lastly, let $f_j^{\mid a}$ be $f_j'$ restricted to paths going through $a$. That is, the value of $f^{\mid a}_j$ on path $P$ is
     \begin{align*}
         \left(f^{\mid a}_j \right)_P := \begin{cases}
             \left(f_j' \right)_P & \text{if $P \cap a \neq \emptyset$}\\
             0 & \text{otherwise}
         \end{cases}
     \end{align*}
    Then, we let $f_S'$ be the flow which is all of our $f_j'$ flows appropriately restricted and preferred by arcs in $\delta^-(T)$
     \begin{align*}
         f'_S := \sum_{a \in \delta^-(T)} \sum_{f_j' : \text{ preferred by } a} f^{\mid a}_j.
     \end{align*}
     We construct flow $f_T'$ symmetrically to $f_S'$ (switching the roles of $S$ and $T$) and let our final flow $f'$ be
     \begin{align*}
         f' := \frac{1}{2} \left(f_S' + f_T'\right).
     \end{align*}

     $f'$ is trivially $h$-length (according to $\l$) and from $S$ to $T$ since $f$ is such a flow. Similarly, $f_S'$ and $f_T'$ are each feasible for capacities $\{\U_a\}_{a \in \delta^\pm(S,T)} \cup \{\eta \cdot \U_a\}_{a \not \in \delta^\pm(S,T)}$: arcs not in $\delta^\pm(S,T)$ are not over capacitated since $\sum_j f_j$ is feasible for capacities $\eta \cdot \U$ and arcs in $\delta^\pm(S,T)$ are not over capacitated by construction of $f'$. It follows that $f'$ is also feasible for these capacities.

    We now argue that $\sum_{a \in A'} f'(a) \cdot \w_a \geq \frac{\eta}{8} \cdot \sum_{a \in A'}f(a) \cdot \w_a$. To do so, observe that it suffices to show that
    \begin{align*}
        \sum_{a \in A'} f_S'(a) \geq \frac{\eta}{4} \cdot \sum_{a \in A' \cap \delta^+(S)} f(a) \cdot \w_a &\text{ and symmetrically }& \sum_{a \in A'} f_T'(a) \geq \frac{\eta}{4} \cdot \sum_{a \in A' \cap \delta^-(T)} f(a) \cdot \w_a
    \end{align*}
    We will argue $\sum_{a \in A'} f_S'(a) \geq \frac{\eta}{4} \cdot \sum_{a \in A' \cap \delta^+(S)} f(a) \cdot \w_a$ (the other inequality for $f_T'$ is symmetric).
    
    Let $\hat{f} := \sum_j f_j$ and let $\hat{f}' := \sum_{j} f_j'$. We begin by arguing that our $f_j'$ retain $\eta/2$ of the mass of $\hat{f}$. Observe that by definition we have
    \begin{align*}
        \eta \cdot \sum_{a \in A' \cap \delta^+(S)} \hat{f}(a) \cdot \w_a = \sum_{a \in A' \cap \delta^+(S)} f(a) \cdot \w_a
    \end{align*}
    Furthermore, observe by a standard bin-packing-type argument we know that
    \begin{align*}
        \sum_{a \in A' \cap \delta^+(S)} \hat{f}'(a) \cdot \w_a \geq \frac{\eta}{2} \sum_{a \in A' \cap \delta^+(S)} \hat{f}(a) \cdot \w_a
    \end{align*}
    and so it follows that
    \begin{align}
        \sum_{a \in A' \cap \delta^+(S)} \hat{f}'(a) \cdot \w_a \geq \frac{1}{2} \cdot \sum_{a \in A' \cap \delta^+(S)} f(a) \cdot \w_a\label{eq:asfsa}
    \end{align}
    Likewise, by another standard bin-packing-type argument we know that
    \begin{align*}
        \sum_{a \in A'} f_S'(a) \geq \frac{\eta}{2} \cdot \sum_{a \in A' \cap \delta^+(S)} \hat{f}'(a) \cdot \w_a
    \end{align*}
    which when combined with \Cref{eq:asfsa} shows $\sum_{a \in A'} f_S'(a) \geq \frac{\eta}{4} \cdot \sum_{a \in A' \cap \delta^+(S)} f(a) \cdot \w_a$ as required.

    We now argue the runtime. Observe that each $f_j'$ can be computed in a distributed manner by iterating over each $f_j$, having each arc $a \in \delta^+(S)$ decide if it prefers this flow and then forwarding the value of $f^{\mid a}_j$ along its flow paths. Since all $f_j$ are $h$-length by assumption this takes at most $h$ rounds of forwarding. This forwarding takes $O(h)$ rounds of CONGEST or parallel time. Computing $f_S'$ and $f_T'$ from the $f_j'$ is symmetric. Lastly, computing $f'$ is trivial to do from $f_S'$ and $f_T'$.
\end{proof}

\subsection{Our Length-Constrained Cutmatch Algorithm}
We conclude this section by giving our length-constrained cutmatch algorithm and its guarantees.
\cutMatchTheorem*
\begin{proof}
We iteratively build up the flow $\hat{f}$ and cut $\hat{\w}$ for our cutmatch as follows. The basic idea is to use \Cref{thm:main} to compute moving cuts and length-constrained flows and then to either add to $\hat{\w}$ using our moving cut or to add our flow to $\hat{f}$ depending on which case of \Cref{lem:cutmatchTwoCases} we are in. Getting this to run in CONGEST will require using sparse neighborhood covers since we cannot efficiently check which of the two cases of \Cref{lem:cutmatchTwoCases} we are in in CONGEST.

More formally, we initialize as follows.
\begin{itemize}
    \item We fix $\eps = .01$ for the course of our algorithm.
    \item We initialize $\hat{f}$ to assign $0$ to every path and $\hat{w}$ to assign $0$ to every arc. We will update these values over the course of our algorithm.
     \item We initialize the capacities we work with $\U'$ to scale all arcs in $A \setminus \delta^\pm(S,T)$ by $\gamma'$ for $\gamma' \leq \gamma$ to be described later. That is, initially $\U' := \{\U_a \}_{a \in \delta^\pm(S,T)}\cup \{ \gamma' \cdot \U_a \}_{a \not \in \delta^\pm(S,T)}$. We will update $\U'$ over the course of our algorithm.
     \item Given $\hat{\w}$ we will always let $\l' := \{\l_a\}_{a \in \delta^\pm(S,T)} \cup \{\l_a +  h \cdot \hat{\w}_a\}_{a \not \in \delta^\pm(S,T)}$ be the lengths that we work with. Note that initially $\l = \l'$.
\end{itemize}

 Next, our algorithm runs in phases each of which consists of iterations. At the beginning of each phase we compute a moving cut $\w$ using \Cref{thm:main} with the above value of $\eps$ so that $\eta = \tilde{\Theta}(1)$. In each iteration of our phase we apply \Cref{thm:main} to compute an $h$-length flow $f$ using lengths $\l'$ and capacities $\U'$ with $\eps$ as above. We next check if $(f, \w)$ is a $(1 \pm 2\eps)$-approximate $h$-length $S$-$T$ flow, moving cut pair in $D$ with capacities $\U'$ and lengths $\l'$ (note that the $\w$ we are using here is the one from the beginning of the phase, not the one we compute when we also compute $f$ using \Cref{thm:main}). If it is not then we move onto the next phase. If it is then by \Cref{lem:cutmatchTwoCases} we know one of two things must be true. Specifically, letting $\OPT$ be the maximum value of an $h$-length $S$-$T$ flow in $D$ with capacities $\U'$ and lengths $\l'$ at the beginning of our iteration then we have either:
 \begin{enumerate}
     \item \textbf{Flow Mostly Covers Cut on $\delta^\pm(S,T)$:} $\sum_{a \in A'} \U_a' \cdot \w_a \geq \left(\frac{1}{2}-6\eps \right) \cdot \sum_{a} \U_a' \cdot \w_a$ where $A' = \{a \in \delta^\pm(S,T) : \text{$a$ is }\frac{1}{2}\text{-saturated by $f$}\}$; or \label{eq:firstCase}
    \item \textbf{Reduce Optimal with Moving Cut on $A \setminus \delta^\pm(S,T)$:} $\OPT_w \leq \frac{1}{2} \cdot \OPT$ where $\OPT_w$ is the maximum value of an $h$-length $S$-$T$ flow in $D$ with capacities $\U'$ and lengths $\{\l_a\}_{a \in \delta^\pm(S,T)} \cup \{\l_a' + \frac{1}{2\eps} \cdot h \cdot \w_a\}_{a \not \in \delta^\pm(S,T)}$.\label{eq:secondCase}
 \end{enumerate}
 If the former case (\ref{eq:firstCase}) is true then we add to our flow $\hat{f}$ and in the latter case (\ref{eq:secondCase}) we add to $\hat{\w}$ using $\w$ and move onto the next phase. More formally, in the former case (\ref{eq:firstCase}) we apply \Cref{lem:decongestCutmatch} to $f$ to get integral flow $f'$ using $A' = \{a \in \delta^\pm(S,T) : \text{$a$ is }\frac{1}{2}\text{-saturated by $f$}\}$ as above. We then update $\hat{f}$ to $\hat{f} + f'$. Likewise, for each arc $a \in \delta^\pm(S,T)$ we update $\U'_a$ to $\U'_a - f'(a)$ and delete $a$ if we now have $\U'_a = 0$; for each arc $a \not \in \delta^\pm(S,T)$ we let $\U'_a$ be unchanged (observe that it follows for each arc $a \not \in \delta^\pm(S,T)$ we always have $\U_a' = \U_a$). In the former case we stay in this phase. In the latter case (\ref{eq:secondCase}), we let 
 \begin{align*}
     \w'_a = \begin{cases}\frac{1}{2\eps} \cdot \w_a& \text{if $a \not \in \delta^\pm(S,T)$}\\
     0 & \text{otherwise}
     \end{cases}
 \end{align*}
 and then update $\hat{\w}$ to be $\hat{\w} + \w'$. In the latter case we then move onto the next phase. We repeat this until the optimal $h$-length $S$-$T$ flow has value $0$.

 We now analyze this process. First, we claim that the number of iterations in each phase is at most $\tilde{O}(1)$. To do so, it suffices to show that in (\ref{eq:firstCase}) we reduce the cost of $\w$ by a multiplicative $1-1/\tilde{O}(1)$; this is because the cost of $\w$ is polynomially-bounded and so this can happen at most $\tilde{O}(1)$-many times. Towards this, fix an iteration and let $\U'$ and $\U''$ be our working capacities before and after updating for (\ref{eq:firstCase}) in this iteration. We claim that $\w$ has its cost reduced by a multiplicative $\left(1- 1/\tilde{O}(1) \right)$, namely
 \begin{align}
     \sum_a \U_a'' \cdot \w_a \leq \left(1- 1/\tilde{O}(1) \right) \cdot \sum_a \U_a' \cdot \w_a.
 \end{align}
 To see why this holds, observe that combining the guarantees of \Cref{lem:decongestCutmatch} and the fact that $\eta = \tilde{\Theta}(1)$ by \Cref{thm:main} we know that
 \begin{align*}
     \sum_{a \in A'} f(a) \cdot \w_a \leq \tilde{O}(1) \cdot \sum_{a \in A'} f'(a) \cdot \w_a.
 \end{align*}
But by definition of $A'$ and the fact that we are in (\ref{eq:firstCase}) for every $a$ we know that 
 \begin{align*}
     \left(\frac{1}{2}-6\eps \right) \cdot \sum_{a} \U_a' \cdot \w_a \leq \sum_{a \in A'} \U_a' \cdot \w_a \leq 2 \sum_{a \in A'}  f(a) \cdot \w_a.
 \end{align*}
 Combining these inequalities and our choice of $\eps$, we get that
 \begin{align*}
     \frac{1}{\tilde{O}(1)} \cdot \sum_{a} \U_a' \cdot \w_a  \leq \sum_{a \in A'}  f'(a) \cdot \w_a.
 \end{align*}
 Thus we get,
 \begin{align*}
     \sum_a \U_a'' \cdot \w_a &= \sum_a\left(\U_a' - f'(a) \right) \cdot \w_a \\
     &\leq \sum_{a \not \in A'}\U_a' \cdot \w_a + \sum_{a  \in A'}\left(\U_a' - f'(a) \right) \cdot \w_a\\
     & = \sum_a \U'_a \cdot \w_a - \sum_{a \in A'} f'(a) \cdot \w_a\\
     & \leq \left(1- 1/\tilde{O}(1) \right) \cdot \sum_a \U'_a \cdot \w_a
 \end{align*}
as desired. 

We now claim that the number of phases is at most $\tilde{O}(1)$. Specifically, recall that a phase ends when either $(f, \w)$ is not a $(1 \pm 2\eps)$-approximate $h$-length $S$-$T$ flow, moving cut pair or we are in (\ref{eq:secondCase}). We claim that in either case we multiplicatively reduce the value of the optimal $h$-length $S$-$T$ flow by at least a fixed constant strictly larger than $0$.
\begin{enumerate}
    \item Suppose $(f, \w)$ is not a $(1 \pm 2\eps)$-approximate $h$-length $S$-$T$ flow. Let $\U'$ and $\l'$ be the capacities and lengths when $\w$ is computed and let $\U''$ be the capacities according to which $f$ is computed. Let $\OPT_0$ be the the value of the optimal $h$-length $S$-$T$ flow when $\w$ is computed and let $\OPT_0'$ be the optimal such length-constrained flow when $f$ is computed. Since $f$ is part of a $(1\pm \eps)$-approximate pair we know that
    \begin{align*}
        (1-\eps) \cdot \OPT_0' \leq \st(f)
    \end{align*}
    On the other hand, since $(f, \w)$ is not a $(1 \pm 2\eps)$-approximate pair we know that
    \begin{align*}
        \st(f) < (1-2\eps) \cdot \sum_{a} \U''_a \cdot \w_a
    \end{align*}
    and so combining these inequalities we get
    \begin{align}
        \OPT_0' \leq \frac{1-2\eps}{1-\eps} \cdot \sum_{a} \U_a'' \cdot \w_a. \label{eq:hre}
    \end{align}

    Furthermore, recall that when $\w$ was computed it was part of a $(1 \pm \eps)$-approximate pair with capacities $\U'$ and $\l'$, meaning that 
    \begin{align}\label{eq:ojre}
        \sum_{a} \U'_a \cdot \w_a \leq \frac{1}{1-\eps} \cdot \OPT_0
    \end{align}
    Combining Equations \ref{eq:hre} and \ref{eq:ojre} along with the fact that $\U''_a \leq \U'_a$ for every $a$ since we only decrease capacities over the course of a phase, we get
    \begin{align*}
        \OPT_0' \leq \frac{1-2\eps}{(1-\eps)^2} \cdot \OPT_0. 
    \end{align*}
    By our choice of $\eps$ we have $\frac{1-2\eps}{(1-\eps)^2} \in (0, 1)$ as required.
    \item Suppose that a phase ends because we are in (\ref{eq:secondCase}). Recall that, as described above, in (\ref{eq:secondCase}) we know that we reduce the optimal flow by a multiplicative $\frac{1}{2}$; that is, as described above, $\OPT_w \leq \frac{1}{2} \cdot \OPT_0$.
\end{enumerate}

Since the number of phases is at most $\tilde{O}(1)$ and the number of iterations in each phase is at most $\tilde{O}(1)$, we have that the maximum number of iterations across all phases is at most $\tilde{O}(1)$. Letting $x$ be this upper bound on the number of possible iteration, we now set $\gamma'$ be $\Theta \left(x \cdot \frac{1}{\phi} \right)$ for a sufficiently large hidden constant.

The fact that $\hat{f}$ is feasible for capacities $\{\U_a\}_{a \in \delta^{\pm}(S,T)} \cup \{\gamma \cdot \U_a\}_{a \not \in \delta^{\pm}(S,T)}$ where $\gamma = \tilde{O}(\gamma') = \tilde{O}(\frac{1}{\phi})$ is then immediate: arcs in $\delta^\pm(S,T)$ are not over capacitated by construction of $\hat{f}$ and arcs not in $\delta^\pm(S,T)$ are not over capacitated since each $f'$ is feasible for capacities $\{\U_a \}_{a \in \delta^\pm(S,T)}\cup \{ \gamma' \cdot \U_a \}_{a \not \in \delta^\pm(S,T)}$ and we add such a flow to $\hat{f}$ at most the total-number-of-iterations-times, which, as argued above, is $\tilde{O}(1)$. Likewise, $\hat{f}$ is trivially $h$-length according to $\l$ since we only increase the lengths of arcs. Next, observe that our bound on the support size of $\hat{f}$ of $|\supp(\hat{f})| \leq \tilde{O}(h^{10} \cdot |A|)$ for our parallel algorithm is immediate from the support size bound of \Cref{thm:main} and the fact that we have $\tilde{O}(1)$-many iterations.

We now argue that $\hat{\w}$ satisfies the required properties. Specifically, we argue that $\hat{\w}$ is indeed a moving cut of value $ \sum_a \hat{\w}_a \cdot \U_a \leq \phi \left( \U^+(S) - \st(\hat{f})\right)$ such that $d_{\l'}(S,T) > h$ where
    \begin{align*}
            \l'_a := 
            \begin{cases}
                \l_a + h \cdot \hat{\w}_a & \text{if $f(a) < \U_a$}\\
                0 & \text{otherwise}.
        \end{cases}
    \end{align*}
Observe that since our algorithm runs until the optimal $h$-length $S$-$T$ flow has value $0$, we know that when our algorithm completes we indeed have $d_{\l'}(S,T) > h$. It remains to argue that 
\begin{align}\label{eq:tne}
\sum_a \hat{\w}_a \cdot \U_a \leq \phi \left( \U^+(S) - \st(\hat{f})\right).    
\end{align}
Towards showing \Cref{eq:tne}, we define the following quantities for the $i$th iteration. Let $\hat{f}_i$ be $\hat{f}$ in the $i$th iteration and let $\w^{(i)}$ be $\w'$ in the $i$th iteration (recall $\w'$ is equal to $\frac{1}{2\eps} \cdot \w$ on arcs not in $\delta^\pm(S,T)$ and $0$ otherwise).
Let $\OPT_i$ be the maximum $h$-length $S$-$T$ in the $i$th iteration and let $\U^{(i)}$ be our capacities in this iteration.

Observe that we trivially have
\begin{align}\label{eq:ter}
    \OPT_i \leq \sum_{a \in \delta^+(S)} \U_a^{(i)} = \sum_{a \in \delta^+(S)} \U_a - \hat{f}_i(a)
\end{align}
since the maximum value of a flow is at most the total capacity of arcs leaving $S$. We also trivially have $\hat{f}(a) \leq \sum_i\hat{f}_i(a)$ for any arc $a$ and, in particular, it follows that
\begin{align}\label{eq:ij}
    \frac{1-2\eps}{2\eps \cdot \gamma'} \cdot \sum_i \sum_{a \in \delta^+(S)} \hat{f}_i(a) \geq \phi \cdot \sum_{a \in \delta^+(S)} \hat{f}(a)
\end{align}
Similarly, since in each iteration we reduce the optimal value by a constant and since we do not scale arcs in $\delta^\pm(S,T)$ we have 
\begin{align}\label{eq:ijsfga}
    \frac{1-2\eps}{2\eps \cdot \gamma'} \cdot \sum_i \sum_{a \in \delta^+(S)} \U_a^{(i)} = \frac{1-2\eps}{\gamma'} \cdot \sum_i \sum_{a \in \delta^+(S)} \U_a \leq \phi \cdot \sum_{a \in \delta^+(S)} \U_a.
\end{align}
Applying Equations \ref{eq:ter}, \ref{eq:ij} and \ref{eq:ijsfga}, the fact that the support of $\hat{\w}$ is contained in $A \setminus \delta^\pm(S,T)$, $\w$ is part of a $(1\pm 2 \eps)$-approximate pair when we apply it and how we initially scaled the capacity of arcs in $A \setminus \delta^\pm(S,T)$ by $\gamma'$ we get that
\begin{align*}
    \sum_{a} \hat{\w}_a \cdot \U_a &= \frac{1}{\gamma'} \cdot \sum_i \sum_{a \not \in \delta^\pm(S,T)} \w_a^{(i)} \cdot \U_a^{(i)} \\
    & \leq \frac{1-2\eps}{2\eps \cdot \gamma'} \cdot \sum_i \OPT_i\\
    & \leq \frac{1-2\eps}{2\eps \cdot \gamma'} \cdot \sum_i \sum_{a \in \delta^+(S)} \U_a^{(i)} - \hat{f}_i(a)\\
    & \leq \phi \cdot \sum_{a \in \delta^+(S)} \U_a - \hat{f}(a),
\end{align*}
showing \Cref{eq:tne} as required.

It remains to discuss how to implement our algorithm in the stated times. Implementing the above algorithm in parallel in the stated running time is trivial by \Cref{lem:decongestCutmatch} and \Cref{thm:main} and the fact that we have $\tilde{O}(1)$ total iterations.

Our CONGEST algorithm requires more work and, in particular, the use of sparse neighborhood covers. Specifically, we first apply \Cref{lem:sparseNeigh} to compute a sparse $h$-neighborhood cover $\mcV_1, \mcV_2, \ldots, \mcV_{s}$. For $j \in [s]$, we let $S_j$ be a set of nodes whose $h$-neighborhood is contained in a single part of $\mcV_j$ (where we put each $v \in S$ in exactly one $S_j$). Then, we run the above algorithm to compute an $h$-length $S_j$-$T$ cutmatch $(\hat{f}_j, \hat{\w}_j)$ in $D$ with $S \setminus S_j$ deleted for $j = 1, \ldots, s$ where for each cutmatch we keep the length increases of the previous cutmatch. That is, we compute $(\hat{f}_j, \hat{\w}_j)$ with initial capacities $\U^{(j)}$ and lengths $\l^{(j)}$ where $\U^{(1)} = \U$ and $\l^{(1)} = \l$ and for $j > 1$ and each $a$ we have 
\begin{align*}
    \U^{(j)}_a = \begin{cases} \U_a & \text{if $a \not \in \delta^\pm(S,T)$}\\
    \U^{(j-1)}_a - \hat{f}_j(a) & \text{otherwise}
    \end{cases}
\end{align*}
and
\begin{align*}
    \l_a^{(j)} = \l_a + h \cdot \sum_{j' < j} (\hat{\w}_{j'})_a.
\end{align*}

$(\hat{f}, \hat{\w})$ is an $h$-length $S$-$T$ cutmatch by construction. Furthermore, it is easy to verify that the above process increases the congestion of this cutmatch and the support size of $\hat{f}$ by at most $s = \tilde{O}(1)$. To bound the cost of $\hat{\w}$ observe that applying our guarantees on each $\hat{\w}_j$ from above and the fact that the $S_j$ partition $S$ we get that the cost of $\hat{\w}$ is
\begin{align*}
    \sum_j \sum_a (\hat{\w}_j)_a \cdot \U_a &\leq \sum_j \phi \left(\left(\U^{(j)}\right)^+(S_j) - \st(\hat{f}_j)\right)\\
    &= \sum_j \phi \left(\U^+(S_j) - \st(\hat{f}_j)\right)\\
    &= \phi \cdot \left( U^+(S) - \st(\hat{f}) \right)
\end{align*}
as required.

We now discuss how to compute one of the cutmatches in the above process. Updating $\hat{\w}_j$ and $\hat{f}_j$ in each iteration is trivial. The non-trivial steps to implement our algorithm are:
\begin{itemize}
    \item Computing moving cut $\w$ at the beginning of the phase;
    \item Computing flow $f$ in each iteration of the phase and decongested flow $f'$ from $f$;
    \item Checking if $(f, \w)$ is a $(1 \pm 2\eps)$-approximate pair; and
    \item Checking which of (\ref{eq:firstCase}) and (\ref{eq:secondCase}) holds.
\end{itemize}
 Computing each moving cut $\w$ and each length-constrained flow $f$ can be done by \Cref{thm:main} in $\tilde{O}(h^{17})$ randomized CONGEST time. Computing $f'$ from $f$ can be done in deterministic $O(h)$ CONGEST time by \Cref{lem:decongestCutmatch}. Checking if $(f, \w)$ is a $(1 \pm 2\eps)$-approximate pair and checking which of which of (\ref{eq:firstCase}) and (\ref{eq:secondCase}) we are in can be done in time $\tilde{O}(h)$ by \Cref{lem:sparseNeigh}. Thus, each of our $\tilde{O}(1)$ iterations has its running time dominated by our call to \Cref{thm:main}, giving the stated running times.
\end{proof}

\section{Conclusion and Future Work}

In this work we gave the first efficient randomized and deterministic algorithms for computing $(1-\epsilon)$-approximate length-constrained flows both in parallel and in the CONGEST model of distributed computation. We used these algorithms to give new results in maximal and maximum disjoint path problems, expander decompositions, bipartite $b$-matching and length-constrained cutmatches. We conclude with several open questions and directions for future work.

\begin{enumerate}
    \item Our length-constrained flow algorithms have a dependence of $\poly(h)$ which when plugged into the techniques of \citet{haeuplerExpander2022} give CONGEST algorithms for many distributed problems, e.g.\ MST, whose running time is $\poly(\OPT)$ (up to sub-linear factors) where $\OPT$ is the optimal CONGEST running time for the input problem. It would be exciting to improve the dependence on $h$ of our algorithms to, say, $O(h)$ as this result when combined with those of \citet{haeuplerExpander2022} would give CONGEST algorithms running in time $O(\OPT)$ (up to sub-linear factors).
    \item The running time of many of our algorithms depends on $\rho_{CC}$, the best quality of a CONGEST algorithm for cycle cover (as defined in \Cref{dfn:rhoCC}). It is known that $\rho_{CC} \leq 2^{O(\sqrt{\log n})}$ but it would be extremely interesting to show that $\rho_{CC} \leq \tilde{O}(1)$. Such an improvement would immediately improve the dependency on $n$ from $n^{o(1)}$ to $\tilde{O}(1)$ for our CONGEST algorithms for deterministic length-constrained flows, deterministic maximal and maximum disjoint paths, $(1-\epsilon)$-approximate $b$-matching and length-constrained cutmatches. Such a result does not seem to be known even for the randomized case.
    \item Lastly, many classic problems can be efficiently solved by reducing to flows but, in particular, by reducing to length-constrained flows with a length-constraint $h = O(1)$. Indeed this is how we were able to give new algorithms for $b$-matching in this work. It would be interesting to understand which additional classic problems our length-constrained flow algorithms give new algorithms for in CONGEST.
\end{enumerate}

\appendix

\section{Generalizing Our Results to Multi-Commodity}\label{sec:multiComm}
In this section we generalize our main result for computing length-constrained flows and moving cuts to the setting where have many source sink pairs and we are trying to maximize the total flow between corresponding pairs subject to congestion constraints. We also generalize our cutmatch algorithms to the multi-commodity setting.

\subsection{Multi-Commodity Flows, Cutmatches and Results}
We now more formally define a multi-commodity length-constrained flow and moving cut. Suppose we are given a digraph $D = (V,A)$ with arc capacities $\U$, lengths $\l$ and $\kappa$ source set, sink set pairs $\{(S_i, T_i)\}_i$. Then, we have the following LP with a variable $f_P^{\{i\}}$ for each $i$ and path $P \in \mcP_h(S_i,T_i)$. We let $f^{\{i\}}$ gives the entire flow for commodity $i$.
\begin{align*}\label{LP:hopConFlowMulti}
    \max \sum_i \sum_{P \in \mcP_h(S_i, T_i)} f_P^{\{i\}}  \text{ s.t.} \tag{Multi Length-Constrained Flow LP} \\
    \sum_i \sum_{P \ni a}  f_P^{\{i\}} \leq \U_a & \qquad \forall a \in A\\
    0 \leq f_P^{\{i\}} & \qquad \forall i\in [\kappa], P \in \calP_h(S_i,T_i)
\end{align*}

For a multi-commodity length-constrained flow $f$, we will use the shorthand $f(a) = \sum_i \sum_{P \ni a} f_P^{\{i\}}$. Likewise we let $\st(f) = \sum_{i} \st(f^{\{i\}})$ be the total flow we send. An $h$-length multi-commodity flow, then, is simply a feasible solution to this LP.

\begin{definition}[Multi-Commodity $h$-Length Flow]
 Given digraph $D = (V,A)$ with lengths $\l$, capacities $\U$ and source, sink pairs $\{(S_i, T_i)\}_i$, an $h$-length $\{(S_i, T_i)\}_i$ flow is any feasible solution to \ref{LP:hopConFlowMulti}.
\end{definition}

With the above definition of multi-commodity length-constrained flows we can now define moving cuts as the dual of length-constrained flows. In particular, taking the dual of the above LP we get the multi-commodity moving cut LP with a variable $\w_a$ for each $a \in A$ and a variable $y_i$ for every $i \in [\kappa]$.
\begin{align*}\label{LP:movingCutMulti}
    \min \sum_{a \in A} \U_a \cdot \w_a \quad \text{s.t.} \tag{Multi Moving Cut LP}  \\
    \sum_{a \in P} \w_a \geq 1 & \qquad \forall i \in [\kappa], P \in \calP_h(S_i,T_i)\\
    0 \leq \w_a & \qquad \forall a \in A, i \in [\kappa]
\end{align*}
A multi-commodity $h$-length moving cut is simply a feasible solution to this LP.

\begin{definition}[Multi-Commodity $h$-Length Moving Cut]
 Given digraph $D = (V,A)$ with lengths $\l$, capacities $\U$ and source, sink pairs $\{(S_i, T_i)\}_i$, a multi-commodity $h$-length moving cut is any feasible solution to \ref{LP:movingCutMulti}.
\end{definition}

We will use $f$ and $\w$ to stand for solutions to \ref{LP:hopConFlowMulti} and \ref{LP:movingCutMulti} respectively. We say that $(f, \w)$ is a feasible pair if both $f$ and $\w$ are feasible for their respective LPs and that $(f, \w)$ is $(1 \pm \epsilon)$-approximate for $\epsilon > 0$ if the moving cut certifies the value of the length-constrained flow up to a $(1 - \epsilon)$; i.e.\ if $(1-\epsilon)\sum_{a} \U_a \cdot \w_a \leq \min_i \st(f^{\{i\}})$. 

When we are working in CONGEST we will say that $f$ is computed if each vertex $v$ stores the value $f_a^{(h',i)} := \sum_{P \in \mcP_{h,h'}(s,a,t)} f_P^{\{i\}}$. Here, we let $ \mcP_{h,h'}(s,a,t)$ be all paths in $\mcP_h(S,T)$ of the form $P' =(a_1, a_2, \ldots a, b_1, b_2, \ldots)$ where the path $(a, b_1, b_2, \ldots)$ has length exactly $h'$ according to $l$. We say multi-commodity moving cut $\w$ is computed in CONGEST if each vertex $v$ knows the value of $w_a$ for every arc incident to $v$. Likewise, we imagine that each node in the first round knows the capacities and lengths of its incident edges.

With the above notions, we can now state our main result for multi-commodity length-constrained flows and moving cuts which say that one can compute a feasible pair $(f, \w)$ in parallel and distributedly. In the following we say that length-constrained flow $f$ is integral if $f_P^{\{i\}}$ is an integer for every path in $\mcP_h(S_i,T_i)$ for every $i$.

More generally than $\kappa$ commodities, we solve the problem provided our commodoties can be grouped into $\kappa$ batches that are far apart.
\begin{definition}[$\kappa$-Batchable]
 Given digraph $D$ with lengths $\l$ and source, sink set pairs $\{S_i, T_i\}_i$ we say that a $\{S_i, T_i\}_i$ is $\kappa$-batchable if the pairs of $\{S_i, T_i\}_i$ can be partitioned into batches $\{\mcS_j, \mcT_j\}_j$ if 
 \begin{enumerate}
     \item \textbf{Covering:} For each $i$ there some $j$ such that $S_i \in \mcS_j$ and $T_i \in \mcT_j$;
     \item \textbf{Well-Separated:} For each $i$ and $i'$, if $v \in S_i \cup T_i$ and $v' \in S_{i'} \cup T_{i'}$ and $S_i, S_{i'} \in \mcS_j$ for some $j$ then $d_{\l}(v,v') > 2h$.
 \end{enumerate}
\end{definition}
Observe that if the number of commodities is $\kappa$ then the set of source, sink pairs is trivially $\kappa$-batchable.

The following summarizes our main result for computing multi-commodity length-constrained flows and moving cuts.
\begin{restatable}{thm}{mainThmMulti}\label{thm:mainMulti}
Given a digraph $D=(V,A)$ with capacities $\U$, lengths $\l$, length constraint $h \geq 1$, $0 < \eps < 1$ and source and sink vertices $S, T \subseteq V$, and $\kappa$-batchable source, sink pairs $\{S_i, T_i\}_i$, one can compute a feasible multi-commodity $h$-length flow, moving cut pair $(f, \w)$ that is $(1 \pm \epsilon)$-approximate in:
\begin{enumerate}
    \item Deterministic parallel time $\tilde{O}(\kappa \cdot \frac{1}{\eps^9} \cdot h^{17})$ with $m$ processors where $|\supp(f)| \leq \tilde{O}(\kappa \cdot \frac{h^9}{\eps^3} \cdot |A|)$;
    \item Randomized CONGEST time $\tilde{O}(\kappa \cdot \frac{1}{\eps^{9}} \cdot h^{17})$ with high probability;
    \item Deterministic CONGEST time $\tilde{O}\left(\kappa \cdot \frac{1}{\eps^9} \cdot h^{17}  + \kappa \cdot \frac{1}{\eps^7} \cdot h^{16} \cdot (\rho_{CC})^{10} \right)$.
\end{enumerate}
Furthermore, $f = \eta \cdot \sum_{j=1}^k f_j$ where $\eta = \tilde{\Theta}(\epsilon^2)$, $k = \tilde{O}\left(\kappa \cdot \frac{h}{\epsilon^4} \right)$ and $f_j$ is an integral $h$-length $S_i$-$T_i$ flow for some $i$.
\end{restatable}

Using the above algorithm, we can compute a multi-commodity version of cutmatches. As before, for a vertex subset $W \subseteq V$ we let $\U^+(W) = \sum_{v \in W}\sum_{a \in \delta^+(v)}\U_a$ and $\U^-(W) = \sum_{v \in W} \sum_{a \in \delta^-(v)}\U_a$. We also let $\delta^{\pm}(S,T) := \bigcup_{v \in S} \delta^+(v) \cup \bigcup_{v \in T} \delta^-(T)$. The following formalizes the object we compute. Symmetrically to the single-commodity setting, below we assume that each path in the support of each $f_i$ contains exactly one vertex of $S_i$ and one vertex of $T_i$.

\begin{definition}[Multi-Commodity $h$-Length Cutmatch]\label{def:multcutmatch}
Given digraph $D = (V,A)$ with capacities $\U$ and lengths $\l$, an $h$-length $\phi$-sparse cutmatch of congestion $\gamma$ between disjoint source, sink node set pairs $\{(S_i,T_i)\}_i$ where $S_i, T_i \subseteq V$ for each $i$ consists of:
\begin{itemize}
    \item An integral $h$-length flow $f = \sum_i f_i$ in $D$ with capacities $\gamma \cdot \U$ and lengths $\l$ where each $f_i$ is an $S_i$-$T_i$ flow satisfying $f_i(a) \leq \U_a$ for each $a \in \delta^\pm(S_i,T_i)$; 
    \item A moving cut $\w$ in $D$ of value $\sum_a \w_a \cdot \U_a \leq \phi \cdot \left(\sum_i \U^+(S_i) - \st(f_i) \right)$ such that for every $i$ we have $d_{\l_i}(S_i, T_i) > h$ where
    \begin{align*}
            (\l_i)_a := 
            \begin{cases}
                h+1 & \text{if $a \in \delta^\pm(S_i,T_i)$ and $f_i(a) = \U_a$}\\
                \l_a + h \cdot \w_a & \text{otherwise}
            \end{cases}
    \end{align*}
\end{itemize}
\end{definition}

Using the above algorithm for multi-commodity $h$-length flows, we can efficiently compute multi-commodity $h$-length cutmatches.
\begin{restatable}{thm}{multiCutmatch}\label{thm:multiCutmatch}
Suppose we are given a digraph $D = (V,A)$ with capacities $\U$ and lengths $\l$. There is an algorithm that, given $\kappa$-batchable source sink pairs $\{(S_i,T_i)\}_i$ where $S_i, T_i \subseteq V$ for every $i$ and two integer parameters $h \geq 1$ and $\phi \leq 1$, outputs a multi-commodity $h$-length $\phi$-sparse cutmatch $(\hat{f}, \hat{\w})$ of congestion $\gamma$ between $S$ and $T$, where $\gamma=\tilde{O}(\frac{1}{\phi})$. This algorithm runs in:
\begin{enumerate}
    \item Deterministic parallel time $\tilde{O}(\kappa \cdot h^{17})$ with $m$ processors with $|\supp(\hat{f})| \leq \tilde{O}(\kappa \cdot h^{10} \cdot |A|)$;
    \item Randomized CONGEST time $\tilde{O}(\kappa \cdot h^{17})$ with high probability;
\end{enumerate}
\end{restatable}

\subsection{Computing Multi-Commodity Length-Constrained Flows and Moving Cuts}
We proceed to use our $(1+\epsilon)$-lightest path blockers and multiplicative weights to compute multi-commodity length-constrained flows and moving cuts. Our strategy is more or less that of \Cref{sec:MW} but now we iterate through our batches of commodities; our analysis is mostly unchanged but we include it here for completeness.

Formally, our algorithm is given in \Cref{alg:mwMulti}. Throughout our analysis we will refer to the innermost loop of \Cref{alg:mwMulti} as one ``iteration.''
\begin{algorithm}
    \caption{Multi-Commodity Length-Constrained Flows  and Moving Cuts}
    \label{alg:mwMulti}
    \begin{algorithmic}[0] 
            \State \textbf{Input:} digraph $D = (V,A)$ with lengths $\l$, capacities $\U$, length constraint $h$ and $\kappa$-batchable source, sink pairs $\{S_i, T_i\}_i$ where $S_i, T_i \subseteq V$ for every $i$ and an $\eps \in (0,1)$.
            \State \textbf{Output:} $(1 \pm \epsilon)$-approximate $h$-length multi-commodity flow $f$ and moving cut $\w$.
            \State Let $\epsilon_0 = \frac{\epsilon}{6}$, let $\zeta = \frac{1+2 \eps_0}{\eps_0} + 1$ and let $\eta = \frac{\eps_0}{(1 + \eps_0) \cdot \zeta} \cdot \frac{1}{\log m}$.
            \State Initialize $\w_a \gets \left(\frac{1}{m}\right)^{\zeta}$ for all $a \in A$.
            \State Initialize $\lambda \gets  \left(\frac{1}{m}\right)^{\zeta}$.
            \State Initialize $f_P^{\{i\}} \gets 0$ for all $i$ and $P \in \mcP_h(S_i,T_i)$.
            \While{$\lambda < 1$}:

                \For{$j \in [\kappa]$ and each batch $(\mcS_j, \mcS_j)$}
                    \For{each $(S_i, T_i)$ with $S_i \in \mcS_j$ and $T_i \in \mcT_j$ in parallel}
                        \For{$\Theta\left(\frac{h \log_{1+\epsilon_0} n}{ \epsilon_0} \right)$ repetitions}
                                \State Compute an $h$-length $(1+\epsilon_0)$-lightest path blocker $\hat{f}$ (using \Cref{thm:pathBlockerAlg} with $\lambda$).
                                \State \textbf{Length-Constrained Flow (Primal) Update:} $f^{\{i\}} \gets f^{\{i\}} + \eta \cdot \hat{f}$.
                                \State \textbf{Moving Cut (Dual) Update:} $\w_a \gets (1+\epsilon_0)^{\hat{f}(a)/ \U_a} \cdot \w_a$ for every $a \in A$.
                        \EndFor
                    \EndFor
                \EndFor
                \State $\lambda \gets (1+\epsilon_0) \cdot \lambda$
            \EndWhile
            \State \Return $(f,\w)$.
    \end{algorithmic}
\end{algorithm}

We begin by observing that $\lambda$ always lower bounds $d^{(h)}_{\w}(S_i,T_i)$ for every $i$.
\begin{lemma}\label{lem:lambLowerBoundMulti}
It always holds that $\lambda \leq d^{(h)}_{\w}(S_x,T_x)$ for every $x$ in \Cref{alg:mwMulti}.
\end{lemma}
\begin{proof}
Fix an $x$ and a value of $\lambda$ and let $S = S_x$ and $T = T_x$. Our proof is by induction. The statement trivially holds at the beginning of our algorithm.

Let $\lambda_i$ be the value of $\lambda$ at the beginning of the $i$th iteration. We argue that if $d^{(h)}_{\w}(S,T) = \lambda_i$ then after $\Theta\left(\frac{h \log_{1+\epsilon_0} n}{ \epsilon_0} \right)$ additional iterations we must have $d^{(h)}_{\w}(S,T) \geq (1+\eps_0) \cdot \lambda_i$. Let $\lambda_i' = (1 + \epsilon_0) \cdot \lambda$ be $\lambda$ after these iterations. Let $\hat{f}_j$ be our lightest path blocker in the $j$th iteration for $(S_x, T_x)$.

Assume for the sake of contradiction that $d_{\w}^{(h)}(S,T) < \lambda_i'$ after $i + \Theta\left(\frac{h \log_{1+\epsilon_0} n}{ \epsilon_0} \right)$ iterations. It follows that there is some path $P \in \mcP_h(S,T)$ with weight at most $\lambda_i'$ after $i + \Theta\left(\frac{h \log_{1+\epsilon_0} n}{ \epsilon_0} \right)$ many iterations. However, notice that by definition of an $h$-length $(1+\epsilon_0)$-lightest path blocker $\hat{f}_j$ (\Cref{dfn:alphaPathBlocker}), we know that for every $j \in \left[i, i + \Theta\left(\frac{h \log_{1+\epsilon_0} n}{ \epsilon_0}\right)\right]$ there is some $a \in P$ for which $\hat{f}_j(a) = \U_a$. By averaging, it follows that there is some single arc $a \in P$ for which $\hat{f}_j(a) = \U_a$ for at least $\Theta\left(\frac{\log_{1+\epsilon_0} n}{\epsilon_0}\right)$ of these $j \in [i, i + \Theta\left(\frac{h \log_{1+\epsilon_0} n}{ \epsilon_0}\right)]$. Since every such arc starts with dual value $(\frac{1}{m})^\zeta$ and multiplicatively increases by a $(1+\epsilon_0)$ factor in each of these updates, such an arc after $i + \Theta\left(\frac{h \log_{1+\epsilon_0} n}{ \epsilon_0} \right)$ many iterations must have $\w_a$ value at least $(\frac{1}{m})^\zeta \cdot (1+\epsilon_0)^{\Theta\left(\frac{\log_{1+\epsilon_0} n}{ \epsilon_0}\right)} \geq n^{2}$ for an appropriately large hidden constant in our $\Theta$. However, by assumption, the weight of $P$ is at most $\lambda_i'$ after $i + \Theta\left(\frac{h \log_{1+\epsilon_0} n}{ \epsilon_0} \right)$ iterations and this is at most $2$ since $\lambda_i < 1$ since otherwise our algorithm would have halted. But $2 < n^{2}$ and so we have arrived at a contradiction.

Repeatedly applying the fact that $\lambda_i' = (1+ \epsilon_0) \lambda_i$ gives that $\lambda$ is always a lower bound on $d^{(h)}_{\w}(S,T)$.
\end{proof}

We next prove the feasibility of our solution.
\begin{lemma}\label{lem:ana1Multi}
The pair $(f, \w)$ returned by \Cref{alg:mwMulti} are feasible for \ref{LP:hopConFlowMulti} and \ref{LP:movingCutMulti} respectively.
\end{lemma}
\begin{proof}
First, observe that by \Cref{lem:lambLowerBoundMulti} we know that $\lambda$ is always a lower bound on $d_{\w}^{(h)}(S_i,T_i)$ for every $i$ and so since we only return once $\lambda > 1$, the $\w$ we return is always feasible.

To see that $f$ is feasible it suffices to argue that for each arc $a$, the number of times a path containing $a$ has its primal value increased is at most $\frac{\U_a}{\eta}$. Notice that each time we increase the primal value on a path containing arc $a$ by $\eta$ we increase the dual value of this edge by a multiplicative $(1+\epsilon_0)^{1/\U_a}$. Since the weight of our arcs according to $\w$ start at $(\frac{1}{m})^{\zeta}$, it follows that if we increase the primal value of $k$ paths incident to arc $a$ then $\w_a = (1+\epsilon_0)^{k / \U_a} \cdot (\frac{1}{m})^{\zeta}$. On the other hand, by assumption when we increase the dual value of an arc $a$ it must be the case that $\w_a < 1$ since otherwise $d_{\w}^{(h)}(S,T) \geq 1$, contradicting the fact that $\lambda$ always lower bounds $d_{\w}^{(h)}(S,T)$. It follows that $(1+\epsilon_0)^{k/\U_a} \cdot (\frac{1}{m})^{\zeta} \leq 1$ and so applying the fact that $\ln(1+ \epsilon_0) \geq \frac{\epsilon_0}{1+\epsilon_0}$ for $\epsilon_0 > -1$ and our definition of $\zeta$ and $\eta$ we get
\begin{align*}
    k & \leq \frac{\zeta \cdot (1+\eps_0)}{\eps_0} \cdot \U_a \log m  \\
    & = \frac{\U_a}{\eta}
\end{align*}
as desired.
\end{proof}

We next prove the near-optimality of our solution. 
\begin{lemma}\label{lem:ana2Multi}
The pair $(f, \w)$ returned by \Cref{alg:mwMulti} satisfies $(1-\epsilon)\sum_{a} \w_a \leq \sum_i \sum_{P \in \mcP_h(S_i, T_i)} f_P$.
\end{lemma}
\begin{proof}
Fix an iteration $i$ which is an iteration for the $j$th batch and let $\hat{f}$ be the sum of all lightest path blockers that we compute in parallel for each $(S_i, T_i) \in (\mcS_j, \mcT_j)$ in this iteration. Let $k_i$ be $\st(\hat{f})$, let $\lambda_i$ be $\lambda$ at the start of this iteration and let $D_i := \sum_{a} \w_a$ be our total dual value at the start of this iteration. Notice that $\frac{1}{\lambda_i} \cdot w$ is dual feasible and has cost $\frac{D_i}{\lambda_i}$ by \Cref{lem:lambLowerBoundMulti}. If $\beta$ is the optimal dual value then by optimality it follows that $\beta \leq \frac{D_i}{\lambda_i}$, giving us the upper bound on $\lambda_i$ of $\frac{D_i}{\beta}$. By how we update our dual, our bound on $\lambda_i$ and $(1+x)^{r} \leq 1 + xr$ for any $x \geq 0$ and $r \in (0,1)$ we have that
\begin{align*}
   D_{i+1} &= \sum_a (1+ \epsilon_0)^{\hat{f}(a)/\U_a} \cdot \w_a \cdot \U_a\\ &
   \leq \sum_a \left(1+ \frac{\epsilon_0 \hat{f}(a)}{\U_a} \right) \cdot \w_a \cdot \U_a \\
   &= D_i  + \epsilon_0 \sum_a  \hat{f}(a) \w_a\\
   & \leq D_i + \epsilon_0 (1 + 2 \eps_0) \cdot k_i \lambda_i\\
    &\leq D_i \left(1+\frac{(1+2\eps_0)\eps_0 \cdot k_i}{\beta} \right)\\
    & \leq D_i \cdot \exp\left(\frac{(1+2\eps_0)\eps_0 \cdot k_i}{\beta} \right).
\end{align*}
Let $T-1$ be the index of the last iteration of our algorithm; notice that $D_T$ is the value of $\w$ in our returned solution. Let $K := \sum_i k_i$. Then, repeatedly applying this recurrence gives us
\begin{align*}
    D_T &\leq D_0 \cdot \exp\left(\frac{(1+2\eps_0)\eps_0 \cdot K}{\beta} \right) \\
    &= \left(\frac{1}{m} \right)^{\zeta-1} \exp\left( \frac{(1+2\eps_0)\eps_0 \cdot K}{\beta} \right)
\end{align*}

On the other hand, we know that $\w$ is dual feasible when we return it, so it must be the case that $D_T \geq 1$. Combining this with the above upper bound on $D_T$ gives us $1 \leq \left(\frac{1}{m} \right)^{\zeta} \exp\left(\frac{(1+2\eps_0)\eps_0 \cdot K}{\beta} \right)$. Solving for $K$ and using our definition of $\zeta$ gives us
\begin{align*}
    \beta \log m \cdot \frac{\zeta-1}{(1 + 2 \eps_0) \cdot \eps_0}  &\leq K\\
    \beta \log m \cdot \frac{1}{\eps_0^2}  &\leq K.
\end{align*}
However, notice that $K\eta$ is the primal value of our solution so using our choice of $\eta$ and rewriting this inequality in terms of $K \eta$ by multiplying by $\eta = \frac{\eps_0}{(1 + \eps_0) \cdot \zeta} \cdot \frac{1}{\log m}$ and applying our definition of $\zeta = \frac{1+2 \eps_0}{\eps_0} + 1$ gives us
\begin{align}
    \frac{\beta}{\eps_0 \cdot (1+\eps_0) \cdot \zeta}&\leq K \eta \nonumber\\
    \frac{\beta}{(1+\eps_0)(1+3\eps_0)} &\leq K \eta.\label{eq:first}
\end{align}
Moreover, by our choice of $\eps_0 = \frac{\eps}{6}$ and the fact that $\frac{1}{1+x+x^2} \geq 1-x$ for $x \in (0,1)$ we get
\begin{align}
    1 - \eps &\leq \frac{1}{1 + \eps + \eps^2} \nonumber\\
    &\leq \frac{1}{(1+\frac{1}{2}\eps)^2} \nonumber\\
    &\leq \frac{1}{(1+3\eps_0)^2} \nonumber\\
    &\leq \frac{1}{(1+\eps_0)(1+3\eps_0)}. \label{eq:second}
\end{align}
Combining \Cref{eq:first} and \Cref{eq:second} we conclude that 
\begin{align*}
    (1-\eps) \cdot \beta \leq K \eta. 
\end{align*}\qedhere
\end{proof}

We conclude with our main theorem by proving that we need only iterate our algorithm $\tilde{O}\left(\kappa \cdot \frac{h}{\epsilon^4} \right)$ times.
\mainThmMulti*
\begin{proof}
By \Cref{lem:ana1Multi} and \Cref{lem:ana2Multi} we know that our solution is feasible and $(1\pm \epsilon)$-optimal so it only remains to argue the runtime of our algorithm and that the returned flow decomposes in the stated way.

We argue that we must only run for $O\left(\kappa \cdot \frac{h \log^2 n}{\epsilon^4} \right)$ total iterations. Since $\lambda$ increases by a multiplicative $(1+\epsilon_0)$ after every $\Theta\left(\kappa \cdot \frac{h \log n}{\epsilon_0^2} \right)$ iterations and starts at least $\left(\frac{1}{m}\right)^{\Theta(\frac{1}{\epsilon_0})}$, it follows by \Cref{lem:lambLowerBoundMulti} that after $y \cdot \Theta \left(\kappa \cdot \frac{h \log n}{\epsilon_0^2} \right)$ total iterations the $h$-length distance between every $S_i$ and $T_i$ is at least $(1 + \epsilon_0)^y \cdot \left(\frac{1}{m}\right)^{\Theta(1/\epsilon_0)}$. Thus, for $y \geq \Omega\left(\frac{\ln_{1+\epsilon_0}m}{\epsilon_0}\right) = \Omega\left(\frac{\ln n}{\epsilon_0^2} \right)$ we have that every $S_i$ and $T_i$ are at least $1$ apart in $h$-length distance. Consequently, our algorithm must run for at most $O\left(\kappa \cdot \frac{h \log^2 n}{\epsilon_0^4}\right) = O\left(\kappa \cdot \frac{h \log^2 n}{\epsilon^4} \right)$ many iterations.

Our running time is immediate from the the bound of $O\left(\kappa \cdot \frac{h \log^2 n}{\epsilon^4} \right)$ on the number of iterations of the while loop, the fact that commodities in the same batch can be updated in parallel and the running times given in \Cref{thm:pathBlockerAlg} for computing our $h$-length $(1+\epsilon_0)$-lightest path blocker.

Lastly, the flow decomposes in the stated way because we have at most $O\left(\kappa \cdot \frac{h \log^2 n}{\epsilon^4} \right)$ iterations and each $f_j$ is an integral $S$-$T$ flow by \Cref{thm:pathBlockerAlg}. Thus, our final solution is $\eta \cdot \sum_{j=1}^k f_j$ and $k = \tilde{O}\left(\frac{h}{\epsilon^4} \right)$. Likewise we have $|\supp(f)| \leq \tilde{O}(\kappa \cdot \frac{h^{10}}{\eps^7})$ for our parallel algorithm since we have $O\left(\kappa \cdot \frac{h \log^2 n}{\epsilon^4} \right)$ iterations and the fact that \Cref{thm:pathBlockerAlg} guarantees each $(1+\eps_0)$-lightest path blocker has support size at most $\tilde{O}(\frac{h^9}{\eps^3} \cdot |A|)$.
\end{proof}

\subsection{Computing Multi-Commodity Length-Constrained Cutmatches}
We proceed to compute multi-commodity length-constrained cutmatches.

\multiCutmatch*
\begin{proof}
The proof is entirely analogous to that of \Cref{lem:matchcut-old-purecut} except we use \Cref{thm:mainMulti} instead of \Cref{thm:main} to compute our flows and moving cuts. The only very minor difference is in order to ensure that $f_i(a) \leq \U_a$ for every $a \in \delta^\pm(S_i,T_i)$ we must slightly change the graph in which we work. Specifically, for each $S_i$ and each $s \in S_i$ where $\delta^+(s) = \{(s, t_1), (s, t_2), \ldots\}$ we add new vertices $s_1, s_2, \ldots$ and new arcs $(s_1, s), (s_2, s)$ each of length $0$ and respective capacities $\U_{(s,t_1)}, \U_{(s,t_2)}, \ldots$. We then replace $s$ in $S_1$ with $s_1, s_2, \ldots$. We do the same symmetric thing for each $T_i$. As a result of this we know that every set in $\{S_i\}_i \cup \{T_i\}_i$ is pairwise vertex disjoint. Also note that we can simulate any parallel or CONGEST algorithm on $D$ in this new graph with $O(1)$ time overhead. It is easy to verify that running the algorithm from \Cref{lem:matchcut-old-purecut} on the resulting graph using \Cref{thm:mainMulti} instead of \Cref{thm:main} then gives the result.
\end{proof}

\section{Deferred Proofs}\label{sec:deferredProofs}

\randCovPaths*
\begin{proof}
The basic idea is to have each path $P$ sample about $\U(P) / \tilde{\Delta}$ copies of itself.

More formally, we do the following. Consider the (multi) digraph $D'$ that is created by starting with $D$ and replacing each arc $a$ with $\U_a$ copies. For a given path $P$ in $D'$ from $S$ to $T$, we let $\Delta_P'$ be the number of distinct $S$ to $T$ paths in $D'$ which share an arc with $P$. Likewise, we let $\Delta' = \max_P \Delta_P'$ where this max is taken over all $S$ to $T$ paths in $D'$. We let $\mcP_{\approx \max}'$ be all paths $P$ for which $\Delta_P' \geq \Delta' / 2$. By how we defined the degree of paths in $D$, if a given path $P$ is in $\mcP_{\approx \max}'$ then so too is its corresponding path in $D$ in $\mcP_{\approx \max}$. Lastly, we let $N(P)$ be all paths from $S$ to $T$ in $D'$ which share an arc with $P$ other than $P$ itself and let $N^+(P) := N(P) \cup \{P\}$. 

In what follows we show how to sample a collection of arc-disjoint paths $\mcP_2$ in $D'$ where each $P \in \mcP_{\approx \max}'$ is such that with probability at least $\frac{1}{1024}$ the set $\mcP_2 \cap N^+(P)$ is non-empty. Before doing so, we observe that this suffices to show our claim. In particular, we can construct a flow $f$ by setting its value on arc $a$ to be $|\{P \in \mcP_2 : a \in P\}|$. Observe that by the arc-disjointness of $\mcP_2$ and how we constructed $D'$, $f$ is indeed a feasible $S$-$T$ flow. Moreover, we claim that for a given $\tilde{P} \in \mcP_{\approx \max}$ in $D$ we have $\prod_{a \in \tilde{P}}(\U_a - f_a) \leq \frac{1}{2} \U(\tilde{P})$ with probability $\Omega(1)$. In particular, let $X_P$ be the indicator of whether a given path $P$ in $D'$ from $S$ to $T$ is such that $N^+(P) \cap \mcP_2 = \emptyset$ so that $\E[X_P] \leq \frac{1023}{1024}$. Also, let $\tilde{\mcP}$ be all the paths in $D'$ that visit the same vertices as $\tilde{P}$ in $D$. Then we have 
\begin{align*}
    \prod_{a \in \tilde{P}}(\U_a - f_a) = \sum_{P \in \tilde{\mcP}} X_P.
\end{align*}
But, looking at the expectation of this, we have 
\begin{align*}
   \E\left[\sum_{P \in \tilde{\mcP}} X_P \right] &\leq \sum_{P \in \tilde{\mcP}} \frac{1023}{1024}\\
   & = \frac{1023}{1024} \cdot \U(\tilde{P})
\end{align*}
Thus, by Markov's inequality we have that $\sum_{P \in \tilde{\mcP}} X_P \geq \frac{2047}{2046} \cdot \E\left[\sum_{P \in \tilde{\mcP}} X_P \right]$ with probability at most $\frac{2046}{2047}$ and so with probability $\Omega(1)$ we get that $\sum_{P \in \tilde{\mcP}} X_P \leq \frac{2047}{2046} \cdot \E\left[\sum_{P \in \tilde{\mcP}} X_P \right]  \le \frac{2047}{2048} \cdot \U(\tilde{P})$.

Thus, it remains to show how to sample our collection of arc-disjoint paths $\mcP_2$ in $D'$ where each $P \in \mcP_{\approx \max}'$ is such that with probability at least $\frac{1}{1024}$ the set $\mcP_2 \cap N^+(P)$ is non-empty. We will sample $\mcP_2$ as follows. Imagine that $s$ initially receives $B\left(n_s^+, \frac{1}{64\tilde{\Delta}}\right)$-many balls where $B(n,p)$ is a binomial with $n$ trials each with probability of success $p$. We let $n_a$ and $n_v^+$ be as defined in \Cref{sec:pathCounts} for $D'$ where $\U_{a'} = 1$ for every arc $a'$ in $D'$.

When a vertex $v$ receives a ball, it tosses it to vertex $u \in N^+(v)$ with probability $n^+_{u}/n_v^+$. As $n_v^+ = \sum_{w \in N^+(v)}n_{w}^+$ this induces a valid probability distribution. Let $\mcP_1$ be the (multi) set of all paths traced out by balls. We will let $\mcP_2$ be all paths in $\mcP_1$ which are arc-disjoint (in $D'$) from all other paths in $\mcP_1$.

We first consider this process from the perspective of a single path $P$ from $S$ to $T$ in $D'$. Specifically, notice that the probability that a ball traces out a path $P = (s=v_1, v_2, \ldots, v_{h+1}=t)$ where $s \in S$ and $t \in T$ is uniform over paths. In particular, the probability that a given ball traces out path $P$ in $D'$ from $s$ to $t$ nicely telescopes as
\begin{align*}
    \frac{n^+_{v_2}}{n_{v_1}^+} \cdot \frac{n^+_{v_3}}{n_{v_2}^+} \cdot \ldots \cdot \frac{n^+_{v_{h+1}}}{n_{v_{h}}^+} &= \frac{n_{v_{h+1}}^+}{n_{v_1}^+}\\
    &= \frac{1}{n_s^+}.
\end{align*}

Thus, each ball that starts at $s$ traces out a uniformly random path incident to $s$ in $\mcP(S,T)$. Applying the parameters of our binomial distribution, it follows that the expected number of times a given path $P$ is included in $\mcP_1$ is $\frac{1}{64\cdot \tilde{\Delta}}$. Markov's inequality then shows that a given path has some copy in $\mcP_1$ with probability at most $\frac{1}{64 \cdot \tilde{\Delta}} \leq \frac{1}{32 \cdot \Delta}$. On the other hand, $P$ has exactly one copy included in $\mcP_1$ with probability $\frac{1}{64 \tilde{\Delta}} n^+_s \cdot \frac{1}{n^+_s}\left(1 - \frac{1}{n_s^+}  \right)^{n^+-1} \geq \frac{1}{128 \tilde{\Delta}}$. Thus, $P$ has at least one copy in $\mcP_1$ with probability at least $\frac{1}{128 \tilde{\Delta}} \geq \frac{1}{128\Delta}$.

We proceed to bound two simple probabilities regarding how paths are sampled. In particular, fix a path $P \in \mcP'_{\approx \max}$ in $D'$ from $S$ to $T$. Next, fix a $P' \in N^+(P)$. Then, let $\mcE_1(P')$ be the event that some copy of $P'$ is in $\mcP_1$ and no other path in $N^+(P)$ has a copy in $\mcP_1$. Likewise, let $\mcE_2(P')$ be the event that no path in $N(P')$ is in $\mcP_1$. Notice that if $\mcE_1(P')$ and $\mcE_2(P')$ hold then we have $P' \in \mcP_2$.

\begin{itemize}
    \item \textbf{Bounding $\Pr(\mcE_1(P'))$}. We will argue that $\Pr(\mcE_1(P')) \geq \frac{1}{256 \Delta}$. 
    
    Notice that since $N^+(P) \setminus \{P'\}$ consists of at most $\Delta$-many paths, the expected number of copies of paths in $N^+(P) \setminus \{P'\}$ in $\mcP_1$ is at most $\frac{1}{32}$. It follows by a Markov bound that with probability at least $\frac{1}{2}$ we have $N^+(P) \setminus \{P'\} \cap \mcP_1 = \emptyset$.
    
    Next, imagine that we condition on the event $N^+(P) \setminus \{P'\} \cap \mcP_1 = \emptyset$. Conditioning on this event can only increase the probability that a ball traces out $P'$. Since some copy of $P'$ is included in $\mcP_1$ with probability at least $\frac{1}{128\Delta}$ when we don't condition on this event, we conclude that 
    \begin{align*}
        \Pr(\mcE_1(P')) &= \Pr(N^+(P) \setminus \{P'\} \cap \mcP_1 = \emptyset) \cdot \Pr(P' \in \mcP_1 \mid N^+(P) \setminus \{P'\} \cap \mcP_1 = \emptyset)  \\
        & \geq \Pr(N^+(P) \setminus \{P'\} \cap \mcP_1 = \emptyset) \cdot \Pr(P' \in \mcP_1)  \\
        & \geq \frac{1}{256 \Delta}.
    \end{align*}
    
    

    
    \item \textbf{Bounding $\Pr(\mcE_2(P') \mid \mcE_1(P'))$}. We argue that $\Pr(\mcE_2(P') \mid \mcE_1(P')) \geq \frac{1}{2}$. 
    
    Notice that $\Pr(\mcE_2(P') \mid \mcE_1(P'))$ is minimized when $N^+(P)$ is of size exactly $\Delta+1$. However, in this case we have $\Pr(\mcE_2(P') \mid \mcE_1(P')) \geq \Pr(\mcE_2(P'))$. Thus, we conclude by a union bound that in general $\Pr(\mcE_2(P') \mid \mcE_1(P')) \geq \Pr(\mcE_2(P')) \geq 1-\Delta \cdot \frac{1}{32\Delta} \geq \frac{1}{2}$.
\end{itemize}

Putting these facts together and applying the fact that $P \in \mcP'_{\approx \max}$, we have that there is path in $N^+(P)$ included in $\mcP_2$ with probability at least
\begin{align*}
    \sum_{P' \in N^+(P)} \Pr(\mcE_1(P'))\cdot \Pr(\mcE_2(P') \mid \mcE_2(P')) &\geq \sum_{P' \in N^+(P)}\frac{1}{512\Delta}\\
    &\geq \frac{1}{1024}.
\end{align*}
as required.

It remains to argue that we can accomplish the above sampling of $\mcP_1$ and the construction of our flow $f$ in the stated times. Constructing $f$ from $\mcP_1$ is trivial to do in parallel and CONGEST so we focus on sampling $\mcP_1$. By \Cref{lem:computePathCounts} we can compute $n^+_v$ in the stated times. Passing balls to construct $\mcP_1$ and then $\mcP_2$ and constructing the above flow is trivial to do in the stated parallel time. For the CONGEST algorithm, we note that expected number of balls to cross any one arc in $D'$ when constructing $\mcP_1$ is at most $1$ and so a Chernoff and union bound shows that with high probability we never need to transmit more than $O(\log n)$ balls across an arc in $D'$ when constructing $\mcP_1$, with high probability. It follows that we never need to transmit more than $\tilde{O}(\U_{\max})$ balls across any one arc in $D$. Since it suffices to just transmit the number of balls, this can be done in $\tilde{O}(\log \U_{\max}) = \tilde{O}(1)$ rounds with high probability. Thus we can pass all balls from one layer to the next in $\tilde{O}(1)$ rounds of CONGEST with high probability. Lastly, constructing $\mcP_2$ from $\mcP_1$ is trivial to do in $O(h)$ rounds of CONGEST.
\end{proof}

\randMax*
\begin{proof}
Our algorithm simply repeatedly calls \Cref{lem:randCovPaths}. In particular we initialize our output flow $\hat{f}$ to be $0$ on all arcs and our working capacities on $D$ to be $\hat{\U} = \U$. Then for each $\tilde{\Delta} = (n \cdot \U_{\max})^{h}, (n \cdot \U_{\max})^{h}/2, (n \cdot \U_{\max})^{h}/4, \ldots$ we repeat the following $\Theta(h \cdot \log n \cdot \log \U_{\max})$ times. Let $f$ be the flow computed according to \Cref{lem:randCovPaths}. Update $\hat{\U}_a = \hat{\U}_a - f_a$ for every $a$ and update $\hat{f} = \hat{f} + f$. Clearly $\hat{f}$ is an integral $S$-$T$ flow. 

We need only verify that $\hat{f}$ is blocking. Since initially $\Delta \leq (n \cdot \U_{\max})^{h}$, to do so it suffices to argue that when we fix a value of $\tilde{\Delta}$ for which $\frac{\Delta}{2} \leq \tilde{\Delta} \leq \Delta$, then over the course of the $\Theta(h \cdot \log n \cdot \log \U_{\max})$ iterations where we use this value of $\tilde{\Delta}$ we have that $\Delta$ decreases by at least a factor of $2$ with high probability.

Consider $\Theta (h \cdot \log n \cdot  \log \U_{\max})$ contiguous iterations of the above with a $\tilde{\Delta}$ that satisfies $\frac{\Delta}{2} \leq \tilde{\Delta} \leq \Delta$ at the beginning of these iterations. Let $\mcP_0$ be $\mcP_{\approx \max}$ at the beginning of these iterations. To show that $\Delta$ decreases by at least a factor of $2$ over the course of these $\Theta (h \cdot \log n \cdot  \log \U_{\max})$ iterations it suffices to show that no path in $\mcP_0$ is in $\mcP_{\approx \max}$ for all of these iterations. Suppose for the sake of contradiction that some path $P \in \mcP_0$ is in $\mcP_{\approx \max}$ for all of these iterations. Then, applying the guarantees of \Cref{lem:randCovPaths}, we get that with high probability $\U(P)$ decreases by a $\frac{2047}{2048}$ factor at least $\Theta(h \cdot \log \U_{\max}))$ times. However, since $\U(P) \leq O((\U_{\max})^h)$, we get that after these iterations we would have reduced $\U(P)$ to $0$ with high probability by a union bound, i.e.\ $\Delta$ must have reduced by at least a factor of $2$.



The running time of our algorithm is immediate from the fact that we simply invoke \Cref{lem:randCovPaths} $\tilde{O}(h^2)$ times.\end{proof}
\newpage
\bibliography{abb,main}
\bibliographystyle{plainnat}

\end{document}

%% file: ellisPreamble.tex
\usepackage{algorithm}
\usepackage[noend]{algpseudocode}
\usepackage{amsmath,amssymb,amsthm}
\usepackage{bm}
\usepackage{bbm}
\usepackage{bbold}
\usepackage{booktabs}
\usepackage{caption}
\usepackage{enumerate}
\usepackage{float}
\usepackage{comment}
\usepackage{subcaption}
\usepackage[normalem]{ulem}
\usepackage[numbers,sort&compress]{natbib}

\newcounter{note}[section]

\newcommand{\mcC}{\mathcal{C}}

\newcommand{\mcE}{\mathcal{E}}

\newcommand{\mcH}{\mathcal{H}}
\newcommand{\mcI}{\mathcal{I}}

\newcommand{\mcP}{\mathcal{P}}

\newcommand{\mcS}{\mathcal{S}}
\newcommand{\mcT}{\mathcal{T}}

\newcommand{\mcV}{\mathcal{V}}

\allowdisplaybreaks